%% file: popl19-dgv.tex
\documentclass[acmsmall]{acmart}

\acmJournal{PACMPL}
\acmVolume{1}
\acmNumber{CONF} 
\acmArticle{1}
\acmYear{2018}
\acmMonth{1}
\acmDOI{} 
\startPage{1}

\setcopyright{none}

\usepackage{booktabs}   
\usepackage{subcaption} 

\usepackage{amsmath}
\usepackage{thm-restate}
\usepackage{paralist}
\usepackage{mathpartir}
\renewcommand {\MathparLineskip}{\lineskip=1.5ex} 
\usepackage{stmaryrd}
\usepackage{listings,color}
\usepackage[disable]{todonotes}
\usepackage{tikz}

\input{macros}
\input{rulemacros}

\begin{document}

\title[Label-Dependent Session Types]{Label-Dependent Session Types}


\author{Peter Thiemann}
\orcid{0000-0002-9000-1239}             
\affiliation{
  \institution{Faculty of Engineering, University of Freiburg}            
  \country{Germany}                    
}
\email{thiemann@acm.org}          

\author{Vasco T. Vasconcelos}
\affiliation{
  \institution{LASIGE, Department of Informatics, Faculty of Sciences, University of Lisbon}            
  \country{Portugal}                    
}
\email{vv@di.fc.ul.pt}          


\begin{abstract}
\input{abstract}
\end{abstract}



\maketitle

\input{introduction}

\input{motivation}

\input{binary-session-types}

\input{ld-session-calculus}


\input{natural-numbers}

\input{ldgv-meta-light}

\input{algorithmic-typing}


\input{embedding}

\input{implementation}

\input{related}

\input{conclusions}

\begin{acks}
  This work was supported by FCT through the LASIGE Research Unit,
  ref.\ UID/CEC/00408/2019, and by Cost Action CA15123 EUTypes.
\end{acks}

\bibliography{abbrv,papers,books,collections,misc,theses}

\iffinal\else
\clearpage 
\appendix
\input{appendix}
\fi

\end{document}

%% file: macros.tex
\newcommand{\grmeq}{\; ::= \;\;}
\newcommand{\grmor}{\;\mid\;}

\usepackage{MnSymbol}

\newif\ifdraft
\drafttrue 
\draftfalse

\newcommand{\vv}[1]{\ifdraft\textcolor{red}{\textit{[#1--VV]}}\fi}

\newif\iffinal
\finaltrue
\finalfalse 

\newcommand\finalreport[2]{\iffinal#1\else#2\fi}
\newcommand\finalconcrete{\finalreport{our technical report \cite{ThiemannVasconcelos2019-corr}}}

\newcommand\Tag[1]{\tag{\TirName{#1}}}

\newcommand\syntax[1]{\textup{\textbf{\textsf{#1}}}}

\newcommand\SENDK{\syntax{send}}

\newcommand\SELECTK{\syntax{select}}
\newcommand\CLOSEK{\syntax{close}}

\newcommand\RECVK{\syntax{recv}}
\newcommand\CASEK{\syntax{case}}
\newcommand\RCASEK{\syntax{rcase}}
\newcommand\WAITK{\syntax{wait}}
\newcommand\FORKK{\syntax{fork}}

\newcommand{\FIXK}{\syntax{fix}}

\newcommand\APP[1]{#1\,}

\newcommand\LAM[1]{\lambda#1.}

\newcommand\FORK{\FORKK\,}
\newcommand\NEW{\syntax{new}}
\newcommand\CASE[2]{\CASEK\,#1\,\syntax{of}\,\{#2\}}
\newcommand\RCASE[2]{\RCASEK\,#1\,\syntax{of}\,\{#2\}}
\newcommand\CASEBRANCH[3]{#1:#2.#3}

\newcommand\SELECT{\SELECTK\,}
\newcommand\SEND[1]{\SENDK\,#1}
\newcommand\RECV{\RECVK\,}
\newcommand\CLOSE{\CLOSEK\,}
\newcommand\WAIT{\WAITK\,}

\newcommand{\RECURSOR}[5]{\syntax{rec}\,{#1}\,{#2}\,{#3}.{#4}.{#5}}

\newcommand\END{\syntax{End}} 
\newcommand\ENDS{\END_!}
\newcommand\ENDR{\END_?}

\newcommand\LabelUniv{\ensuremath{\mathcal L}}
\newcommand\LabelSet{L}

\newcommand\EOFLabel{\textsc{eos}} 

\newcommand\TEnd{\END}
\newcommand\TUnit{\syntax{Unit}}

\newcommand\TInt{\syntax{Int}}

\newcommand\TNat{\syntax{Nat}}
\newcommand{\TRec}[5]{\syntax{rec}\:{#1}\,{#2}\,[{#3}]\,{#5}}

\newcommand\TEnv{\Gamma}
\newcommand\TOut{\Delta}
\newcommand\SEnv{\Delta}
\newcommand\EmptyEnv{{\cdot}}

\newcommand\GVKind[2]{\vdash #2:#1}

\newcommand\EnvForm[2]{\vdash #2 : #1}

\newcommand\DGVEnvForm[2]{\vdash_\DGV #2 : #1}
\newcommand{\AEnvForm}[2]{ #2 \Rightarrow #1}
\newcommand\Equation[3][{}]{#2 \mathbin{\syntax{=}} #3} 
\newcommand\Convertible[2]{#1 \equiv #2} 

\newcommand\ReducesTo{\boldsymbol{\longrightarrow}}
\newcommand\ReducesToSym\longleftrightarrow
\newcommand\TransTo\leadsto
\newcommand\Hole\Box
\newcommand\ECN[1][E]{\mathcal{#1}}
\newcommand\Context[2][E]{\ECN[#1]{[#2]}}
\newcommand\FV{\textit{fv}}

\newcommand\Setof[1]{\{#1\}}

\newcommand\Multl{l}
\newcommand\Multm{m}
\newcommand\Multn{n}

\newcommand\MultLT\prec
\newcommand\MultLE\preceq 
\newcommand\MultLub\curlyvee
\newcommand\MultGlb\curlywedge

\newcommand\MultLin\OccOne
\newcommand\MultUn\OccInfty

\newcommand\Subst\sigma

\newcommand\DemoteOp\downarrow
\newcommand\Demote{{\DemoteOp}\,}

\newcommand\Lin{\syntax{lin}}
\newcommand\Un{\syntax{un}}
\newcommand\Fin{\syntax{fin}}
\newcommand\UNR[1]{#1^\OccInfty} 

\newcommand\Occp{p}

\newcommand\OccInfty\Un 
\newcommand\OccOne\Lin 

\newcommand\EUnit{()}
\newcommand\EPair[3][\Multm]{\langle#2, #3\rangle}
\newcommand\DPair[4][\Multm]{\langle#2=#3, #4\rangle}
\newcommand\ELet[2]{\syntax{let}\,#1=#2\,\syntax{in}\,}

\newcommand\DecomposeSymbol{\circ} 
\newcommand{\DecomposeOp}[2]{#1 \DecomposeSymbol #2}
\newcommand\Decompose[3]{#1 = \DecomposeOp{#2}{#3}}

\newcommand\with\binampersand

\newcommand\Kind\Multm
\newcommand\Kindm\Multm
\newcommand\Kindn\Multn

\newcommand\KindLin\OccOne
\newcommand\KindUn{\OccInfty}
\newcommand\KindSession\OccOne
\newcommand\KindUnSession\OccInfty

\newcommand\NUC[2]{(\nu#1#2)}

\newcommand{\TVar}{\alpha}

\newcommand\Pol{p}
\newcommand\PolPos{\oplus}
\newcommand\PolNeg{\ominus}

\newcommand{\Zero}{Z}
\newcommand{\Succ}{S}

\newcommand\Subkind\MultLE
\newcommand\KindLub\sqcup

\newcommand\DUAL[1]{{#1}^\perp}
\newcommand\Dual[1]{\DUAL{(#1)}}
\newcommand\Proc[1]{\langle#1\rangle}

\newcommand\Chanc{c}
\newcommand\Chand{d}

\newcommand\PAR{\mathbin{\rule[-.75ex]{1pt}{2.5ex}}} 

\newcommand\Embed[1]{\llangle#1\rrangle}

\newcommand{\acronym}[1]{\textsf{#1}}

\newcommand\LDGV{\acronym{LDST}}
\newcommand\DGV{\acronym{LDST}}

\newcommand\LAST{\acronym{LAST}}
\newcommand\LSST{\acronym{LSST}}

\newcommand\meta[1]{\textup{\textit{#1}}}

\newcommand\Dom{\meta{dom}}

\newcommand\DemoteExtend[3]{#1 \lhd #2:#3}
\newcommand\PlainExtend[3]{#1, #2:#3} 

\newcommand\DYN{\mathord{\star}}
\newcommand\DC{\mathord{\text{\textcircled{$\star$}}}}

\newcommand\Locked\bullet
\newcommand\Subtype\le

\newcommand\SYNTH{\Rightarrow}
\newcommand\CHECK{\Leftarrow}
\newcommand\UNFOLD{\Downarrow}

\newcommand\Algdash\vdash
\newcommand\JAlgConv[4]{%
  #1 \Algdash #2 \Rightarrow #3 
}
\newcommand\JAlgSubtype[4]{%
  #1 \Algdash #2 \Subtype #3 : #4
}
\newcommand\JAlgSubtypeSynth[4]{%
  #1 \Algdash #2 \Subtype #3 \SYNTH #4
}
\newcommand\JAlgSubtypeCheck[4]{%
  #1 \Algdash #2 \Subtype #3 \CHECK #4
}
\newcommand\JAlgSubtypeSynthExt[5]{%
  \JAlgSubtypeSynth{#1}{#2}{#3}{#4} ; #5
}
\newcommand\JAlgKindSynth[3]{%
  #1 \Algdash #2 \SYNTH #3
}
\newcommand\JAlgKindCheck[3]{%
  #1 \Algdash #2 \CHECK #3
}
\newcommand\JAlgTypeSynth[4]{%
  #1 \Algdash #2 \SYNTH #3 ; #4
}
\newcommand\JAlgTypeCheck[4]{%
  #1 \Algdash #2 \CHECK #3 ; #4
}
\newcommand\JAlgTypeCheckExt[5]{%
  \JAlgTypeCheck{#1}{#2}{#3}{#4} ; #5
}
\newcommand\JAlgTypeUnfold[3]{%
  #1 \Algdash #2 \UNFOLD #3
}

\lstdefinelanguage{gst}{
  language=haskell,
  basicstyle=\sffamily\small,
  extendedchars=true,
  breaklines=true,
  morekeywords={new,send,receive,recv,skip,select,dualof,int,string,bool,unit,fork,close,wait,case,rcase,raise,Lab,Unit,rec,Nat,End,with},
  deletekeywords={sum,Eq},
  tabsize=8,
  literate=
    {end!}{{$\textbf{end}_{!}$}}4 
    {end?}{{$\textbf{end}_{?}$}}4 
    {(+)}{$\oplus$}1 
    {Sigma}{$\Sigma$}1
    {otimes}{$\otimes$}1 
    {forall}{$\forall$}1
    {alpha}{$\alpha$}1 
    {beta}{$\beta$}1 
    {DYN}{$\DYN$}1
    {DC}{$\DC$}1
    {=>}{$\Rightarrow$}1 
    {->}{$\rightarrow$}1 
    {lambda}{$\lambda$}1 
    {->}{$\rightarrow$}1 
    {-o}{$\multimap$}1
}



%% file: rulemacros.tex
\newcommand\RuleDualRec{%
  \inferrule[]{}{
    \Dual{\TRec VS\TVar m R} =
    \TRec V{\DUAL S}\alpha  m {\DUAL R[\TVar_\PolNeg/\TVar_\PolPos, \TVar_\PolPos/\TVar_\PolNeg]}}
}

\newcommand\RuleDualPosVar{
  \inferrule{}{
    \DUAL{\TVar_\PolPos} = \TVar_\PolNeg
  }
}

\newcommand\RuleDualNegVar{
  \inferrule{}{
    \DUAL{\TVar_\PolNeg} = \TVar_\PolPos
  }
}

\newcommand\RuleEqualityF{%
  \inferrule[Equality-F]{
    \TEnv \vdash V : A  \\
    \TEnv \vdash W : A \\
    \TEnv \vdash A : \Kind
  }{
    \TEnv \vdash \Equation VW : \Kind
  }
}

\newcommand\RuleUnitF{%
  \inferrule[Unit-F]{\EnvForm \MultUn \TEnv}{\TEnv \vdash \TUnit : \KindUnSession}
}

\newcommand\RuleEndF{%
  \inferrule[End-F]{\EnvForm \MultUn \TEnv}{\TEnv \vdash \TEnd : \KindUnSession}
}

\newcommand\RuleLabF{%
  \inferrule[Lab-F]{
    \EnvForm \MultUn \TEnv
  }{\TEnv \vdash L : \KindUn}
}

\newcommand\RuleLabET{%
  \inferrule[Lab-E']{
    \UNR\TEnv \vdash V : L \\
    (\forall\ell\in L)~\DemoteExtend\TEnv x {\Equation[L] V \ell} \vdash A_\ell : \Kind
  }{\TEnv \vdash \CASE V {\overline{\ell: A_\ell}^{\ell\in L}} : \Kind
  }
}

\newcommand\RulePiF{%
    \inferrule[Pi-F]{
      \DemoteExtend\TEnv x A \vdash B : \Kindn
  }{
    \TEnv \vdash \Pi_\Multm (x:A)B : \Kindm
  }
}

\newcommand\RuleSigmaF{%
  \inferrule[Sigma-F]{
    \TEnv \vdash A : \Kind \\
    \DemoteExtend\TEnv x A \vdash B : \Kind
  }{\TEnv \vdash \Sigma (x: A)B : \Kind}
}

\newcommand\RuleSsnOutF{%
  \inferrule[Ssn-Out-F]{
    \DemoteExtend\TEnv x A \vdash S : \KindSession
  }{
    \TEnv \vdash {!(x: A)}S : \KindSession
  }
}

\newcommand\RuleSsnInF{%
    \inferrule[Ssn-In-F]{
      \DemoteExtend\TEnv x A \vdash S : \KindSession
  }{
    \TEnv \vdash {?(x: A)}S : \KindSession
  }
}

\newcommand\RuleSubKind{%
    \inferrule[Sub-Kind]
  {\TEnv \vdash A : \Kindm \\ \Kindm \Subkind \Kindn}
  {\TEnv \vdash A : \Kindn}
}

\newcommand\RuleConvRefl{%
  \inferrule[Conv-Refl]{\TEnv \vdash A : \Kind}{\TEnv \vdash \Convertible A A : \Kind}
}

\newcommand\RuleConvSym{%
  \inferrule[Conv-Sym]{\TEnv \vdash \Convertible A B : \Kind}{\TEnv \vdash \Convertible B A : \Kind}
}

\newcommand\RuleConvTrans{%
  \inferrule[Conv-Trans]{\TEnv \vdash \Convertible A B : \Kind \\
    \TEnv \vdash \Convertible B C}{\TEnv \vdash \Convertible A C : \Kind}
}

\newcommand\RuleConvSubst{%
  \inferrule[Conv-Subst]{
    \TEnv \vdash y : \Equation VW\\
    \TEnv \vdash V : L \\
    (\forall\ell\in L)~ \DemoteExtend\TEnv z {\Equation V\ell} \vdash A_\ell : \Kind
  }{
    \TEnv \vdash
    \Convertible
    {\CASE{V}{\overline{\ell:A_\ell}^{\ell\in L}}}
    {\CASE{W}{\overline{\ell:A_\ell}^{\ell\in L}}}
    : \Kind
  }
}

\newcommand\RuleConvBeta{%
  \inferrule[Conv-Beta]
  {(\forall\ell\in L)~ \TEnv \vdash A_{\ell}: \Kind \\ \ell'\in L}
  {\TEnv \vdash \Convertible
    {\CASE {\ell'}{\overline{\ell : A_\ell}^{\ell\in L}}}
    {A_{\ell'}}
    : \Kind}
}
  
\newcommand\RuleConvEta{%
  \inferrule[Conv-Eta]
  {\TEnv \vdash A : \Kind \\
    \TEnv \vdash x : L
  }
  {\TEnv \vdash
    \Convertible {A} {\CASE x {\overline{\ell : A}^{\ell\in L}}} : \Kind
  }
}

\newcommand\RuleSubConv{%
  \inferrule[Sub-Conv]{\TEnv \vdash \Convertible A B : \Kind}{
    \TEnv \vdash A \Subtype B : \Kind}
}

\newcommand\RuleSubLab{%
  \inferrule[Sub-Lab]
  { \EnvForm \MultUn \TEnv  \\ L \subseteq L'
  }
  { \TEnv \vdash L \Subtype L' : \KindUn }
}

\newcommand\RuleSubTrans{%
  \inferrule[Sub-Trans]{
    \TEnv \vdash  A \Subtype B : \Kind \\
    \TEnv \vdash  B \Subtype C : \Kind}{
    \TEnv \vdash  A \Subtype C : \Kind}
}

\newcommand\RuleSubSub{%
  \inferrule[Sub-Sub]{
    \TEnv \vdash A \Subtype B : \Kind \\
    \Kind \Subkind \Kindn
  }{
    \TEnv \vdash
    A \Subtype B
    : \Kindn
  }
}

\newcommand\RuleSubPi{%
  \inferrule[Sub-Pi]
  { \TEnv \vdash A' \Subtype A : \Kind_A
    \\ \DemoteExtend\TEnv x {A'} \vdash B \Subtype B' : \Kind_B
    \\ \Multm \MultLE \Multn }
  {\TEnv \vdash \Pi_\Multm (x:A) B \Subtype \Pi_\Multn (x:A')B' :
    \Kindn
  }
}

\newcommand\RuleSubSigma{%
  \inferrule[Sub-Sigma]
  { \TEnv \vdash A \Subtype A' : \Kind
    \\ \DemoteExtend\TEnv x A \vdash B \Subtype B' : \Kind
  }
  {\TEnv \vdash \Sigma (x:A) B \Subtype \Sigma (x:A')B' : \Kind}
}

\newcommand\RuleSubSend{%
  \inferrule[Sub-Send]
  {
    \TEnv \vdash A' \Subtype A : \Kind \\
    \DemoteExtend\TEnv x {A'} \vdash S \Subtype S' : \KindSession
  }
  {\TEnv \vdash {!(x:A)}S \Subtype {!(x:A')}S' : \KindSession}
}

\newcommand\RuleSubRecv{%
  \inferrule[Sub-Recv]
  {
    \TEnv \vdash A \Subtype A' : \Kind \\
    \DemoteExtend\TEnv x A \vdash S \Subtype S' : \KindSession
  }
  {\TEnv \vdash {?(x:A)}S \Subtype {?(x:A')}S' : \KindSession}
}

\newcommand\RuleSubTVar{%
  \inferrule[Sub-TVar]
  {
    \TEnv \vdash \TVar_\Pol : \Multm
  }
  {\TEnv \vdash \TVar_\Pol \Subtype \TVar_\Pol : \Multm}
}

\newcommand\RuleSubCase{%
  \inferrule[Sub-Case]
  {    (\forall\ell\in L\setminus L')~    
    \DemoteExtend{\DemoteExtend{\TEnv\setminus x} x {(L\setminus L')}} y
    \Equation[L\cap L'] x \ell \vdash A_\ell : \Kind
    \\
    (\forall\ell\in L'\setminus L)~    
    \DemoteExtend{\DemoteExtend{\TEnv\setminus x} x {(L'\setminus L)}} y
    \Equation[L\cap L'] x \ell \vdash A'_\ell : \Kind
    \\
     \UNR\TEnv \vdash x : L\cap L'
    \\
    (\forall\ell\in L\cap L')~
    \DemoteExtend \TEnv y \Equation[L\cap L'] x \ell \vdash A_\ell \Subtype A'_\ell : \Kind
  }
  { \TEnv \vdash
    \CASE x {\overline{\ell: {A_\ell}}^{\ell\in L}}
    \Subtype 
    \CASE x {\overline{\ell: {A'_\ell}}^{\ell\in L'}} : \Kind
  }
}

\newcommand\RuleSubType{%
  \inferrule[Sub-Type]{
    \TEnv \vdash M : A \\
    \UNR\TEnv \vdash A \Subtype B :\Kind
  }{\TEnv \vdash M : B}
}

\newcommand\RuleName{%
  \inferrule[Name]{
    \EnvForm \MultUn{\DemoteExtend{\TEnv_1} z A,\TEnv_2}
  }{
    \TEnv_1, z:A, \TEnv_2 \vdash z: A
  }
}

\newcommand\RuleUnitI{%
  \inferrule[Unit-I]{\EnvForm \MultUn \TEnv}{\TEnv \vdash \EUnit :
    \TUnit}
}

\newcommand\RuleLabI{%
  \inferrule[Lab-I]{
    \EnvForm \MultUn \TEnv
  }{\TEnv \vdash \ell : \{\ell\}}
}

\newcommand\RuleLabE{%
  \inferrule[Lab-E]{
    \UNR\TEnv \vdash V : L \\ 
    (\forall\ell\in L)~\DemoteExtend\TEnv y {\Equation[L] V \ell} \vdash N_\ell : B
  }{\TEnv \vdash \CASE V {\overline{\ell: N_\ell}^{\ell\in L}} : B 
  }
}

\newcommand\RulePiI{%
  \inferrule[Pi-I]{
    \EnvForm \Multm \TEnv \\
    \PlainExtend\TEnv x A \vdash M : B 
  }{\TEnv \vdash \lambda_\Multm (x: A).M : \Pi_\Multm (x:  A)B}
}

\newcommand\RulePiE{%
  \inferrule[Pi-E]{
    \Decompose{\TEnv}{\TEnv_1}{\TEnv_2} \\
    \TEnv_1 \vdash M : \Pi_\Multm (x: A)B \\
    \TEnv_2 \vdash N : A
  }{
    \TEnv \vdash M\,N : B[N/x]
  }
}

\newcommand\RuleASigmaI{%
  \inferrule[A-Sigma-I]{
    \JAlgTypeSynth{\TEnv_1}{  V}{ A}{ \TEnv_2} \\
    \JAlgTypeSynth{\DemoteExtend{\PlainExtend{\TEnv_2} x A} z {\Equation[A]x V}}
    {N}{ B}{ \DemoteExtend{\DemoteExtend{\TEnv_3} x A} z { \Equation[A]x V}}
  }{
    \JAlgTypeSynth{\TEnv_1}{ \DPair {x} V N}{ \Sigma (x: A)B}{ \TEnv_3}}
}

\newcommand\RuleASigmaE{%
  \inferrule[A-Sigma-E]{
    \JAlgTypeSynth{\TEnv_1} M{D}{ \TEnv_2}\\
    \JAlgTypeUnfold{\UNR{\TEnv_2}}{D}{{\Sigma(x: A)}B} \\\\
    \JAlgTypeSynth{\TEnv_2, x: A, y: B} N
    C{ \DemoteExtend{\DemoteExtend{\TEnv_3} x A} y B}\\
    x,y \notin \FV (C)
  }{
    \JAlgTypeSynth{\TEnv_1}{ \ELet{\EPair[]xy}{M}{N}} C{    \TEnv_3}
  }
}

\newcommand\RuleASigmaG{%
  \inferrule[A-Sigma-G]{
    \JAlgTypeSynth{\TEnv_1} M{D}{ \TEnv_2}\\
    \JAlgTypeUnfold{\UNR{\TEnv_2}}{D}{{\Sigma(x: L)}B} \\
    \JAlgTypeSynth{\TEnv_2, x: L, y: B}
    {\CASE x {\overline{\ell:N}^{\ell \in L}}}
    C{ \DemoteExtend{\TEnv_3, x : L} y B}\\
    x,y \notin \FV (C)
  }{
    \JAlgTypeSynth{\TEnv_1}{ \ELet{\EPair[]xy}{M}{N}} C{    \TEnv_3}
  }
}

\newcommand\RuleAFork{%
  \inferrule[A-Fork]{
    \JAlgTypeCheck{\TEnv_1} M \TUnit{ \TEnv_2}
  }{
    \JAlgTypeSynth{\TEnv_1}{ \FORK M}{ \TUnit}{ \TEnv_2}
  }
}

\newcommand\RuleASsnI{%
  \inferrule[A-Ssn-I]{
    \JAlgKindCheck{\UNR\TEnv} S \KindSession
  }{
    \JAlgTypeSynth\TEnv \NEW{ \Sigma (x:S) \DUAL S} \TEnv
  }
}

\newcommand\RuleASsnSendE{%
  \inferrule[A-Ssn-Send-E]{
    \JAlgTypeSynth{\TEnv_1}{ M}{C}{ \TEnv_2}\\
    \JAlgTypeUnfold{\UNR{\TEnv_2}}{C}{{!(x: A)}S}
  }{
    \JAlgTypeSynth{\TEnv_1}{ \SEND M}{ \Pi_\KindLin (x:A) S}{ \TEnv_2}
  }
}

\newcommand\RuleASsnRecvE{%
  \inferrule[A-Ssn-Recv-E]{
    \JAlgTypeSynth{\TEnv_1}{ M}{  C}{ \TEnv_2} \\
    \JAlgTypeUnfold{\UNR{\TEnv_2}}{C}{{?(x: A)}S}
  }{
    \JAlgTypeSynth{\TEnv_1}{ \RECV M}{ \Sigma (x: A)S}{ \TEnv_2}
  }
}

\newcommand\RuleASubType{%
  \inferrule[A-Sub-Type]{
    \JAlgTypeSynth{\TEnv} M B{\TOut}
    \\
    \JAlgSubtypeSynth{\UNR\TEnv} B A \Kind
  }{
    \JAlgTypeCheck{\TEnv}{ M}{ A}{ \TOut}
  }
}

\newcommand\RuleAUnfoldType{%
  \inferrule[A-Unfold-Refl]{
    \text{$A$ not a case}
  }{
    \JAlgTypeUnfold{\TEnv} A A
  }
}

\newcommand\RuleAUnfoldCase{%
  \inferrule[A-Unfold-Case]{
    \JAlgConv{\UNR\TEnv} V {\ell'} L \\
    \JAlgTypeUnfold{\TEnv} {B_{\ell'}}A
  }{
    \JAlgTypeUnfold\TEnv {\CASE V {\overline{\ell: B_\ell}^{\ell\in L}}} A 
  }
}

\newcommand\RuleAUnfoldCaseA{%
  \inferrule[A-Unfold-Case1]{
    (\not\exists\ell)~\JAlgConv{\UNR\TEnv} x \ell {L'}  \\
    (\forall \ell \in L)~\JAlgTypeUnfold\TEnv {A_\ell}{L'}
  }{
    \JAlgTypeUnfold\TEnv {\CASE x {\overline{\ell: A_\ell}^{\ell\in L}}}
    {L'}
  }  
}

\newcommand\RuleAUnfoldCaseB{%
  \inferrule[A-Unfold-Case2]{
    (\not\exists\ell)~\JAlgConv{\UNR\TEnv} x \ell {L'}  \\
    (\forall \ell \in L)~\JAlgTypeUnfold\TEnv {A_\ell}{\Context[P]{B_\ell}} \\
    \ECN[P] \in \{\Pi_\Multm (y:A)\Hole, \Sigma (y:A)\Hole, {!(y:A)}\Hole, {?(y:A)}\Hole \}
  }{
    \JAlgTypeUnfold\TEnv {\CASE x {\overline{\ell: A_\ell}^{\ell\in L}}}
    {\Context[P]{\CASE x {\overline{\ell: {B_\ell}}^{\ell\in L}}} } 
  }  
}

\newcommand\RuleAUnfoldRecZ{%
  \inferrule[A-Unfold-Rec-Z]{
    \JAlgConv{\UNR\TEnv} V \Zero \TNat \\
    \JAlgTypeUnfold{\TEnv} {A}C
  }{
    \JAlgTypeUnfold\TEnv {\TRec V A\alpha\Kind B} C
  }
}

\newcommand\RuleAUnfoldRecS{%
  \inferrule[A-Unfold-Rec-S]{
    \JAlgConv{\UNR\TEnv} V {\Succ (W)} \TNat \\
    \JAlgTypeUnfold{\TEnv} {B[{\TRec WA\alpha\Kind B}/\alpha]}C
  }{
    \JAlgTypeUnfold\TEnv {\TRec VA\alpha\Kind B} C
  }
}

\newcommand\RuleAUnfoldRecX{%
  \inferrule[A-Unfold-Rec-X]{
    (\not\exists V)~\JAlgConv{\UNR\TEnv} x V \TNat  \\
    \ECN[P] \in \{ L, \TNat, \Pi_\Multm (y:A)\Hole, \Sigma (y:A)\Hole, {!(y:A)}\Hole, {?(y:A)}\Hole \} \\
    \JAlgTypeUnfold\TEnv {A}{\Context[P]{A'}} \\
    \JAlgTypeUnfold{\TEnv, \alpha:\Kind} {B}{\Context[P]{B'}}
  }{
    \JAlgTypeUnfold\TEnv
    {\TRec xA\alpha\Kind B}
    {\Context[P]    {\TRec x{A'}\alpha\Kind {B'}}    } 
  }  
}

\newcommand\RuleSigmaI{%
  \inferrule[Sigma-I]{
    \Decompose{\TEnv}{\TEnv_1}{\TEnv_2} \\
    \TEnv_1 \vdash  V : A \\
    \DemoteExtend{\PlainExtend{\TEnv_2} x A} z {\Equation[A]xV}  \vdash N : B
  }{
    \TEnv \vdash \DPair {x} V N : \Sigma (x: A)B
  }
}

\newcommand\RuleSigmaE{%
  \inferrule[Sigma-E]{
    \Decompose\TEnv{\TEnv_1}{\TEnv_2} \\
    \TEnv_1 \vdash M : \Sigma (x: A)B \\
    \PlainExtend{\PlainExtend{\TEnv_2} x A} y B \vdash N : C\\
    x,y \notin \FV (C)
  }{
    \TEnv \vdash \ELet{\EPair[]xy}{M}{N} : C 
  }
}

\newcommand\RuleSigmaG{%
  \inferrule[Sigma-G]{
   \Decompose\TEnv{\TEnv_1}{\TEnv_2} \\
    \TEnv_1 \vdash M : \Sigma (x: L)B \\
    x \in \FV (B) \\
    \PlainExtend{\PlainExtend{\TEnv_2} x L} y B \vdash \CASE x
    {\overline{\ell:N}^{\ell\in L}} : C\\
    x,y \notin \FV (C)
  }{
    \TEnv \vdash \ELet{\EPair[]xy}{M}{N} : C 
  }
}

\newcommand\RuleFork{%
  \inferrule[Fork]{
    \TEnv \vdash M : \TUnit
  }{
    \TEnv \vdash \FORK M : \TUnit
  }
}

\newcommand\RuleSsnI{%
  \inferrule[Ssn-I]{
    \UNR\TEnv \vdash S : \KindSession
  }{\TEnv \vdash \NEW : \Sigma (x:S) \DUAL S}
}

\newcommand\RuleSsnSendE{%
  \inferrule[Ssn-Send-E]{
    \TEnv \vdash M : {!(x: A)}S
  }{
    \TEnv \vdash \SEND M : \Pi_\KindLin (x:A) S
  }
}

\newcommand\RuleSsnRecvE{%
  \inferrule[Ssn-Recv-E]{
    \UNR\TEnv \vdash A : m
    \\
    \TEnv \vdash M :  {?(x: A)}S
  }{
    \TEnv \vdash \RECV M : \Sigma (x: A)S
  }
}

\newcommand\RuleAEqualityF{%
  \inferrule[A-Equality-F]{
    \JAlgTypeSynth\TEnv V L {\_}  \\
    \JAlgTypeCheck\TEnv \ell L {\_}
  }{
    \JAlgKindSynth\TEnv{ \Equation[L] V \ell} \KindUn
  }
}

\newcommand\RuleAUnitF{%
  \inferrule[A-Unit-F]{\AEnvForm \MultUn \TEnv}{
    \JAlgKindSynth\TEnv \TUnit \KindUnSession}
}

\newcommand\RuleAEndF{%
  \inferrule[A-End-F]{\AEnvForm \MultUn \TEnv}{
    \JAlgKindSynth\TEnv \TEnd \KindUnSession}
}

\newcommand\RuleALabF{%
  \inferrule[A-Lab-F]{
    \AEnvForm \MultUn \TEnv
  }{\JAlgKindSynth\TEnv L\KindUn}
}

\newcommand\RuleALabET{%
  \inferrule[A-Lab-E']{
    \JAlgTypeCheck\TEnv V L {\_} \\
    (\forall\ell\in L)~
    \JAlgKindSynth{\DemoteExtend\TEnv x {\Equation[L] V \ell}}{ A_\ell}{ \Kind_\ell}
  }{\JAlgKindSynth\TEnv{ \CASE V {\overline{\ell: A_\ell}^{\ell\in L}}}
    {\bigsqcup \Kind_\ell}
  }
}

\newcommand\RuleARecET{%
  \inferrule[A-Rec-E']{
    \JAlgTypeCheck\TEnv V \TNat {\_} \\
    \JAlgKindCheck{\TEnv, x : \Equation[\TNat] V \Zero}{ A}{\Kind} \\
    \JAlgKindCheck{\TEnv, \alpha:\Kind, x : \Equation[\TNat] V {\Succ (W)}}{ B}{ \Kind}
  }{\JAlgKindSynth\TEnv{
      \TRec V A \alpha \Kind B}
    {\Kind}
  }
}

\newcommand\RuleARecAlpha{%
  \inferrule[A-Rec-Var]{}{
    \JAlgKindSynth{\TEnv, \alpha:\Kind, \TEnv'} \alpha \Kind
  }
}

\newcommand\RuleAPiF{%
  \inferrule[A-Pi-F]{
    \JAlgKindSynth\TEnv A{ \Kind_A} \\
    \JAlgKindSynth{\DemoteExtend\TEnv x A}{ B}{ \Kind_B}
  }{
    \JAlgKindSynth\TEnv{ \Pi_\Multm (x:A)B}{ \Kindm}
  }
}

\newcommand\RuleASigmaF{%
  \inferrule[A-Sigma-F]{
    \JAlgKindSynth\TEnv A{ \Kind_A } \\
    \JAlgKindSynth{\DemoteExtend\TEnv x A}{ B}{
      {\Kind_B}}
  }{
    \JAlgKindSynth\TEnv{ \Sigma (x: A)B}{
      \Kind_A\KindLub\Kind_B}
  }
}

\newcommand\RuleASsnOutF{%
  \inferrule[A-Ssn-Out-F]{
    \JAlgKindSynth\TEnv A{ \Kind} \\
    \JAlgKindCheck{\DemoteExtend\TEnv x A}{ S}{ \KindSession}
  }{
    \JAlgKindSynth\TEnv{ {!(x: A)}S} \KindSession
  }
}

\newcommand\RuleASsnInF{%
  \inferrule[A-Ssn-In-F]{
    \JAlgKindSynth\TEnv A{ \Kind} \\
    \JAlgKindCheck{\DemoteExtend\TEnv x A} S \KindSession
  }{
    \JAlgKindSynth\TEnv{ {?(x: A)}S} \KindSession 
  }
}

\newcommand\RuleASubKind{%
  \inferrule[A-Sub-Kind]
  {\JAlgKindSynth\TEnv A \Kindm \\ \Kindm \Subkind \Kindn}
  {\JAlgKindCheck\TEnv A{  \Kindn}}
}

\newcommand\RuleASUnit{%
    \inferrule[AS-Unit]{}
    {\JAlgSubtypeSynth\TEnv \TUnit \TUnit \KindUnSession}
}

\newcommand\RuleASLabel{%
  \inferrule[AS-Label]{L \subseteq L'}
  {\JAlgSubtypeSynth\TEnv L{ L'} \KindUn}
}

\newcommand\RuleASPi{%
  \inferrule[AS-Pi]
  { \JAlgSubtypeSynth \TEnv {A'}{ A} {\Kind_A}
    \\ \JAlgSubtypeSynth{\DemoteExtend\TEnv x{A'}} B{ B'}{ \Kind_B}
    \\ \Multm \MultLE \Multn
  }
  {\JAlgSubtypeSynth{\TEnv}{ \Pi_\Multm (x:A) B}{ \Pi_\Multn (x:A')B'}{ \Kindn}}
}

\newcommand\RuleASSigma{%
  \inferrule[AS-Sigma]
  { \JAlgSubtypeSynth{ \TEnv} A{ A'}{ \Kind_A}
    \\ \JAlgSubtypeSynth{\DemoteExtend\TEnv xA}{ B}{ B'}{ \Kind_B}
  }
  {\JAlgSubtypeSynth{\TEnv}{ \Sigma (x:A) B}{ \Sigma (x:A')B'}{
      \Kind_A \sqcup \Kind_B}}
}

\newcommand\RuleASSend{%
  \inferrule[AS-Send]
  {
    \JAlgSubtypeSynth{\TEnv}{ A'}{ A}{ \Kind} \\\\
    \JAlgSubtypeCheck{\DemoteExtend\TEnv x{A'}}{ S}{ S'} \KindSession
  }
  {\JAlgSubtypeSynth{\TEnv}{ {!(x:A)}S}{ {!(x:A')}S'} \KindSession}
}

\newcommand\RuleASRecv{%
  \inferrule[AS-Recv]
  {
    \JAlgSubtypeSynth{\TEnv}{ A}{ A'}{ \Kind} \\\\
    \JAlgSubtypeCheck{\DemoteExtend\TEnv xA}{ S}{ S'} \KindSession
  }
  {\JAlgSubtypeSynth{\TEnv}{ {?(x:A)}S}{ {?(x:A')}S'} \KindSession}
}

\newcommand\RuleASCaseLeftA{%
  \inferrule[AS-Case-Left1]
  { \JAlgConv\TEnv V {\ell'} {L'} \\
    \ell' \in L \\
    \JAlgSubtypeSynth{\TEnv}{ A_{\ell'}}{ B} \Kind
  }
  { \JAlgSubtypeSynth{\TEnv}
    {\CASE V {\overline{\ell: {A_\ell}}^{\ell\in L}}}
    {B} \Kind
  }
}

\newcommand\RuleASCaseLeftB{%
  \inferrule[AS-Case-Left2]
  { \JAlgTypeUnfold\TEnv x L \\
    L \subseteq L' \\
    (\not\exists\ell')~\JAlgConv\TEnv x {\ell'}{ L'} \\
    (\forall\ell\in L)~
    \JAlgSubtypeSynth{\TEnv, y: \Equation[L] x \ell}{ A_\ell}{ B}{ \Kind_\ell}
  }
  { \JAlgSubtypeSynth{\TEnv}
    {\CASE x {\overline{\ell: {A_\ell}}^{\ell\in L'}}}
    {B}{ \bigsqcup_{\ell\in L}\Kind_\ell}
  }
}

\newcommand\RuleASTVarLeft{%
  \inferrule[AS-TVar-Left]
  {}
  { \JAlgSubtypeSynthExt{\TEnv}
    {\TVar}
    {A} \Kind {[\TVar \mapsto A]}
  }
}

\newcommand\RuleASRecLeftZ{%
  \inferrule[AS-Rec-Left-Z]
  { \JAlgConv\TEnv V \Zero \TNat \\
    \JAlgSubtypeSynthExt{\TEnv}{ A}{ C} \Kind \Subst
  }
  { \JAlgSubtypeSynthExt{\TEnv}
    {\TRec V A \alpha\Kind B}
    {C} \Kind \Subst
  }
}

\newcommand\RuleASRecLeftS{%
  \inferrule[AS-Rec-Left-S]
  { \JAlgConv\TEnv V {\Succ (W)} \TNat \\
    \JAlgSubtypeSynthExt{\TEnv}{ B[{\TRec W A \alpha\Kind B}/\alpha]}{
      C} \Kind \Subst
  }
  { \JAlgSubtypeSynthExt{\TEnv}
    {\TRec V A \alpha\Kind B}
    {C} \Kind \Subst
  }
}

\newcommand\RuleASRecLeftX{%
  \inferrule[AS-Rec-Left-X]
  { \TEnv \vdash x : \TNat \\
    (\not\exists V)~\JAlgConv\TEnv x V{\TNat} \\
    \JAlgSubtypeSynthExt{\TEnv, y: \Equation[\TNat] x \Zero}{ A}{ C}{
      \Kind}{\Subst_Z} \\
    \JAlgSubtypeSynthExt{\TEnv, z:\TNat, y: \Equation[\TNat] x {\Succ
        (z)}}{ B[\TRec z A \alpha\Kind B/\alpha]}{ C}{ \Kind}{\Subst_S}
  }
  { \JAlgSubtypeSynthExt{\TEnv}
    {\TRec x A \alpha\Kind B}
    {C}{\Kind}{\Subst_Z \circ \Subst_S}
  }
}

\newcommand\RuleASCaseRightA{%
  \inferrule[AS-Case-Right1]
  { \JAlgConv\TEnv V {\ell'} L \\
    \ell' \in L\\
    \JAlgSubtypeSynth\TEnv B{ A_{\ell'}} \Kind
  }
  { \JAlgSubtypeSynth\TEnv
    B 
    {\CASE V {\overline{\ell: {A_\ell}}^{\ell\in L}}}
    \Kind
  }
}

\newcommand\RuleASCaseRightB{%
  \inferrule[AS-Case-Right2]
  { \JAlgTypeUnfold\TEnv x L \\
    L \subseteq L' \\
    (\not\exists\ell')~\JAlgConv\TEnv x {\ell'}{ L'} \\
    (\forall\ell\in L)~
    \JAlgSubtypeSynth{\TEnv, y: \Equation[L] x \ell} B{ A_\ell}{ \Kind_\ell}
  }
  { \JAlgSubtypeSynth \TEnv
    B
    {\CASE x {\overline{\ell: {A_\ell}}^{\ell\in L'}}}
    { \bigsqcup_{\ell\in L}\Kind_\ell}
  }
}

\newcommand\RuleASCheck{%
  \inferrule[AS-Check]
  { \JAlgSubtypeSynth\TEnv A B \Kindm \\
    \Kindm \Subkind \Kindn
  }{
    \JAlgSubtypeCheck\TEnv A B {\Kindn}
  }
}

\newcommand\RuleACRefl{%
  \inferrule[AC-Refl]{}{
    \JAlgConv\TEnv\ell\ell{\Setof{\ell}}
  }
}

\newcommand\RuleACAss{%
  \inferrule[AC-Assoc]{}{
    \JAlgConv{\TEnv, y : \Equation[L]x\ell , \SEnv} x\ell L
  }
}

\newcommand\RuleACReflNat{%
  \inferrule[AC-Refl-A]{
    \JAlgTypeUnfold\TEnv V A
  }{
    \JAlgConv\TEnv V V{A}
  }
}

\newcommand\RuleACAssNat{%
  \inferrule[AC-Assoc-A]{}{
    \JAlgConv{\TEnv, y : \Equation[A]xV , \TEnv'} x V A
  }
}

\newcommand\RuleAName{%
  \inferrule[A-Name]{}{
    \JAlgTypeSynth{\PlainExtend{\TEnv_1} zA, \TEnv_2}{ z}{ A}
    {\DemoteExtend{\TEnv_1} zA, \TEnv_2}
  }
}

\newcommand\RuleAUnitI{%
  \inferrule[A-Unit-I]{}{\JAlgTypeSynth\TEnv \EUnit  \TUnit\TEnv}
}

\newcommand\RuleALabI{%
  \inferrule[A-Lab-I]{}{\JAlgTypeSynth\TEnv \ell{ \Setof{\ell}} \TEnv}
}

\newcommand\RuleALabEA{%
  \inferrule[A-Lab-E1]{
    \JAlgConv{\UNR\TEnv} V {\ell'} { L'}\\
    \ell' \in L \\
    \JAlgTypeSynth{\TEnv}{ M_{\ell'}}{ A}{ \TOut}\\
  }{\JAlgTypeSynth{\TEnv}{ \CASE V {\overline{\ell: M_\ell}^{\ell\in L}}}{ A}{ \TOut}
  }
}

\newcommand\RuleALabEB{%
  \inferrule[A-Lab-E2]{
    (\not\exists\ell')~\JAlgConv{\UNR\TEnv} x {\ell'} {L'} \\
    \JAlgTypeUnfold{\UNR\TEnv} x L \\
    L \subseteq L' \\
    (\forall\ell\in L)~
    \JAlgTypeSynth{\TEnv, y : \Equation[L] x \ell}
    { M_\ell}{ A_\ell}{ \TOut, y : \Equation[L] x \ell}
  }{
    \JAlgTypeSynth{\TEnv}{ \CASE x {\overline{\ell: M_\ell}^{\ell\in L'}}} { \CASE x {\overline{\ell:A_\ell}^{\ell\in L}} }{ \TOut}
  }
}

\newcommand\RuleAPiI{%
  \inferrule[A-Pi-I]{
    \JAlgKindSynth{\UNR\TEnv} A{ \Kindn} \\
    \JAlgTypeSynth{\PlainExtend{\TEnv} x A} M
    B {\DemoteExtend{\TOut} x A}
   \\
   \Multm = \MultUn \text{ implies } \TEnv = \TOut 
  }{\JAlgTypeSynth{\TEnv}{ \lambda_\Multm (x: A).M}{ \Pi_\Multm
    (x:  A)B}{ \TOut}
  }
}

\newcommand\RuleAPiE{%
  \inferrule[A-Pi-E]{
    \JAlgTypeSynth{\TEnv_1}{ M}{ C}{ \TEnv_2}\\
    \JAlgTypeUnfold{\UNR{\TEnv_2}}{C}{{\Pi_\Multm(x: A)}B} \\
    \JAlgTypeCheck{\TEnv_2}{ N}{ A}{ \TEnv_3}\\
    \JAlgKindSynth{\UNR\TEnv_1}{ B[N/x]} \Kindn
  }{
    \JAlgTypeSynth{\TEnv_1}{ M\,N}{ B[N/x]}{ \TEnv_3}
  }
}

\newcommand\RuleProcExpr{%
  \inferrule[Proc-Expr]{\TEnv \vdash M : \TUnit}{\TEnv \vdash \Proc{M}}
}

\newcommand\RuleProcChannel{%
  \inferrule[Proc-Channel]{
    \UNR\TEnv \vdash S : \KindLin\\
    \TEnv, \Chanc: S, \Chand: \DUAL S
    \vdash P
  }{
    \TEnv \vdash \NUC \Chanc\Chand~P
  }
}

\newcommand\RuleProcPar{%
  \inferrule[Proc-Par]{
    \TEnv_1 \vdash P \\
    \TEnv_2 \vdash Q
  }{
    \DecomposeOp{\TEnv_1}{\TEnv_2} \vdash P \PAR Q
  }
}

\newcommand\RuleNatF{%
  \inferrule[Nat-F]{
    \EnvForm \MultUn \TEnv
  }{
    \TEnv \vdash \TNat:\KindUn
  }
}

\newcommand\RuleTVarF{%
  \inferrule[TVar-F]{
    \EnvForm \MultUn {\TEnv,\TVar, \TEnv'}
  }{
    \TEnv, \TVar, \TEnv' \vdash \TVar_p : \Un 
  }
}

\newcommand\RuleRecF{%
  \inferrule[Rec-F]{
    \TEnv \vdash V : \TNat\\
    \TEnv \vdash A : \Kind\\
    \TEnv,\TVar \vdash B : \Kind\\
  }{
    \TEnv \vdash \TRec{V}{A}{\TVar}{\Kind}{B}: \Kind
  }
}

\newcommand\RuleConvZ{%
  \inferrule[Conv-Z]{
    \TEnv \vdash A : \Kind\\
    \TEnv,\TVar \vdash B : \Kind\\
  }{
    \TEnv \vdash \Convertible {\TRec{\Zero}{A}{\TVar}{\Kind}{B}} A : \Kind
  }
}

\newcommand\RuleConvS{%
  \inferrule[Conv-S]{
    \TEnv \vdash V : \TNat\\
    \TEnv \vdash A : \Kind\\
    \TEnv,\TVar \vdash B : \Kind 
  }{
    \TEnv \vdash \Convertible
    {\TRec{\Succ(\!V\!)}{A}{\TVar}{\Kind}{B}}
    {B[\TRec{V}{A}{\TVar}{\Kind}{B}/\TVar]} : \Kind
  }
}

\newcommand\RuleZI{%
  \inferrule[Nat-I-Z]{
    \EnvForm \MultUn \TEnv
  }{
    \TEnv \vdash \Zero : \TNat
  }
}

\newcommand\RuleSI{%
  \inferrule[Nat-I-S]{
    \TEnv \vdash M : \TNat
  }{
    \TEnv \vdash \Succ(M) : \TNat
  }
}

\newcommand\RuleANatE{%
  \inferrule[A-Nat-E]{
    \JAlgTypeCheck{\UNR\TEnv} V \TNat{\_} \\
    \JAlgTypeCheckExt{\TEnv,\_: \Equation[\TNat]V\Zero} M
    {C(\TVar_Z)} {\TOut_Z,\_: \Equation[\TNat]V\Zero} {\Subst_Z}\\
    \JAlgTypeCheckExt{\UNR\TEnv, x:\TNat,\TVar\colon\Kind, y\colon C(\TVar),\_\colon \Equation[\TNat] V {\Succ(x)}}
    N {C(\TVar_S)} {\TOut_S, x:\TNat,\TVar\colon\Kind, y\colon
      C(\TVar),\_\colon \Equation[\TNat] V {\Succ(x)}} {\Subst_S} \\
    A = \Subst_Z (\TVar_Z) \\
    B = \Subst_S (\TVar_S)
  }{
    \JAlgTypeSynth\TEnv{ \RECURSOR VMx{y:C (\TVar)}N}
    {C(\TRec V A\TVar\Kind B)} \TOut
  }
}

\newcommand\RuleRecI{%
  \inferrule[Nat-E]{
    \UNR\TEnv \vdash V : \TNat \\
    \DemoteExtend\TEnv z{ \Equation V \Zero} \vdash M : A[V/x] \\
    \DemoteExtend{\UNR\TEnv, x:\TNat,y\colon A} z {\Equation V {\Succ(x)}} \vdash
    N : A[V/x]
  }{
    \TEnv \vdash \RECURSOR VMxyN : A[V/x]
  }
}

\newcommand\RuleSubRec{%
  \inferrule[Sub-Rec]{
    \TEnv \vdash V : \TNat \\
    \TEnv \vdash A \Subtype A' : \Kind \\
    \TEnv,\TVar \vdash B \Subtype B' : \Kind 
  }{
    \TEnv \vdash \TRec VA\TVar\Kind B \Subtype \TRec V{A'}\TVar\Kind {B'}  : \Kind 
  }
}


%% file: abstract.tex
Session types have emerged as a typing discipline for communication
protocols.  Existing calculi with session types come equipped with
many different primitives that combine communication with the
introduction or elimination of the transmitted value.

We present a foundational session type calculus with a lightweight
operational semantics.  It fully decouples communication from the
introduction and elimination of data and thus features a single
communication reduction, which acts as a rendezvous between senders
and receivers.  We achieve this decoupling by introducing
label-dependent session types, a minimalist value-dependent session type
system with subtyping. The system is sufficiently powerful to simulate existing
functional session type systems. Compared to such systems,
label-dependent session types place fewer restrictions on the code.
We further introduce primitive recursion over natural numbers at the type
level, thus allowing to describe protocols whose behaviour depends on
numbers exchanged in messages.
An algorithmic type checking system is introduced and proved
equivalent to its declarative counterpart.
The new calculus showcases a novel lightweight integration of dependent types and
linear typing, with has uses beyond session type systems. 


%% file: introduction.tex
\section{Introduction}
\label{sec:introduction}

Session types enable fine-grained static control over communication
protocols. They evolved from a structuring device for two-party
communication in $\pi$-calculus \cite{Honda1993,TakeuchiHondaKubo1994,HondaVasconcelosKubo1998} over calculi embedded
in functional languages
\cite{VasconcelosRavaraGay2006,GayVasconcelos2010-jfp} to a powerful
means of describing multi-party orchestration of communication
\cite{HondaYoshidaCarbone2008,DBLP:journals/jacm/HondaYC16}. There are embeddings in
object-oriented languages
\cite{GayVasconcelosRavaraGesbertCaldeira2010,DezaniCiancagliniDrossopoulouMostrousYoshida2009}
and uses in the context of scripting languages
\cite{hondamukhamedovbrownchenyoshida2011}, just to mention a few.
Their logical foundations
have been investigated with interpretations in intuitionistic and
classical linear logic \cite{CairesPfenning2010,DBLP:journals/mscs/CairesPT16,Wadler2012}.

There is a range of designs for foundational calculi for session types
\cite{CastagnaDezaniCiancagliniGiachinoPadovani2009,DBLP:journals/iandc/Vasconcelos12,CairesPfenning2010,DBLP:journals/mscs/CairesPT16}.
They all use a session type to describe a sequence of messages. Its
primitive constituents are sending and receiving a typed message
($!A.S$ and $?A.S$), signaling an internal choice ($R \oplus S$),
reacting to an external choice ($R \with S$), and marking the end of a
conversation (\textbf{end}), which is sometimes decomposed into an
active and a passive end marker (\lstinline{end!}  and
\lstinline{end?}). Types in dependent session calculi
\cite{ToninhoCairesPfenning2011,DBLP:conf/fossacs/ToninhoY18}
furthermore contain quantifiers $\forall x:A.S$ and $\exists x:A.S$.
This distinction is well-motivated by logical concerns and results in
different proof term constructions for each of the session operators.
At the operational level, however, the types $!A.S$, $R \oplus S$,
\lstinline{end!}, and $\forall x:A.S$ are implemented by sending a
message and then acting on it in some manner.  It seems wasteful to
have many different syntactic forms that fundamentally perform the
same operation. Moreover, it would be closer to an actual
implementation to have primitive operations for message passing and
have the subsequent actions performed using standard types.

Other researchers also strived to reduce the number of primitive
communication operations in session calculi. For instance,
\citet{DBLP:conf/icfp/LindleyM16} (following
\citet{DBLP:conf/ppdp/DardhaGS12,DBLP:conf/unu/Kobayashi02}) elide
special expressions for internal and external choice by expressing
choice using a standard sum type. In their encoding $\overline{R \oplus
  S}$ is $!(\overline{R} + \overline{S}).$\lstinline{end!}, which has the
same high-level behavior, but the actual messages that are exchanged
are quite different.
The standard implementation of  $\overline{R \oplus S}$ on the wire
sends a single bit to indicate the choice to the
receiver and then continues on the same conversation,
whereas the implementation of $!(\overline{R} +
\overline{S}).$\lstinline{end!} sends a representation of the sum
value, one bit and then a serialization of a channel for $R$ or a
channel $S$, closes the conversion, and continues the protocol on the
other end of the $R$ or $S$ channel.
Clearly, the two implementations are not wire compatible with one
another. Moreover, the encoding using sum types is more
expensive to implement as it involves higher-order channel passing: a
new channel must be created, serialized, sent over the existing channel, and
deserialized at the other end \cite{HuYoshidaHonda2008}.

\citet{DBLP:conf/esop/Padovani17} proposes a different
encoding that does not require the creation of new channels or channel
passing. The encoding of internal
choice sends one of the constructor functions of the sum type and clever typing
guarantees that the type of the channel changes appropriately. 

Our calculus \LDGV{} of Label-Dependent Session Types
is yet more economic in that it requires just a single pair of
communication operations, send and receive, to implement a binary
session-type calculus. Moreover, this implementation does not require
higher-order communication, nor transmission of functions, nor clever
retyping to achieve type soundness, session fidelity, and
communication safety. In particular, the encoding of binary choice
only needs to transmit one bit.

Label dependency is a very limited form of dependent types where
values can depend on labels drawn from a finite set. The labels play
the role of labels in internal and external choices, similar to
variant labels in polymorphic variant types
\cite{DBLP:conf/icfp/CastagnaP016,Garrigue1998} or first-class record
labels \cite{Nishimura1998}. Hence, session types in \LDGV{} are
dependent as in $!(x:A)B$ or $?(x:A)B$, which means to send $A$
(or receive $A$) and continue as $B$, which may depend on $x$.

Labels can further serve as end markers
in protocols and thus the label-dependent calculus is ``wire
compatible'' to the standard encoding of functional session types like
the \LAST{} calculus~\cite{GayVasconcelos2010-jfp} of Linear
Asynchronous Session Types. In fact, a
synchronous version of \LAST{} can be fully emulated in \LDGV.

Label dependency does not require a full-blown lambda calculus in the
types: a large elimination construct for a finite set of labels---a
case expression on labels---suffices. As a more general example, we
outline an extension with natural numbers and large 
elimination with a recursor.

\LDGV{} reinforces the connection between session types and linear
logic~\cite{CairesPfenning2010,ToninhoCairesPfenning2011,DBLP:journals/mscs/CairesPT16,Wadler2012}. The send
operation maps a channel of dependent type $!(x:A)B$ into a single-use
function $\Pi_\MultLin (x:A)B$ whereas the receive operation takes a
channel of dependent type $?(x:A)B$ to a single-use dependent sum
$\Sigma (x:A)B$.


Last, but not least, \LDGV{} proposes a novel, lightweight approach to integrate
linear types with dependent types. The key is an operator
$\DemoteExtend\TEnv xA$ that conditionally extends a type
environment. Roughly speaking, if $A$ is a linear type, then it
returns $\TEnv$ unchanged; if $A$ is unrestricted, then it returns
$\TEnv, x:A$. This operator enables a uniform treatment of dependent
and non-dependent Pi, Sigma, and other types.
For example, type formation for a Pi type like
$\Pi_\Multm (x:A)B$ checks the type $B$ by using $\DemoteExtend\TEnv x
A$. Conditional extension 
automatically degrades the type $\Pi_\Multm (x:A)B$ to a non-dependent
function type if $A$ is linear. On the other hand, $B$ can depend on
$x$, if $A$ is unrestricted. As we will see, \LDGV{} can only have
meaningful dependencies on label types (and natural numbers in the
extended version).

After providing some motivation in Section~\ref{sec:motivation} and
reminding the reader of binary session types in
Section~\ref{sec:binary-session-types}, we claim the following
contributions for this work, starting in
Section~\ref{sec:proposed-calculus}.
\begin{itemize}
\item A foundational functional session type calculus \LDGV{} with a minimal
  set of communication primitives.
\item A type system with label-dependent types,
  linear session types, and subtyping. Besides the usual $\Pi$- and
  $\Sigma$-types, there are label-dependent types for sending and
  receiving.
\item Support for natural numbers and primitive recursion at the type level
  (Section~\ref{sec:naturals}).
\item A novel approach to integrating linear types and dependent types
  using conditional extension.
\item Standard metatheoretical results (Section~\ref{sec:metatheory-light}).
\item Decidable subtyping and type checking that is sound and complete
  (Section~\ref{sec:algorithmic-typing}). 
\item A typing- and semantics-preserving embedding of synchronous \LAST{} in
  \LDGV{} (Section~\ref{sec:embedding-gv-into}).
\item Implementation of a type checker (Section~\ref{sec:implementation}).
\end{itemize}

Full sets of typing rules, proofs, auxiliary lemmas, and an example
type derivation are available in
\finalconcrete{appendices~\ref{sec:compl-typing-rules}-\ref{sec:proofs-embedding-gv}}.


%% file: motivation.tex
\section{Motivation}
\label{sec:motivation}

Functional session types extend functional programming languages like
Haskell and ML with precise typings for structured communication on
bidirectional heterogeneously typed channels. The typing guarantees
that communication actions never mismatch (session fidelity) and that
only values of the expected type arrive at the receiving end
(communication safety).\footnote{There are session type systems that
  guarantee deadlock freedom, but the systems we consider in this
  paper do not.
}

\subsection{Binary Session Types}
\label{sec:an-example}

\input{fig-motivation-server}

As an example for a typical system with binary session types
\cite{GayVasconcelos2010-jfp,DBLP:journals/jfp/Padovani17}, let's consider the code in
Listing~\ref{lst:motivation:compute-server}. It describes a compute server,
\texttt{cServer}, that accepts two commands, \texttt{Neg} and \texttt{Add}, on a channel,
then receives one or two integer arguments depending on the command, performs the
respective operation, sends the result, and closes the channel.

The channel type in the argument of \lstinline{cServer} starts with an
external choice \lstinline{&} between the two commands. After
receiving command \lstinline{Neg}, the channel has type
\lstinline{?Int. !Int. end!}, that is, receive an integer, send an
integer, and then close the channel. The case for command
\lstinline{Add} is analogous.

The structure of the code follows the structure of the type. Variable
\lstinline{c} is processed linearly as it changes type according to
the state of the channel bound to it. The \lstinline{rcase} (receiving
case) receives a label on the channel and branches accordingly. The
\lstinline{c} following labels \lstinline|Neg| and \lstinline|Add| is
a binder for the updated channel endpoint in the two branches. The
\lstinline{recv} operation returns a pair of the updated channel and
the received value, \lstinline{send} returns the updated channel, and
\lstinline{close} consumes the channel end and returns unit.

Listing~\ref{lst:motivation:compute-client} shows a potential client for this server.
The type of the channel end \lstinline{d} is an internal choice
$\oplus$ with a single \lstinline{Neg}-labeled branch. The code 
again follows the type structure.
It performs a \lstinline{select}
operation, which sends the \lstinline{Neg} label (an internal choice), then sends an
integer operand, receives an integer result, and acknowledges channel
closure with the \lstinline{wait}
operation. This single-branch channel type is a supertype of the two-branch type
in Listing~\ref{lst:motivation:dual-cserver-type},
which is the \emph{dual} of the server's channel type. The dual type has the
same structure as the original type, but with sending and receiving
types exchanged. All session type systems require that the types of
the two endpoints of a channel are duals of one another to guarantee
session fidelity. Of course, channel ends can be used at any suitable supertype.

\subsection{The Case for Economy}
\label{sec:an-issue}

The example demonstrates an issue that makes programming with
session types more arcane than necessary. There are three different
send operations, \lstinline{send}, \lstinline{select}, and
\lstinline{close}, and three matching receive operations,
\lstinline{recv}, \lstinline{rcase}, and \lstinline{wait}. They are reasonably
easy to use, but they lead to bloated APIs for  session types. The multitude of operations
also bloats the syntax and semantics of foundational calculi for session types. 

Wouldn't it be enticing if there was a session calculus with just a
single pair of primitives for sending and receiving messages? The
resulting functional session type calculus would be close to an
implementation as it would have just one reduction for communication
alongside the standard expression reductions.

The problem with the existing calculi is that they entangle the
sending/receiving of data with another unrelated operation, which
introduces the sent data or eliminates the received data.  The
\textbf{calculus of label-dependent session types} disentangles
communication from introducing and eliminating the data values. It
comes with a type of first-class labels that plays the role of labels
in choice and branch types of traditional session type systems. The
calculus features dependent product and sum types where the dependency
is limited to labels. The types of the sending and receiving
operations can be dependent on the transmitted values if they are
labels.

\input{fig-motivation-ldgv-server}

For illustration, we translate the server session type from
\S~\ref{sec:an-example} to a label-dependent session type
\lstinline{TServer} in Listing~\ref{lst:motivation:type-server}. One new ingredient is the type
\lstinline/{l1,...,ln}/, which denotes the non-empty set of labels
\lstinline{l1} through \lstinline{ln}. The other new ingredient is the
\lstinline{case} expression in the type, which dispatches on a label
to determine the type of the subsequent communication. This
type introduces label dependency.

Like the \lstinline{rcase} operation, a channel of this type first receives
a label, on which the rest of the type depends. Type-level reduction of the \lstinline{case} on the label reveals the type of the rest of the channel. 
If the label is \lstinline{Neg},
then the channel can receive an integer (nothing depends on it), send an
integer, and finally send a special end-of-session label (\lstinline|EOS|).
The case for label \lstinline{Add} is analogous.

Listing~\ref{lst:motivation:ldgv-compute-server} contains the code for
a server of this type. It relies entirely on primitive send and
receive operations. The \lstinline{rcase} operation, which is typical
of previous work, decomposes into a standard \lstinline{recv}
operation followed by an ordinary \lstinline{case} on
labels. Moreover, the structure of the code is liberated from the
session type. While the standard session type dictates the placement
of the \lstinline{rcase}, the \LDGV{} version can examine the tag any
time after receiving it. The server code takes advantage of this
liberty by pulling the common receive operation for the first argument
out of the two branches.

A compatible client (Listing~\ref{lst:motivation:ldgv-compute-client}) does not have to
know about the choice. Its channel argument type is a supertype of the dual of
\lstinline{TServer} in our calculus (cf.\ type \lstinline{TClient} in
Listing~\ref{lst:motivation:dual-tserver}).

Section~\ref{sec:embedding-gv-into} shows that any program using binary session types can
be expressed with label-dependent session types in a semantics-preserving way. The
examples shown in this section give a preview of this embedding.

\subsection{Tagged Data and Algebraic Datatypes}
\label{sec:furth-opport}

\input{fig-motivation-node}

Some session calculi support the transmission of tagged data as a
primitive
\cite{DBLP:conf/ecoop/ScalasY16,DBLP:journals/lmcs/ChenDSY17,DBLP:conf/isotas/VasconcelosT93}.
In these works, operations of the form $c!\textit{Node}(42)$ are
used to send a \textit{Node}-tagged message with payload $42$ on
channel $c$, effectively combining a \lstinline{select} operation with
the sending some extra data. The corresponding receiving construct
dispatches on the tag, as in \RCASEK, and also extracts the payload into
variables. This construction is akin to packaging tags with data and
pattern matching as known from algebraic datatypes in functional
programming languages.

Indeed, such algebraic datatypes can be modeled in the functional
sublanguage of \LDGV{} using a label-dependent $\Sigma$-type. As an example,
consider the datatype
\begin{lstlisting}
data Node where { Empty: Node
                , Node : Int -> Node }
\end{lstlisting}
and its \LDGV{} representation
\begin{lstlisting}
type Node = 
     Sigma(tag: {Empty, Node})
     case tag of { Empty: Unit, Node: Int }
\end{lstlisting}
Sending (receiving) a single value of type \lstinline{Node} can be
performed on a channel of type \lstinline{NodeC} (or its dual) as
illustrated with \lstinline{sendNode} and \lstinline{recvNode} in
Listing~\ref{lst:motivation:sending-receiving-nodes}.

Recursive datatypes and session type protocols can be supported by
extending \LDGV{} with recursive types, which we leave to future work.

\subsection{Number-indexed Protocols}
\label{sec:numb-index-prot}

One shortcoming of programming-oriented systems for session types is
that they do not support families of indexed protocols, a quite common
situation in practice. These protocols have variable-length messages
where the first item transmitted gives the number of the subsequent
items.

To demonstrate how \LDGV{} can deal with such protocols, consider the code in
Listing~\ref{lst:motivation:summing} which implements a server that
first receives a number $n$ and then expects to receive $n$ further
numbers, sums them all up and sends them back to the client. The type
of this channel is given by
\begin{lstlisting}
type SumServer = ?(n:Nat). rec n (!Int.End) [alpha]?Int.alpha
\end{lstlisting}
The interesting part is the type of the form \lstinline/rec n S [alpha]R/, which denotes a
type-level recursor on natural numbers. Its first argument is a
number \lstinline/n/, the second argument \lstinline/S/ is used if
\lstinline/n/ is zero, and the third argument \lstinline/[alpha]R/ is
used when \lstinline/n/ is non-zero and \lstinline/alpha/ abstracts
over the recursive use. In this case, the type is
equivalent to \lstinline/R/ where \lstinline/alpha/ is replaced by the
unwinding of the recursor.

In the example, the types evaluate as follows
\begin{itemize}\item 
\lstinline/rec 0 !Int.End [alpha]?Int.alpha/ $\equiv$
  \lstinline/!Int.End/, 
\item \lstinline/rec 1 !Int.End [alpha]?Int.alpha/
  $\equiv$ \lstinline|?Int.rec 0 !Int.End [alpha]?Int.alpha|
  $\equiv$ \lstinline|?Int.!Int.End|,
  and so on.
\end{itemize}

The implementation of the server has to use the corresponding recursor
at the value level. If we write \lstinline/T n/ for
\lstinline/rec n (!Int.End) [alpha]?Int.alpha/, then the recursor
returns a function of type 
\lstinline/Int -> T n -> End/ and is given an expression of type
\lstinline/Int -> !Int.End -> End/ $\equiv$
\lstinline/Int -> T 0 -> End/ for the case zero.
In the successor case, the expression has type
\lstinline/Int -> ?Int.T n -> End/ $\equiv$
\lstinline/Int -> T (S n) -> End/ and
the variable \lstinline/y/ is bound to the function returned by the
unwinding of the recursor. This code does not use the predecessor as
indicated by the underline in \lstinline/S(_)/. The type annotations
for \lstinline/m/, \lstinline/c/, and \lstinline/y/ have to be given to enable type checking. 

\subsection{Assessment}
\label{sec:assessment}

Moving to the dependent calculus \LDGV{} has a number of
advantages over a traditional calculus like \LAST. It liberates the
program structure somewhat from the session type structure and
increases expressivity as shown in the preceding subsections. 
\begin{itemize}
\item In \LDGV, a label-dependent choice in the type can be deferred
  in the program. Listing~\ref{lst:motivation:ldgv-compute-server}
  does not type check in \LAST{} because it defers the
  choice compared to the session type. 
\item \LDGV{} supports first class labels. Listing~\ref{lst:motivation:sending-receiving-nodes} cannot be
  written in this generic way in \LAST{} because the functions \lstinline{sendNode} and
  \lstinline{recvNode} transfer labels without inspecting them. In
  \LAST{} the receive operation \lstinline{rcase} also
  inspects the label: the code would have to be eta-expanded depending on the
  label set used in the \lstinline{Node} type. The \LDGV{} code is
  resilient against such changes. New variants in \lstinline{Node} and
  \lstinline{NodeC} can be processed without rewriting the code. 
\item \LDGV{} supports types and protocols defined by recursion on
  natural numbers. Listing~\ref{lst:motivation:summing} cannot be
  written in \LAST{} without major changes in the protocol that make
  it very inefficient. One would have to change the data stream into a
  list with intervening labels. 
\end{itemize}


%% file: fig-motivation-server.tex
\begin{figure*}[tp]
  \begin{minipage}[t]{0.53\linewidth}
\begin{lstlisting}[caption={Compute server},label={lst:motivation:compute-server},captionpos=b]
cServer :
  & { Neg: ?Int. !Int. end!
    , Add: ?Int. ?Int. !Int. end! }
  -> Unit
cServer c =
  rcase c of {
    Neg: c. let (x, c) = recv c 
                c = send c (-x)
            in  close c,
    Add: c. let (x, c) = recv c
                (y, c) = recv c
                c = send c (x+y)
            in  close c
  }
\end{lstlisting}
 \end{minipage}
 \begin{minipage}[t]{0.46\linewidth}
\begin{lstlisting}[caption={Dual of \lstinline{cServer}'s session type},label={lst:motivation:dual-cserver-type},captionpos=b]
(+) { Neg: !Int. ?Int. end?
  , Add: !Int, !Int. ?Int. end?}
\end{lstlisting}
\medskip
\begin{lstlisting}[caption={Compute client},label={lst:motivation:compute-client},captionpos=b]
negClient :
  (+) { Neg: !Int. ?Int. end? }
  -> Int -> Int
negClient d x =
  let d = select Neg d
      d = send d x
      (r, d) = recv d
      wait d
  in  r
\end{lstlisting}
  \end{minipage}
\end{figure*}

%% file: fig-motivation-ldgv-server.tex
\begin{figure*}[tp]
\begin{minipage}[t]{0.49\linewidth}
\begin{lstlisting}[caption={Label-dependent server type},label={lst:motivation:type-server},captionpos=b]
type TServer =
  ?(l:{Neg, Add}).
  case l of
  { Neg: ?Int. !Int. !{EOS}. End
  , Add: ?Int. ?Int. !Int. !{EOS}. End }
\end{lstlisting}
\begin{lstlisting}[caption={Label-dependent compute server},label={lst:motivation:ldgv-compute-server},captionpos=b]
lServer : TServer -> End
lServer c =
  let (l, c) = recv c
      (x, c) = recv c 
  in case l of 
  { Neg: let c = send c (-x)
         in  send c EOS,
  , Add: let (y, c) = recv c
             c = send c (x+y)
         in  send c EOS
  }
\end{lstlisting}
\end{minipage}
 \begin{minipage}[t]{0.49\linewidth}
\begin{lstlisting}[caption={Dual type of \lstinline{TServer}},label={lst:motivation:dual-tserver},captionpos=b]
type TClient =
  !(l:{Neg, Add}).
  case l of
  { Neg: !Int. ?Int. ?{EOS}. End
  , Add: !Int. !Int. ?Int. ?{EOS}. End }
\end{lstlisting}
\begin{lstlisting}[caption={Compute client},label={lst:motivation:ldgv-compute-client},captionpos=b]
lClient : !{Neg}. !Int. ?Int. ?{EOS}. End -> Int -> Int
lClient d x =
  let d = send d 'Neg
      d = send d x
      (r, d) = recv d
      (_, _) = recv d // ((), EOS)
  in  r
\end{lstlisting}
\end{minipage}
\end{figure*}


%% file: fig-motivation-node.tex
\begin{figure*}[tp]
\begin{minipage}[t]{0.49\linewidth}
\begin{lstlisting}[caption={Sending and receiving nodes},label={lst:motivation:sending-receiving-nodes},captionpos=b]
type NodeC =
  !(tag : {Empty, Node}).
  case tag of { Empty: !Unit. End
              , Node: !Int. End }

sendNode : Node -> NodeC -> Unit
sendNode n c =
  let (tag, v) = n
      c = send c tag
  in  send c v

recvNode : dualof NodeC -> Node
recvNode c =
  let (tag, c) = recv c
      (val, c) = recv c
  in  (tag, val)
\end{lstlisting}
\end{minipage}
\begin{minipage}[t]{0.49\linewidth}
\begin{lstlisting}[caption={Summing a given number of integers},label={lst:motivation:summing},captionpos=b]
type SumServer =
  ?(n:Nat).rec n (!Int.End) [alpha]?Int.alpha

sum : SumServer -> End
sum c =
  let (n,c) = recv c in
  rec n {
    Z: lambda(m:Int).lambda(c:!Int.End).
       send c m,
    S(_) with [alpha](y:Int -> alpha -> End):
      lambda(m:Int).lambda(c:?Int.alpha).
        let (k,c) = recv c in
        y (k+m) c
  } 0 c                   
\end{lstlisting}
\end{minipage}
\end{figure*}


%% file: binary-session-types.tex
\section{Binary Session Types}
\label{sec:binary-session-types}


\input{fig-gv-types}

The type structure for a functional calculus with binary session
types, known as \LAST, adds session types $S$ to the types of an
underlying lambda calculus with functions and products (cf.\
\citet{GayVasconcelos2010-jfp}). The examples in
Section~\ref{sec:an-example} are written in \LAST{} with some syntactic
sugar.

The calculus \LSST{} (for Linear Synchronous Session Types) introduced
in this section is a slight variation of Gay and Vasconcelos \LAST{} 
calculus. First, we choose a synchronous
semantics for the communication primitives. This choice has no impact
on the typing, but greatly simplifies the semantics and
proofs. Second, we adopt the linear-logic inspired end markers $\ENDS$
and $\ENDR$ from Wadler's GV calculus \cite{Wadler2012}. Unlike GV,
the \LSST{} calculus is not free of deadlock.

Figure~\ref{fig:gv-types} defines the syntax of types and the notion
of subtyping.  Function types are annotated with multiplicities (also
called \emph{kinds}), $\Multm, \Multn\in \{\MultLin, \MultUn\}$, that
restrict the number of eliminations that may be applied to a value of
that type: $\MultLin$ denotes a linear value that must be eliminated
exactly once, $\MultUn$ denotes an unrestricted value that may be
eliminated as many times as needed. Session types are always
linear.
The subkinding relation $\Multm\MultLE\Multn$  relates multiplicities:
if a value offers elimination according to $\Multm$ it may also be
eliminated according to $\Multn$. In particular, an unrestricted value
may also serve as a linear value.
The predicate $\EnvForm \Multm A$ determines the multiplicity of a
type.

In session types, the branch labels $\ell$ are drawn from a
denumerable set $ \LabelUniv$ of labels.  Overlining indexed by some
$\ell$ indicates an iteration over a finite non-empty set of labels
$\LabelSet \subseteq \LabelUniv$. The type ${!A}.S$ indicates sending
a value of type~$A$ and continuing according to $S$; ${?A}.S$
indicates receiving a value of type~$A$ and continuing according to
$S$; the type $\oplus \{\overline{\ell: S_\ell}\}$ stands for sending
a label $\ell\in \LabelSet $ and then continuing according to
$S_\ell$; and $\& \{\overline{\ell: S_\ell}\}$ stands for receiving a
label $\ell\in \LabelSet$ and continuing with the chosen $S_\ell$. The
session types $\ENDS$ and $\ENDR$ indicate closing the communication
and waiting for the other end to close.

\LSST{}'s subtyping is driven by multiplicities (linear values may
subsume unrestricted ones) and by varying the number of alternatives
in branch and choice types (corresponding to width subtyping of
records and variants) \cite{GayHole2005,GayVasconcelos2010-jfp}.


\input{fig-gv-typing}

Figure~\ref{fig:gv-typing} describes the syntax of names, expressions,
and processes in \LSST. Names include variables, $x,y$, and channel
endpoints, $c,d$. Expressions comprise names, the unit value, pair and
function introduction and elimination (in $\KindLin$ and $\KindUn$
versions), and the standard primitives of \LSST{}. Process expressions are either expression
processes, parallel processes, or a channel restriction
$\NUC\Chanc\Chand P$ that binds the two channel endpoints $\Chanc$ and
$\Chand$ in the scope provided by process $P$.

We refrain from giving the full set of \LSST{} typing rules here.
As an example, we give the standard typing rules for sending and
receiving data
and go over rule \TirName{GV-SEND} for illustration. The $\SENDK$
operation takes a channel endpoint $M$ of type $!A.S$, which is good
to write a value of type $A$. Then $\SEND M$ is a function from $A$
to $S$, which must be used once (because it is closed over the channel
endpoint).
We flip the arguments for $\SENDK$ with respect to other presentations
in the
literature~\cite{GayVasconcelos2010-jfp,igarashithiemannvasconceloswadler2017,DBLP:conf/icfp/LindleyM16,Wadler2012},
while aligning with those of \citet{DBLP:journals/jfp/Padovani17}.
The rule \TirName{GV-NEW} for creating channels prescribes that
$\NEW$ returns a pair of channel endpoints with dual session
types. Alternatively, to obtain a deadlock-free calculus, we could
couple channel creation with thread creation as in the cut rule of
\citet{Wadler2012} or the fork primitive of
\citet{DBLP:journals/corr/LindleyM14}. 

Like other type systems with a mix of linear and unrestricted
resources
\cite{Walker2005-attapl,DBLP:conf/lics/CervesatoP96,KobayashiPierceTurner1996},
\LSST{} relies on an environment splitting relation
$\Decompose\TEnv{\TEnv_1}{\TEnv_2}$. As a slight abuse of notation, we
sometimes write $\DecomposeOp{\TEnv_1}{\TEnv_2}$ for some $\TEnv$ such that
$\Decompose\TEnv{\TEnv_1}{\TEnv_2}$.


\input{fig-gv-reduction}

The reduction relation for the \LSST{} language is in
Figure~\ref{fig:gv-reduction}. It introduces values $V, W$, which
comprise the usual lambda calculus variety, a communication channel
endpoint~$c$, a partially applied send operation $\SEND v$, and a
select operation with a label $\SELECT\ell$. Evaluation contexts
$\ECN, \ECN[F]$ formalize a left-to-right call-by-value evaluation
order. Unlike in \LAST, communication in \LSST{} is synchronous and we
add the \TirName{Rl-Close} reduction.

We refrain from defining expression reduction; instead we refer the
reader to
\citet{GayVasconcelos2010-jfp}.
But we fully define process reduction. It relies on \emph{structural
  congruence}, $P\equiv Q$, a relation that specifies that parallel
execution is commutative, associative, and compatible with channel
restriction and commutation of channel restriction. We assume the
variable convention: for example, in the rule \TirName{sc-swap-r} for
commuting restrictions it must be that
$\{c,d\} \cap \{c', d'\} = \emptyset$. The rule \TirName{sc-swap-c}
that swaps the endpoints simplifies the statement of the reduction
relation \cite{igarashithiemannvasconceloswadler2017}. Examining the
reduction rules, we observe that each communication reduction first
performs a rendezvous to transmit information, but the
rules~\eqref{eq:rl-branch} and~\eqref{eq:rl-close} do some extra
work. Part of the motivation for this work comes from trying to
disentangle the extra work from the pure communication.


%% file: fig-gv-types.tex
\begin{figure*}[tp!]
  \begin{align*}
    \text{Multiplicities} &&
   \Multm, \Multn  \grmeq& \MultLin \grmor \MultUn
    \\
    \text{Types} &&
    A, B \grmeq& S \grmor \TUnit \grmor A \to_\Multm B \grmor A \times
           B 
    \\
    \text{Session types} &&
    S, R \grmeq& {!A}.S \grmor {?A}.S \grmor  \oplus \{\overline{\ell: S_\ell}^{\ell\in L}\} \grmor
      \&\{\overline{\ell: S_\ell}^{\ell\in L}\} \grmor \ENDS \grmor \ENDR 
  \end{align*}
  Subkinding\hfill\fbox{$m \MultLE n$}
  \begin{mathpar}
    \inferrule{}{\MultUn \MultLE \MultLin}

    \inferrule{}{\MultLin \MultLE \MultLin}

    \inferrule{}{\MultUn \MultLE \MultUn}
  \end{mathpar}
  Kinding\hfill\fbox{$\GVKind m A$}
  \begin{mathpar}
    \inferrule{}{\GVKind \MultUn \TUnit}
    
    \inferrule{}{\GVKind \Multm {A \to_\Multm B}}
    
    \inferrule{\GVKind \Multm{A} \\ \GVKind \Multm{B}
    }{\GVKind \Multm {A \times B}}
    
    \inferrule{}{\GVKind \MultLin S}

    \inferrule{\GVKind \Multm A \\ \Multm \MultLE \Multn}
    {\GVKind \Multn A}
  \end{mathpar}
  Subtyping\hfill\fbox{$A \Subtype B$}
  \begin{mathpar}
  \inferrule{}{\TUnit \Subtype \TUnit}

  \inferrule{}{ \ENDS \Subtype \ENDS }

  \inferrule{}{ \ENDR \Subtype \ENDR }
  \\
  \inferrule{
    A' \Subtype A \\ B \Subtype B' \\ \Multm \MultLE \Multn
  }{
    A \to_\Multm B \Subtype A' \to_\Multn B'
  }

  \inferrule{
    A \Subtype A' \\ B \Subtype B'
  }{
    A \times B \Subtype A' \times B'
  }

  \inferrule{
    A' \Subtype A \\
    S \Subtype S'
  }{
    {!A}.S \Subtype {!A'}.S'
  }

  \inferrule{
    A \Subtype A' \\
    S \Subtype S'
  }{
    {?A}.S \Subtype {?A'}.S'
  }

  \inferrule{
    L' \subseteq L \\
    (\forall\ell\in L')~ S_\ell \Subtype S'_\ell
  }{
    \oplus \{\overline{\ell: S_\ell}^{\ell\in L}\} \Subtype
    \oplus \{\overline{\ell: S'_\ell}^{\ell\in L'}\}
  }

  \inferrule{
    L \subseteq L' \\
    (\forall\ell\in L)~S_\ell \Subtype S'_\ell
  }{
    \& \{\overline{\ell: S_\ell}^{\ell\in L}\} \Subtype
    \& \{\overline{\ell: S'_\ell}^{\ell\in L'}\}
  }
  \end{mathpar}
  \caption{Types and subtyping in \LSST{}
  }
  \label{fig:gv-types}
\end{figure*}


%% file: fig-gv-typing.tex
\begin{figure*}[tp]
  \begin{align*}
    \text{Names} &&
    z \grmeq& x \grmor c
    \\
    \text{Expressions} &&
    M,N \grmeq&
    z
    \grmor \EUnit
    \grmor \lambda_\Multm x.M
    \grmor M\,N
    \grmor \EPair M N
    \grmor \ELet{\EPair[] xy} M N
    \\ && 
    \grmor& \FORK M
    \grmor \NEW
    \grmor \SEND M
    \grmor \RECV M
    \\ && 
    \grmor& \SELECT \ell
    \grmor \RCASE M {\overline{\CASEBRANCH \ell x N_\ell}^{\ell\in L}}
    \grmor \CLOSE M
    \grmor \WAIT M
    \\
    \text{Processes} &&
    O, P, Q \grmeq& \Proc M \grmor P \PAR Q \grmor \NUC\Chanc\Chand P
    \\
    \text{Typing environments} &&
    \TEnv,\SEnv \grmeq& \EmptyEnv \grmor \TEnv, z:A
  \end{align*}
  Environment formation, environment split
  \hfill\fbox{$\EnvForm \Multm \TEnv$}\quad\fbox{$\Decompose\TEnv{\TEnv_1}{\TEnv_2}$}
  \begin{mathpar}
    \inferrule{}{\EnvForm \Multm \EmptyEnv}

    \inferrule{
      \EnvForm \Multm \TEnv
      \\
      \GVKind \Multm A
    }{
      \EnvForm \Multm {(\TEnv,x:A)}
    }
    \\
  \inferrule{}{
    \Decompose\EmptyEnv\EmptyEnv\EmptyEnv
  }

  \inferrule{
    \Decompose\TEnv{\TEnv_1}{\TEnv_2} \\
    \EnvForm \MultUn A 
  }{\Decompose{(\TEnv, z:A)}{(\TEnv_1, z:A)}{(\TEnv_2, z:A)}}
  %

  \inferrule{
    \Decompose\TEnv{\TEnv_1}{\TEnv_2}
    \\ \EnvForm \MultLin A 
  }{
    \Decompose{(\TEnv, z:A)}{(\TEnv_1, z:A)}{\TEnv_2}
  }

  \inferrule{
    \Decompose\TEnv{\TEnv_1}{\TEnv_2}
    \\ \EnvForm \MultLin A 
  }{
    \Decompose{(\TEnv, z:A)}{\TEnv_1}{(\TEnv_2, z:A)}
  }
\end{mathpar}
  Typing expressions (excerpt)
  \hfill\fbox{$\TEnv \vdash M:A$}
  \begin{mathpar}
  \inferrule[gv-send]
  {
    \TEnv \vdash M : {!A}.S 
  }
  {\TEnv \vdash \SEND M : A \to_\KindLin S }

  \inferrule[gv-recv]
  {\TEnv \vdash M : {?A}.S}
  {\TEnv \vdash \RECV M : A \times S}

  %
  \inferrule[gv-select]
  {\EnvForm \MultUn \TEnv\\
    \ell'\in L
  }
  {\TEnv \vdash \SELECT \ell' : \oplus \{\overline{\ell:
      S_\ell}^{\ell\in L}\} \to_\KindLin S_{\ell'}}

  \inferrule[gv-rcase]
  {
    \TEnv_1 \vdash M : \&\{\overline{\ell: S_\ell}\}^{\ell\in L} \\
    (\forall\ell\in L) \TEnv_2, x:S_\ell \vdash M_\ell : A
  }
  {\DecomposeOp{\TEnv_1}{\TEnv_2}  \vdash \RCASE M
    {\overline{\CASEBRANCH \ell x M_\ell}^{\ell\in L}} : A}

  \inferrule[gv-close]
  {
    \TEnv \vdash M : \ENDS
  }
  {\TEnv \vdash \CLOSE M : \TUnit}

  \inferrule[gv-wait]
  {
    \TEnv \vdash M : \ENDR
  }
  {\TEnv \vdash \WAIT M : \TUnit}

  \inferrule[gv-new]
  {
    \EnvForm \MultUn \TEnv
  }
  {\TEnv \vdash \NEW : S \times \Dual S}
  \end{mathpar}

  \caption{Expressions, processes, and typing in \LSST{}}
  \label{fig:gv-typing}
\end{figure*}


%% file: fig-gv-reduction.tex
\begin{figure*}[tp]
  \begin{align*}
    \text{Values} &&
    V,W \grmeq& c \grmor \EUnit \grmor \lambda_\Multm x.M \grmor \EPair V W \grmor
    \SEND V \grmor \SELECT\ell
    \\
    \text{Evaluation contexts} &&
    \ECN, \ECN[F] \grmeq& \Hole
    \grmor \ECN\,M \grmor V\,\ECN
    \grmor \EPair\ECN f \grmor \EPair V \ECN \grmor \ELet{\EPair[] xy} \ECN M
    \\ &&
    \grmor& \SEND\ECN \grmor \RECV\ECN
    \grmor \RCASE \ECN {\overline{\CASEBRANCH \ell x M}}
    \grmor \CLOSE \ECN
    \grmor \WAIT \ECN
  \end{align*}
  Structural congruence
  \hfill\fbox{$P\equiv Q$}
  \begin{mathpar}
    \inferrule[sc-comm]{}{P\PAR Q \equiv Q \PAR P}

    \inferrule[sc-assoc]{}{O\PAR (P\PAR Q) \equiv (O \PAR P)\PAR Q}

    \inferrule[sc-gc]{}{P \PAR  \Proc{\EUnit} \equiv P}
\\
    \inferrule[sc-extrusion]{}{(\NUC\Chanc\Chand P) \PAR Q \equiv \NUC\Chanc\Chand
      (P \PAR Q)}

    \inferrule[sc-swap-c]{}{\NUC\Chanc\Chand P \equiv \NUC\Chand\Chanc P}

    \inferrule[sc-swap-r]{}{\NUC\Chanc\Chand \NUC{\Chanc'}{\Chand'}P \equiv
      \NUC{\Chanc'}{\Chand'} \NUC\Chanc\Chand P}
  \end{mathpar}
  \renewcommand\columnwidth\textwidth 
  Process reduction
  \hfill\fbox{$P \ReducesTo Q$}
\begin{align}
  \Proc{\Context{\FORK M}}
  & \ReducesTo
    \Proc{\Context{\EUnit}} \PAR \Proc{M}
    \Tag{Rl-Fork}\label{eq:rl-fork}
  \\
  \Proc{\Context{\NEW}}
  & \ReducesTo
    \NUC\Chanc\Chand~\Proc{\Context{\EPair[\Lin]\Chanc\Chand}}
    \Tag{Rl-New}\label{eq:rl-new}
  \\
  \NUC\Chanc\Chand~\Proc{\Context{\SEND\Chanc\,V}} \PAR
  \Proc{\Context[F]{\RECV\Chand}}
  & \ReducesTo
  \NUC\Chanc\Chand\Proc{\Context{\Chanc}} \PAR
    \Proc{\Context[F]{\EPair[\Lin]V\Chand}}
    \Tag{Rl-Com}\label{eq:rl-send-recv}
    \\
    \NUC\Chanc\Chand\Proc{\Context{\SELECT\ell'\Chanc}} \PAR
    \Proc{\Context[F]{\RCASE \Chand {\overline{\CASEBRANCH \ell {x}
            M_\ell}}}}
    &\ReducesTo
    \NUC\Chanc\Chand\Proc{\Context{\Chanc}} \PAR
    \Proc{\Context[F]{M_{\ell'}[\Chand/x]}}
    \Tag{Rl-Branch}\label{eq:rl-branch}
    \\
    \NUC\Chanc\Chand~\Proc{\Context{\CLOSE\Chanc}} \PAR
    \Proc{\Context[F]{\WAIT\Chand}}
    & \ReducesTo
    \Proc{\Context{\EUnit}} \PAR
    \Proc{\Context[F]{\EUnit}}
    \Tag{Rl-Close}\label{eq:rl-close}
  \end{align}
  \begin{mathpar}
    \inferrule[Rl-Ctx-Exp]{M \ReducesTo N}{\Proc M \ReducesTo \Proc N}

    \inferrule[Rl-Ctx-Par]{P \ReducesTo Q}{O\PAR P \ReducesTo O\PAR Q}

    \inferrule[Rl-Ctx-Res]{P \ReducesTo Q}{\NUC\Chanc\Chand P \ReducesTo
      \NUC\Chanc\Chand Q}

    \inferrule[Rl-Cong]{P \equiv P' \\ P' \ReducesTo Q' \\ Q' \equiv Q}{P
      \ReducesTo Q}

  \end{mathpar}
%

  \caption{Reduction in \LSST{}}
  \label{fig:gv-reduction}
\end{figure*}


%% file: ld-session-calculus.tex
\section{The Label-dependent Session Calculus}
\label{sec:proposed-calculus}

We propose \LDGV{}, a new calculus for functional sessions.
Compared to \LSST{}, \LDGV{} introduces types that depend on labels and
restricts the communication instructions to the fundamental send and
receive operations.  The example in Section~\ref{sec:an-issue}
hints that every \LSST{} program can be expressed in \LDGV{}, a claim
formally stated and proved in Section~\ref{sec:embedding-gv-into}.

The dynamics of \LDGV{} are simpler than \LSST{}'s, but its statics
are more involved. They build on a range of earlier work, most notably
trellys \cite{DBLP:conf/popl/CasinghinoSW14,DBLP:journals/corr/abs-1202-2923} and F$^*$
\cite{DBLP:journals/jfp/SwamyCFSBY13}, to formalize a flexible
dependently-typed system based on call-by-value execution augmented
with linear types.



\input{fig-ldgv-kinds-types-occurrences}

%
Figure~\ref{fig:ldgv-kinds-types-occurrences} describes \LDGV{}'s values, types,
and some auxiliary operations on type environments.
Kinds are as in \LSST: $\KindLin$ for linear (single use) types and $\KindUn$ for
unrestricted types;
%
%
and unrestricted values can also be used linearly.

Values comprise the usual lambda calculus values and $\SEND V$ as in
Section~\ref{sec:binary-session-types}.  Recall from
Figure~\ref{fig:gv-typing} that $z$ stands for a variable $x$ or a
channel end~$c$. Variables are included in the set of values as they
can only be bound to values as customary when reasoning with open expressions. 

Types of the calculus comprise session types; the unit type;  the label type $\Setof{\ell_1, \dots, \ell_n}$,
for $n>0$, inhabited by the labels $\ell_1, \dots, \ell_n$---for
brevity, we let $L$ range over finite non-empty sets of labels; the
equality type $\Equation[L] V W$
inhabited by evidence that the value $V$ is equal to value $W$;
the dependent
function and product types $\Pi_\Multm (x:A) B$ and
$\Sigma (x:A)B$ of multiplicity $\Multm$.
Session types comprise $\TEnd$ to signify the end of a session; the dependent
session types $!(x:A)S$ and $?(x:A)S$ for endpoints that send or
receive a value of type $A$ and continue as session type $S$, which may depend
on $x$. The type $\CASE V {\overline{\ell : A_\ell} }$ indicates large
elimination for labels and it may occur in both types and session types.

The basic operations on type environments are inherited from \LSST:
environment formation $\EnvForm \Multn \TEnv$ and splitting
$\Decompose\TEnv{\TEnv_1}{\TEnv_2}$. Both rely on kinding (they are
mutually recursive as expected in a dependently typed calculus) and we present a revised definition of kinding (type
formation) shortly.
Environment formation and environment split for \LDGV{} are both
adapted from \LAST{} (Figure~\ref{fig:gv-typing}). In the case of
formation, premise $\GVKind \Multm A$ becomes
$\TEnv \vdash A : \Multm$ to reflect the new type formation rules
(in Figure~\ref{fig:ldgv-types}).
For environment split we require type~$A$ to be well formed in
the relevant contexts, so that $\TEnv_1,z : A$ and $\TEnv_2,z : A$
both become well formed contexts.

The new operations are conditional extension
$\DemoteExtend\TEnv x A = \SEnv$ and projecting the unrestricted part
of an environment $\UNR\TEnv = \SEnv$.
The conditional extension $\DemoteExtend\TEnv x A$ only includes the
binding $x:A$ in the resulting environment if $A$ is unrestricted. In
the upcoming type formation rules, this mechanism is used to keep
linear values out of the environment so that any dependency on linear
objects is ruled out.

Conditional extension is used in the formation rule for all dependent
types. As an example, take the function type
$\TEnv \vdash \Pi_\Multm(x:A)B : \Multm$. Here, we do not wish to
force type $A$ to be unrestricted. Rather, we wish to express that $B$
can depend on $x:A$ iff $A$ is unrestricted. To this end, the premise
uses the conditional extension to check $B$ as in
$\DemoteExtend\TEnv x A \vdash B : \Multn$. Right now, this setup is
more general than strictly needed because we can only compute with
labels in types, that is, we need $A=L$ and $B$ can at most contain a
\CASEK{} on $x$.

The unrestricted part of an environment is used when switching from
expression formation to type formation or subtyping.
As $\Decompose\TEnv\TEnv{\UNR\TEnv}$\finalreport{}{ (Lemma~\ref{lemma:properties-un})},
it is ok to use an 
environment and its unrestricted part side by side.

\input{fig-ldgv-types}
\input{fig-ldgv-subtyping}

The \emph{dual of a session type}, $\DUAL S$, is also defined in
Figure~\ref{fig:ldgv-duality} and has the same structure as the
original type $S$, but swaps the direction of communication.
Duality is an inductive metafunction on session types and is involutory: $\Dual{\DUAL{S}}= S$.

Type formation is defined in Figure~\ref{fig:ldgv-types}. As types do
not depend on linear values, all type environments involved in the
type formation judgment $\TEnv \vdash A : \Kind$ are unrestricted
(i.e., $\EnvForm\MultUn\TEnv$).
Equality types are unrestricted types constructed from a value of
label type and a concrete label
(\TirName{Equality-F}). This rule refers to the
upcoming typing judgment.
The $\TUnit$ type and the $\TEnd$ type both have kind $\KindUnSession$
(\TirName{Unit-F}, \TirName{End-F}).  The label type is an index type
for any non-empty, finite set of labels (\TirName{Lab-F}).
Label elimination \TirName{Lab-E'} for value $V$ constructs a witness
for the equality type $\Equation[\LabelSet] V \ell$ in the branch for
label $\ell$ to model dependent matching
\cite{DBLP:conf/popl/CasinghinoSW14}.
Formation of the type $\Pi_\Multm (x:A)B$ showcases conditional extension
(\TirName{Pi-F}). The type $A$ may be linear or unrestricted. In the
former case, the binding for $x$ must not be used in $B$, in the latter
case, it may. The conditional extension expresses
this desire precisely. The kind of the $\Pi$-type is determined by its annotation~$\Multm$.
The same rationale applies to the formation
rule \TirName{Sigma-F} of the type $\Sigma (x:A)B$, but we need to
check the kind of $A$ explicitly to make sure it matches the kind of
$B$.
Kind subsumption \TirName{Sub-Kind} enables products with components that
have different kinds.
%
%
Rules \TirName{Ssn-Out-F} and \TirName{Ssn-In-F}  manage dependency 
just like functions and products.

Figure~\ref{fig:ldgv-subtyping} describes type conversion and subtyping.
Type conversion $\TEnv \vdash \Convertible A B : \Kind$ specifies that
types $A$ and $B$ of kind $\Kind$ are equal up to substitutions that
can be justified by equations in~$\TEnv$, beta and eta conversion of
cases. Eta conversion enables commuting conversions that move common
(session) type prefixes in and out of $\CASEK$ types.  Conversion is closed under
reflexivity, symmetry, and transitivity.

As an example for type conversion in action, consider typing a function that
returns values of different primitive types \lstinline{Int} and \lstinline{String} depending on its input.
\begin{lstlisting}
lambda (b : {True, False}) case b of { True: 0, False: "foo" }
\end{lstlisting}
The \lstinline{True} branch typechecks with \lstinline{0 : Int}
whereas the \lstinline{False} branch typechecks with
\lstinline{"foo" : String}. Thanks to the \TirName{Conv-Beta}  rule,
we can expand the type of the \lstinline{True} branch to
\lstinline/0 : case True of { True: Int, False: String}/
and in the \lstinline{False} branch to
\lstinline/"foo" : case False of { True: Int, False: String}/.
According to the upcoming \CASEK{} elimination rule \TirName{Lab-E},
each branch for the \lstinline/case b/ adopts the equation of the
respective branch as in \lstinline{b = True} or \lstinline{b = False}.
Hence, the substitution rule \TirName{Conv-Subst} applies to obtain
the type \lstinline/case b of { True: Int, False: String}/ for both
branches and thus for the entire case expression. 

The rule \TirName{Conv-Eta} is needed for typechecking examples like
the code in Listing~\ref{lst:motivation:ldgv-compute-server}. After
the first \lstinline{recv} operation in line 3, the type of
\lstinline{c} is \lstinline/case l of { Neg: ?Int.NegType, Add: ?Int.AddType }/,
but the next operation is \lstinline{recv c}. The trick is to first
beta-expand the continuation types \lstinline/NegType/ and \lstinline/AddType/ to
\lstinline/CType = case l of { Neg: NegType, Add: Addtype }/ in both
branches using \TirName{Conv-Beta} as in the preceding example. The
resulting converted type of \lstinline{c} now reads
\lstinline/case l of { Neg: ?Int.CType, Add: ?Int.CType}/, which is
clearly convertible to
\lstinline/?Int.CType/
using \TirName{Conv-Eta}. Hence, \lstinline/recv c/ typechecks and
returns a channel end of type \lstinline/CType/!

Subtyping, also in Figure~\ref{fig:ldgv-subtyping}, is generated by
conversion (rule \TirName{Sub-Conv}), subsetting of label types (rule
\TirName{Sub-Lab}), and closed under transitivity
(\TirName{Sub-Trans}), subkinding (\TirName{Sub-Sub}), function and
product types (\TirName{Sub-Pi}, \TirName{Sub-Sigma}), as well as
session send and receive types (\TirName{Sub-Send},
\TirName{Sub-Recv}).

Subtyping of $\Pi$- and $\Sigma$-types extends the definitions of
\citet{AspinallCompagnoni2001}. The novel parts are the conditional
binding for the $(x:A)$ part as discussed for the formation rules and
the additional constraints on the multiplicities. \TirName{Sub-Send}
(\TirName{Sub-Recv}) is a simplified variant of \TirName{Sub-Pi}
(\TirName{Sub-Sigma}, respectively).

The rule \TirName{Sub-Case} deserves special attention.  Intended to
derive the premises for $B \Subtype B'$ in the rules \TirName{Sub-Pi},
\TirName{Sub-Sigma}, \TirName{Sub-Send}, and \TirName{Sub-Recv}, it
deals with the typical case that a function has type
$\Pi (x: L)\CASE x {\overline{\ell: {B_\ell}}^{\ell\in L}}$
and we need to determine whether this type is a subtype of
$\Pi (x:L')\CASE x {\overline{\ell: {B'_\ell}}^{\ell\in
    L'}}$. In this case, $L' \subseteq L$ is required and we adopt the
assumption $x : L'$ to prove the subtyping
judgment on the case types. For the corresponding product types,
however, $L \subseteq L'$ is required and the assumption for the case
expression reads $x : L$. Both cases are covered by
the assumption $x : L \cap L'$ which is the premise in the
\TirName{Sub-Case} rule.

In principle, subtyping is not required for \LDGV{} to work. However, it is included
for two reasons. First, it enables us to establish a tight correspondence
with the \LSST{} calculus which features subtyping (cf.\
Section~\ref{sec:embedding-gv-into}). Second, if we elided subtyping it
would be necessary to define a type equivalence relation, say, $\approx$ by a ruleset analogous
to the one in Figure~\ref{fig:ldgv-subtyping}, where all occurrences of
$\le$ would be replaced by $\approx$ and the comparisons between label
sets would change from $\subseteq$ to $=$ (and the same holds for
algorithmic subtyping vs.\ algorithmic type equivalence in
Section~\ref{sec:algorithmic-typing}). Hence, the system without
subtyping would not be simpler than the one presented.


\input{fig-ldgv-typing}

Figure~\ref{fig:ldgv-exps-envs} defines the expressions of \LDGV,
most of which are taken from \LSST.  In a dependent pair
$\DPair {x:A} M N$, the first component $M$ is bound to a variable $x$ which
may be used in the second component.  The expression $\NEW$ creates a new
channel of type $S$ and returns a pair of channel endpoints, one of
type $S$ and the other of type $\DUAL S$.  The expressions $\SEND M$
and $\RECV M$ have the same operational behavior as in \LSST.

Figure~\ref{fig:ldgv-typing} also contains the inference rules for
expression typing.
Most rules are standard, so we only highlight a few specific rules. 
Rule \TirName{Lab-E} is the expression-level
counterpart of the same-named rule at the typing level
(Figure~\ref{fig:ldgv-types}).
%
%
It characterizes a dependent case
elimination on a label type. In each branch it pushes an
equation, $\Equation V \ell$, on the typing environment, which can be
exploited in the type derivation for the branch.

Manipulation of $\Pi$-types is largely standard
(\TirName{Pi-I}). Well-formedness of the $\Pi$-type follows from the
agreement lemma\finalreport{}{~\ref{lem:agreement}}, as for all other type
constructors.  Elimination for $\Pi$-types is
limited to well-formed return types: if the function type 
depends on $x$, then the argument must be a value.

Manipulation of $\Sigma$-types is similarly restricted to dependency
on unrestricted values. \TirName{Sigma-I} introduces a pair, which
binds the first component to a variable that can be used in the second
component. It behaves like a dependent record. If the first component
$V$ is linear, then $x$ can be used in $N$, but it cannot influence
its type due to the well-formedness assumption of the $\Sigma$ type.

The rule \TirName{Sigma-G} is a refined elimination rule that enables
checking the second component of a product repeatedly with all possible
assumptions about the label in the first component.  It performs a
local eta expansion to increase the precision of typing.

As an example for a use of \TirName{Sigma-G} consider the code in
Listing~\ref{lst:motivation:sending-receiving-nodes}. If we naively
typecheck the product elimination \lstinline/let (tag, v) = n/ in the
definition of \lstinline/sendNode/, then
\lstinline/tag : {Empty, Node}/ and
\lstinline/v : case tag of {Empty: Unit, Node: Int}/.
Sending the \lstinline/tag/ in the next line updates the type of the
channel end to
\lstinline/c : case tag of {Empty: !Unit, Node: !Int}/. But now the
typecheck for the final \lstinline/send/ operation fails because the
value of the \lstinline/tag/ is unknown.

The \TirName{Sigma-G} rule prevents this issue. When eliminating a
product on a label type as in \lstinline/let (tag, v) = n/, the rule
checks the body of the \lstinline/let/ for each possible value of
\lstinline/tag/. The rule expresses this repeated check by a premise
that checks a case expression on the first component, \lstinline/tag/,
which replicates the body of the \lstinline/let/ in all branches
(i.e., the body is eta-expanded). As the
\TirName{Lab-E} rule for the case adopts a different equation
\lstinline/tag = .../ for each branch, all ramifications are
typechecked exhaustively. In the above example, the types for
\lstinline{v} and \lstinline{c} could both beta-reduce on the known
\lstinline/tag/ and thus unblock the typechecking for the
\lstinline/send/ operation.

The last block of rules in Figure~\ref{fig:ldgv-typing} governs the
typing of the session operations. The $\NEW$ expression returns a
linear pair of session endpoints where the types are duals of one
another.
The send operation turns a channel which is ready to
send into a single-use dependent function that returns the depleted
channel (\TirName{Ssn-Send-E}). The receive operation turns a channel which is ready to
receive into a linear dependent pair of the received value and the depleted channel (\TirName{Ssn-Recv-E}).

Process typing is standard
(cf.~\cite{GayVasconcelos2010-jfp,DBLP:journals/iandc/Vasconcelos12}%
\finalreport{}{ or Appendix~\ref{sec:compl-typing-rules}}).


\label{sec:dynamics}

\input{fig-ldgv-dynamics}

Figure~\ref{fig:ldgv-dynamics} defines call-by-value reduction in
\LDGV. Evaluation contexts are standard. Given all that, the dynamics
of \LDGV{} is pleasingly simple: it is roughly the dynamics of \LSST{}
with a few rules removed.
Expression reduction comprises a case rule for labels, beta-value
reduction, decomposition of products, and lifting over evaluation contexts.

Process reduction gets simplified to a subset of three base cases from
five in related work
\cite{igarashithiemannvasconceloswadler2017,GayVasconcelos2010-jfp}. From
Figure~\ref{fig:gv-reduction}, only one (out of three) communication
rule remains (rules {\sc Rl-Branch} and {\sc Rl-Close} are not part of
\LDGV). Rules \eqref{eq:rl-new-ldgv} and
\eqref{eq:rl-send-recv-ldgv} behave as before, but on
dependent pairs.


%% file: fig-ldgv-kinds-types-occurrences.tex
\begin{figure}[tp]
\begin{align*}
  \text{Values} &\!\!\!\!&
      V ,W \grmeq&
      z
      \grmor \ell
      \grmor \EUnit
      \grmor \lambda_\Multm (x: A).M
      \grmor \DPair {x:A} V W
      \grmor \SEND V
    \\
    \text{Types} &\!\!\!\!&
    A,B \grmeq&
    S 
    \grmor \TUnit
    \grmor L
    \grmor \Equation[L] VW 
    \grmor \CASE V {\overline{\ell : A_\ell}}
    \grmor \Pi_\Multm (x: A)B
    \grmor \Sigma (x: A)B
    \\
    \text{Session Types} &\!\!\!\!&
    S,R \grmeq&
    \TEnd
    \grmor \CASE V {\overline{\ell : S_\ell}}
    \grmor{!(x:A)S}
    \grmor {?(x:A)S}
  \end{align*}
Environment split \hfill\fbox{$\Decompose\TEnv{\TEnv_1}{\TEnv_2}$}
  \begin{mathpar}
    \inferrule{}{\Decompose\EmptyEnv\EmptyEnv\EmptyEnv}

  \inferrule{
    \Decompose\TEnv{\TEnv_1}{\TEnv_2} \\
    \TEnv_1 \EnvForm \MultUn A \\
    \TEnv_2 \EnvForm \MultUn A 
  }{\Decompose{(\TEnv, z:A)}{(\TEnv_1, z:A)}{(\TEnv_2, z:A)}}
  %

  \inferrule{
    \Decompose\TEnv{\TEnv_1}{\TEnv_2}
    \\ \TEnv_1 \EnvForm \MultLin A 
  }{
    \Decompose{(\TEnv, z:A)}{(\TEnv_1, z:A)}{\TEnv_2}
  }

  \inferrule{
    \Decompose\TEnv{\TEnv_1}{\TEnv_2}
    \\ \TEnv_2 \EnvForm \MultLin A 
  }{
    \Decompose{(\TEnv, z:A)}{\TEnv_1}{(\TEnv_2, z:A)}
  }
  \end{mathpar}
  Conditional extension, the unrestricted part of an env.  \hfill
  \fbox{$\DemoteExtend\TEnv x A = \SEnv$}
  \quad
  \fbox{$\UNR\TEnv = \SEnv$}
  \begin{mathpar}
  \inferrule{\UNR\TEnv \vdash A : \KindLin}{
    \DemoteExtend\TEnv x A =  \TEnv
  }

  \inferrule{\UNR\TEnv \vdash A : \KindUn}{
    \DemoteExtend\TEnv x A =  \TEnv, x:A
  }

  \inferrule{}{\UNR\EmptyEnv = \EmptyEnv}

  \inferrule{}{
    \UNR{(\TEnv, z : A)} = \DemoteExtend{(\UNR\TEnv)} z A}
\end{mathpar}
  Session type duality
  \hfill\fbox{$\DUAL S = R$}
\begin{mathpar}
  \inferrule{}{\Dual{{!(x:A)S}} = ?(x:A)\DUAL S}

  \inferrule{}{\Dual{{?(x:A)S}} = !(x:A)\DUAL S}

  \inferrule{}{\Dual{\CASE V {\overline{\ell : S_\ell}}} =
  \CASE V {\overline{\ell : \DUAL {S_\ell}}}}

  \inferrule{}{\DUAL\END = \END}
\end{mathpar}
  \caption{Values and types in \LDGV{}}
  \label{fig:ldgv-kinds-types-occurrences}
  \label{fig:ldgv-duality}
\end{figure}


%% file: fig-ldgv-types.tex
\begin{figure}[tp]
  Type formation
  \hfill\fbox{$\TEnv \vdash A : \Kind$}
  \begin{mathpar}
    \RuleEqualityF
    
    \RuleUnitF
    
    \RuleEndF
    
    \RuleLabF
    
    \RuleLabET
    
    \RulePiF
    
    \RuleSigmaF
    
    \RuleSsnOutF
    
    \RuleSsnInF
    
    \RuleSubKind
  \end{mathpar}

  \caption{Type formation in \LDGV{}}
  \label{fig:ldgv-types}
\end{figure}


%% file: fig-ldgv-subtyping.tex
\begin{figure*}[tp]
  Type 
  conversion
  \hfill
  \fbox{$\TEnv \vdash \Convertible A B : \Kind$}
  \begin{mathpar}
  %
    %
    \RuleConvSubst
    \\
  %
    \RuleConvBeta

    \RuleConvEta

    \RuleConvRefl

    \RuleConvSym

    \RuleConvTrans
    
  \end{mathpar}
  Subtyping
  \hfill\fbox{$\TEnv \vdash A \Subtype B : \Kind$}
  \begin{mathpar}
    \RuleSubConv
    
    \RuleSubLab
    
    \RuleSubTrans
    
    \RuleSubSub

    %
    \RuleSubPi
    
    \RuleSubSigma
    
    \RuleSubSend
    
    \RuleSubRecv
    
    \RuleSubCase
  \end{mathpar}
  
  \caption{Type conversion and subtyping in \LDGV{}}
  \label{fig:ldgv-subtyping}
\end{figure*}


%% file: fig-ldgv-typing.tex
\begin{figure*}[tp]
  \begin{flushleft}
    Expressions
  \end{flushleft}
  \begin{align*}
    M,N \grmeq&
    V
    \grmor \CASE V {\overline{\ell : N_\ell}}
    \grmor \APP M N
    \grmor \DPair {x:A} V N
    \grmor \ELet{\EPair[]xy}{M}{N}
    \\
    \grmor& \NEW
    \grmor \FORK M
    \grmor \SEND M
    \grmor \RECV M
  \end{align*}
  Expression formation
  \hfill\fbox{$\TEnv \vdash M : A$}
  \begin{mathpar}
    \RuleSubType
    
    \RuleName

    \RuleUnitI

    \RuleLabI

    \RuleLabE

    \RulePiI

    \RulePiE

    \RuleSigmaI

    \RuleSigmaE

    \RuleSigmaG

    \RuleFork
 
    \RuleSsnI

    \RuleSsnSendE

    \RuleSsnRecvE
  \end{mathpar}

  \caption{Expression formation in \LDGV{}}
  \label{fig:ldgv-typing}
  \label{fig:ldgv-exps-envs}
\end{figure*}


%% file: fig-ldgv-dynamics.tex
\begin{figure*}[tp]
  \begin{align*}
    \text{Evaluation contexts} &&
    \ECN{}, \ECN[F]{} \grmeq&
    \Box
    \grmor \CASE {\ECN{}} {\overline{\ell : N_\ell}}
    \grmor \ECN{}\,N
    \grmor V\,\ECN{}
    \\ &&
    \grmor& \DPair {x:A} V {\ECN{}}
    \grmor \ELet{\EPair[]xy}{\ECN{}}{N}
    \grmor \SEND {\ECN{}}
    \grmor \RECV {\ECN{}}
\end{align*}
Expression reduction
  \hfill\fbox{$M \ReducesTo N$}
  \begin{mathpar}
    \mprset{flushleft}
    \inferrule[Rl-Case]{
      \ell'\in L
    }{
      \CASE {\ell'} {\overline{\ell: M_\ell}^{\ell\in L}} \ReducesTo M_{\ell'}
    }

    \inferrule[Rl-Betav]{
    }{
      (\lambda_\Multm (x:A).M)\, V  \ReducesTo M[V/x]
    }

    %
    \inferrule[Rl-Prod-Elim]{
    }{
      \ELet{\EPair[]xy}{\DPair {x:A} V W}{M} \ReducesTo M[W/y][V/x]
    }

    \inferrule[Rl-ctx-exp]{M \ReducesTo N}{\Context M \ReducesTo \Context N}
  \end{mathpar}
  \renewcommand\columnwidth\textwidth 
  Process reduction
  \hfill\fbox{$P \ReducesTo Q$}
  \begin{align}
  \Proc{\Context{\NEW}}
  & \ReducesTo
    \NUC\Chanc\Chand~\Proc{\Context{\DPair[\MultLin]{x}\Chanc\Chand}}
    \Tag{Rl-New}\label{eq:rl-new-ldgv}
    \\
  \NUC\Chanc\Chand~\Proc{\Context{\SEND\Chanc\,V}} \PAR
  \Proc{\Context[F]{\RECV\Chand}}
  & \ReducesTo
  \NUC\Chanc\Chand\Proc{\Context{\Chanc}} \PAR
    \Proc{\Context[F]{\DPair[\MultLin]{x}V\Chand}}
    \Tag{Rl-Com}
    \label{eq:rl-send-recv-ldgv}
  \end{align}
  \begin{center}
    (Plus rule \eqref{eq:rl-fork}, the context and the structural
    congruence rules from Figure~\ref{fig:gv-reduction})
  \end{center}
  \caption{Reduction in \LDGV{}}
  \label{fig:ldgv-dynamics}
\end{figure*}


%% file: natural-numbers.tex
\section{Natural numbers and the recursor}
\label{sec:naturals}

\input{fig-nat}

The infrastructure developed in the previous sections is easily
amenable to extensions. In this section we report on the support for
natural numbers and a type recursor inspired by G\"odel's system~\textbf
T (cf.~\citet{Harper2016-book}). The required extensions are in
Figure~\ref{fig:nats}.

Newly introduced expressions comprise the natural number constructors
($\Zero$ and $\Succ(V)$) and a \emph{recursor}. A natural number~$n$
is encoded as $\overline n = \Succ(\dots\Succ(\Zero))$, where the
successor constructor is applied $n\ge0$ times to the zero
constructor. An expression of the form $\RECURSOR VMxyN$ represents
the $V$-iteration of the transformation $\lambda x.\lambda y.N$
starting from~$M$. The bound variable~$x$ represents the predecessor
and the bound variable~$y$ the result of $x$-iteration.
Its behaviour is clearly captured by the expression reduction rules in
the figure: the recursor evaluates to~$M$ when~$V$ is zero, and to~$N$
with the appropriate substitutions for~$x$ and~$y$, otherwise.


Types now incorporate type variables $\alpha,\beta$ of kind $\Un$, the type $\TNat$ of natural
numbers, and a \emph{type recursor}.
The type formation rules for type variables and natural numbers should
be self-explanatory. The rule for the type recursor,
$\TRec VA\TVar\Kind B$, requires~$V$ to be a natural number and 
types~$A$ and~$B$ to be of the same kind $\Kind$. The recursor
variable $\TVar$ may appear free in~$B$, thus
accounting for the recursive behaviour of the recursor.
For example, if~$n$ is a natural number, then type
$\TRec{\overline n}{(!\TInt)\END}{\alpha}{\MultLin}{(?\TInt)\alpha}$ intuitively
represents the type $(?\TInt)\dots(?\TInt) (!\TInt)\END$ composed
of~$n$ copies of $?\TInt$ and terminated by $(!\TInt)\END$.
As before, we introduce a type recursor for types and for session
types.
For
natural numbers, we need two new instances of the equality type, 
$\Equation[\TNat]{V}{\Zero}$ and 
$\Equation[\TNat]{V}{\Succ(W)}$, which fit in with the previously
defined rule \TirName{Equality-F} in Figure~\ref{fig:ldgv-types}.


The rules for type conversion should be easy to understand based on
those for expressions: a type $\TRec{V}{A}{\TVar}{\Kind}{B}$ may be
converted to~$A$ when~$V$ is zero and to~$B$ (with the appropriate
substitution), otherwise. A third rule (not shown) allows replacing an
expression-variable~$x$ by a natural number $V$ when an entry
$x = V$ can be found in the context (analogous to rule
\TirName{Conv-Subst} in Figure~\ref{fig:ldgv-subtyping}).


Now for duality and subtyping.
Defining the dual of the recursor is subtle and we adopt an approach inspired by
\citet{DBLP:conf/icfp/LindleyM16}'s treatment of general recursive
types. Type variables are adorned with a 
polarity $\Pol \in \{\PolPos, \PolNeg\}$. The polarity ``remembers''
whether the variable $\TVar_\Pol$ stands for the unrolled recursion
($\Pol=\PolPos$) or for its dual ($\Pol=\PolNeg$). 
To dualize the recursor $\TRec{V}{S}{\TVar}{\Kind}{R}$  we first apply
the usual dual to $S$ and $R$. When the transformation reaches a
variable $\TVar_p$ in $R$, it flips its polarity. Next, we
swap the polarities of all occurrences of the recursion variable in
$\DUAL R$. With this definition, duality is an involution on session
types. One caveat is that unrolling the recursion into a negative
variable (cf. \TirName{Conv-S}) will substitute the \emph{dual} of the
type for the variable.

The definition of subtyping for the recursor is fairly 
restrictive to avoid a coinductive definition. Rule \TirName{Sub-Rec}
essentially forces recursive types to synchronize and rule
\TirName{Sub-TVar} enforces an invariant treatment of the recursion
variables. This choice avoids additional complication with the
interplay of variance and the polarity of type variables,  while ensuring
the basic relation between subtyping and duality
($\TEnv \vdash \DUAL S \Subtype \DUAL R : \Kind$ when
$\TEnv \vdash R \Subtype S : \Kind$).
A more flexible approach would proceed coinductively;
we expect that the solution of \citet{GayHole2005} adaptable to our setting.


Finally, a word on the formation rules for the new expressions. Those
for natural numbers $\Zero$ and $\Succ(V)$ are standard. That for the
recursor $\RECURSOR VMxyN$ requires~$V$ to be a natural number and
expressions~$M$ and $N$ to have the same type~$A[V/x]$. The type for
$M$ is extracted from a context containing an extra entry stating that
$V$ is zero ($z: V= \Zero$). For~$N$ we add bindings for
the bound variables~$x$ and~$y$, as well as an extra entry stating
that $V$ is the successor of~$x$ ($z:V=\Succ(x)$).
Moreover, whereas $M$ is certainly used once, $N$ may be used
arbitrarily often. Hence, we must typecheck $N$ in an unrestricted
environment $\UNR\TEnv$!


%% file: fig-nat.tex
\begin{figure*}[p]
  \begin{align*}
    \text{Polarities} &&
                         \Pol \grmeq&\dots \grmor \PolPos \grmor \PolNeg
    \\
    \text{Values} &&
    V ,W \grmeq& \dots \grmor \Zero  \grmor \Succ(V)
    \\
    \text{Types} &&
    A,B \grmeq& \dots \grmor \TVar_\Pol \grmor \TRec{V}{A}{\TVar}{\Kind}{B} \grmor \TNat
    \\
    \text{Session Types} &&
    S,R \grmeq&  \dots \grmor \TVar_\Pol \grmor \TRec{V}{S}{\TVar}{\Kind}{R}
    \\
    \text{Expressions} &&
    M,N \grmeq& \dots \grmor \Zero  \grmor \Succ(M) \grmor \RECURSOR VMxyN
    \\
    \text{Typing environment} &&
    \TEnv, \SEnv \grmeq& \dots \grmor \TEnv, \TVar 
    \\
    \text{Evaluation contexts} &&
    \ECN{}, \ECN[F]{} \grmeq& \dots \grmor \RECURSOR \ECN MxyN
  \end{align*}
  Environment formation
  \hfill\fbox{$\EnvForm \Kind \TEnv$}
  \begin{mathpar}
    \inferrule{
      \EnvForm \Kind \TEnv
    }{
      \EnvForm \Kind {(\TEnv,\TVar)}
    }
    %
\end{mathpar}
  Type formation
  \hfill\fbox{$\TEnv \vdash A : \Kind$}
  \begin{mathpar}
    \RuleNatF

    \RuleTVarF 

    \RuleRecF
  \end{mathpar}
  Type conversion
  \hfill \fbox{$\TEnv \vdash \Convertible A B : \Kind$}
  \begin{mathpar}
    \RuleConvZ

    \RuleConvS
  \end{mathpar}
  Subtyping
  \hfill\fbox{$\TEnv \vdash A \Subtype B : \Kind$}
  \begin{mathpar}
    \RuleSubRec

    \RuleSubTVar
    %
  \end{mathpar}
  Session type duality
  \hfill\fbox{$\DUAL S = R$}
  \begin{mathpar}
    \RuleDualRec

    \RuleDualPosVar

    \RuleDualNegVar
  \end{mathpar}
  Expression formation
  \hfill\fbox{$\TEnv \vdash M : A$}
  \begin{mathpar}
    \RuleZI 

    \RuleSI  

    \RuleRecI 
  \end{mathpar}
  Expression reduction
  \hfill\fbox{$M \ReducesTo N$}
  \begin{mathpar}
    \inferrule[RL-Z]{}{\RECURSOR \Zero MxyN \ReducesTo M}

    \inferrule[RL-S]{}{\RECURSOR {\Succ(\!V\!)} MxyN \ReducesTo
    N[V/x][\RECURSOR VMxyN/y]}
  \end{mathpar}
  \caption{Extensions for natural numbers and recursor}
  \label{fig:nats}
\end{figure*}


%% file: ldgv-meta-light.tex
\section{Metatheory}
\label{sec:metatheory-light}

The main metatheoretical results for {\LDGV} are subject reduction for
expressions, typing preservation for processes, and absence of
run-time errors. All proofs and auxiliary results may be found in
\finalconcrete{Appendix~\ref{sec:metatheory}}.

\begin{restatable}[Typing preservation for expressions]{theorem}{LDGVSubjectReductionExpressions}
  \label{thm:ldgv-subject-reduction-expressions}
  If $\TEnv \vdash {M} : A$ and $M \ReducesTo N$, then
  $\TEnv \vdash {N} : A$.
\end{restatable}

Its proof requires the usual substitution and weakening lemmas along
with lemmas about environment splitting.

\begin{restatable}[Typing preservation for processes]{theorem}{TypingPreservationProcesses}
  \label{thm:ldgv-typing-preservation-processes}
  If $\TEnv \vdash P$ and $P \ReducesTo Q$, then $\TEnv \vdash Q$.
\end{restatable}

Its proof relies on
Theorem~\ref{thm:ldgv-subject-reduction-expressions} and the
adaptation of two results about the manipulation of subderivations by
\citet{GayVasconcelos2010-jfp}.

An absence of runtime errors result for \LDGV{} is based on
\citet{GayVasconcelos2010-jfp,HondaVasconcelosKubo1998,DBLP:journals/iandc/Vasconcelos12,igarashithiemannvasconceloswadler2017}. We
start by defining what it means for a process to be an \emph{error}:
a) an attempt to match against a non-value label or a label that is
not in the expected set (rule \TirName{Rl-Rec},
Figure~\ref{fig:ldgv-dynamics}), eliminate a function, a \FIXK, a pair
or a natural number against the wrong value (rules \TirName{Rl-Betav},
\TirName{Rl-RecBetav}, and \TirName{Rl-Prod-Elim} in
Figure~\ref{fig:ldgv-dynamics}; rules \TirName{Rl-Z} and \TirName{RL-S}
in Figure~\ref{fig:nats}), and b) two processes trying to access the
same channel endpoint, or accessing the different endpoints both for
reading or for writing (rule \TirName{Rl-Com},
Figure~\ref{fig:ldgv-dynamics}).

\begin{restatable}[Absence of run-time errors]{theorem}{AbsenceRuntimeErrors}
  \label{thm:ldgv-absence-runtime-errors}
  If $\vdash P$, then $P$ is not an error.
\end{restatable}


%% file: algorithmic-typing.tex
\section{Algorithmic Type Checking}
\label{sec:algorithmic-typing}

Section~\ref{sec:proposed-calculus} presents a declarative type system
for \LDGV{}. In this section, we prove that type checking is
decidable. Our algorithm for type checking is based on bidirectional
typing
\cite{PierceTurner2000-toplas,DBLP:conf/icfp/DunfieldK13,DBLP:conf/ppdp/FerreiraP14}
and comprises several syntax-directed judgments collected in the table
below.
\begin{center}
  \begin{tabular}{ll}
    {$\JAlgConv\TEnv V W L$}
    & Given $\TEnv$ and $V$, compute a convertible value $W$
    \\
    {$\JAlgTypeUnfold\TEnv A B$}
    & Given $\TEnv$ and $A$, compute a type $B$ convertible to $A$
      which is not a case
    \\
    {$\JAlgSubtypeSynth\TEnv A B \Kind$}
    & Given $\TEnv$, $A$, and $B$, check that $A$ is a subtype of $B$
      and synthesize its kind $\Kind$ 
    \\
    {$\JAlgSubtypeCheck\TEnv A B \Kind$}
    & Given $\TEnv$, $A$, $B$, and $\Kind$, check that $A$ is a subtype
      of $B$ at kind $\Kind$
    \\
    {$\JAlgKindSynth\TEnv A \Kind$}
    & Given $\TEnv$ and type $A$, synthesize its kind $\Kind$
    \\
    {$\JAlgKindCheck\TEnv A \Kind$}
    & Given $\TEnv$, $A$, and $\Kind$, check that $A$ has kind $\Kind$
    \\
    {$\JAlgTypeSynth\TEnv M A \TOut$}
    & Given $\TEnv$ and expression $M$, synthesize its type $A$ and the
      environment after $\TOut$
    \\
    {$\JAlgTypeCheck\TEnv M A \TOut$}
    & Given $\TEnv$, $M$, and type $A$, check that $M$ has type $A$
      and synthesize $\TOut$
  \end{tabular}
\end{center}


\input{fig-ldgv-algorithmic-unfolding}

The first building block is value conversion and unfolding, two
partial functions presented in
Figure~\ref{fig:ldgv-algorithmic-unfolding}.
Value conversion $\JAlgConv\TEnv V W L$ outputs $W$ if $V$ can be
converted to some $W$ given the assumptions $\TEnv$.  There are two
rules. \TirName{AC-Refl} applies if $V$ is already a
label. \TirName{AC-Assoc} locates an assumption $\Equation[L]xW$ in
$\TEnv$ and returns $W$. In our system, all equations have the form
$\Equation[L]xW$ so that no further rules are needed.

The unfolding judgment $\JAlgTypeUnfold\TEnv AB$ is needed in the
elimination rules for expression typing. Unfolding exposes the
top-level type constructor by commuting case types. The exposed type
$B$ is convertible to $A$. If $A$ is not a case type, then no
unfolding happens (\TirName{A-Unfold-Refl}). If the left type is a
case on a known value $V$, then recurse on the selected branch
(\TirName{A-Unfold-Case}). Otherwise, we try to expose the same
top-level type constructor in all branches of the case and commute it
on top of the case (\TirName{A-Unfold-Case2}). Rule
\TirName{A-Unfold-Case1} deals with the special case where the
branches have label type $L'$. Unfolding of a case
fails if no common top-level constructor exists.


\input{fig-ldgv-algorithmic-subtyping-excerpt}

The rules for the algorithmic subtyping judgment
$\JAlgSubtypeSynth\TEnv A B \Kind$ mostly follow the declarative
subtyping rules in Figure~\ref{fig:ldgv-subtyping}. If $A$ is a
subtype of $B$ given the assumptions $\TEnv$, then the judgment
produces the minimal kind $\Kind$ for $B$. The full set of rules is
shown in \finalconcrete{the appendix
  (Figure~\ref{fig:ldgv-algorithmic-subtyping})}. Here, we only
discuss the rules \TirName{AS-Case-Left1} and \TirName{AS-Case-Left2}
(in Figure~\ref{fig:ldgv-algorithmic-subtyping-excerpt}) that deal
with case types when they occur on the left (the rules for case on the
right mirror the left rules).
%
%
%
Rule \TirName{AS-Case-Left1} invokes algorithmic conversion to find
out if $V$ is convertible to a label $\ell$ under $\TEnv$. In that
case, the left hand side (case-) beta reduces to $A_\ell$ so that we
synthesize $A_\ell \Subtype B$ recursively.  If the attempt to convert
the case header to a label fails, then the header must be a variable
and its type must unfold to a label type
(\TirName{AS-Case-Left2}). Hence, we recursively check that each case
branch $A_\ell$ is a subtype of the right hand type $B$ under the
assumption that $x=\ell$.


The algorithmic kinding rules for judgment
$\JAlgKindSynth\TEnv{ A} \Kind$ are straightforward as
the type language is a simply-kinded first-order language with
subkinding. They may be found in \finalconcrete{Figure~\ref{fig:ldgv-algorithmic-kind-inference}}.


\input{fig-ldgv-typing-alg}

The rules for synthesizing a type
(Figure~\ref{fig:ldgv-algorithmic-typing}) define the judgment
{$\JAlgTypeSynth\TEnv M A \TOut$}. From environment $\TEnv$ and
expression $M$, the judgment computes $M$'s least type and the
remaining type environment $\TOut$. The difference between $\TEnv$ and
$\TOut$ indicates which linear resources are used by $M$: if the
binding $x: A \in \TEnv$ with $A:\KindLin$ is used in $M$, then
$\TOut$ does not contain a binding for $x$. No other changes are
possible.
Most of the rules are  adaptations of the declarative
typing rules from Figure~\ref{fig:ldgv-typing} to the bidirectional
setting. We explain the most relevant.

In rule \TirName{A-Pi-I } we synthesize the kind $\Multn$ of the
argument type $A$. After synthesizing the type of the body with the
environment $\TEnv_1, x: A$, the returned environment
must have the form $\DemoteExtend{\TEnv_2} x A$. Thus, we expect
that $x$ is used in the body if $A$ is linear. Moreover, if the
function's multiplicity $\Multm$ is unrestricted, then no resources in
$\TEnv_1$ must be used. This constraint is imposed by checking
$\TEnv_1=\TEnv_2$.

Typing an application $M\,N$ (rule \TirName{A-Pi-E}) first synthesizes
the type of $M$. We cannot expect the resulting type $C$ to be a $\Pi$
type; it may just as well be a case type! Unfolding exposes the
top-level non-case type constructor, which we can check to be a $\Pi$
type and then extract domain and range types $A$ and $B$. Next, we
check that $N$'s type is a subtype of $A$, and finally that $B[N/x]$
is well-formed.

Rule \TirName{A-Lab-E1 }applies if the conversion judgment $
\JAlgConv{\UNR\TEnv_1} V \ell{ L'}$ figures out that $V$ is
convertible to label $\ell$. In this case, we only synthesize the type
for the branch $N_\ell$ and return that type.

In rule \TirName{A-Lab-E2}, if the variable $x$ is not convertible to
a label, then we must synthesize the types for all branches. For each
branch, we adopt the equation $\Equation[L]x\ell$ and remove it from
the returned environment. As all branches must use resources in the
same way, the rule checks that the outgoing environments $\TOut_\ell$
are equal for all branches.

Rule \TirName{A-Sigma-G}, together with its counterpart
\TirName{Sigma-G} in Figure~\ref{fig:ldgv-typing}, is a significant
innovation of our system. It governs the elimination of a sigma type
where the first component of the pair is a variable of label type
$L$. Instead of type checking the body of the eliminating let once,
the rule checks it multiple times, once for each $\ell\in L$. This
eta-expansion of the label type enables us to accurately check this
construct and enable examples such as those in
Section~\ref{sec:furth-opport}.


We now address the metatheory for algorithmic type checking.
As usual, soundness results rely on strengthening and completeness on
weakening, two results that we study below.
Below we write $\DecomposeOp{\TEnv_1}{\TEnv_2}$ to denote the type
environment~$\TEnv$ such that $\Decompose\TEnv{\TEnv_1}{\TEnv_2} $,
when the environment splitting operation is defined. 

\begin{restatable}[Algorithmic Weakening]{lemma}{AlgorithmicWeakening}
  \label{lem:AlgorithmicWeakening}
  \label{lemma:weakening}
  \
  \begin{enumerate}
  \item\label{item:JAlgConv} If $\JAlgConv{\TEnv_1} V W L$, then
    $\JAlgConv{(\DecomposeOp{\TEnv_1}{\TEnv_2})} V W L$.
  \item\label{item: JAlgTypeUnfold} If $\JAlgTypeUnfold{\TEnv_1} A B$, then
    $\JAlgTypeUnfold{(\DecomposeOp{\TEnv_1}{\TEnv_2})} A B$.
  \item\label{JAlgSubtypeSynth} If
    $\JAlgSubtypeSynth{\TEnv_1} A B \Kind$, then
    $\JAlgSubtypeSynth {(\DecomposeOp{\TEnv_1}{\TEnv_2})} A B
    {\Kind}$. 
  \item\label{item: JAlgSubtypeCheck} If
    $\JAlgSubtypeCheck{\TEnv_1} A B \Kind$, then
    $\JAlgSubtypeCheck{(\DecomposeOp{\TEnv_1}{\TEnv_2})} A B {\Kind}$.
  \item\label{item:JAlgKindSynth} If $\JAlgKindSynth{\TEnv_1} A \Kind$, then
    $\JAlgKindSynth{(\DecomposeOp{\TEnv_1}{\TEnv_2})} A \Kind$.
  \item\label{item:JAlgKindCheck} If $\JAlgKindCheck{\TEnv_1} A \Kind$, then
    $\JAlgKindCheck{(\DecomposeOp{\TEnv_1}{\TEnv_2})} A \Kind$.
  \item\label{item:JAlgTypeSynth} If
    $\JAlgTypeSynth{\TEnv_1} M A {\TEnv_2}$, then
    $\JAlgTypeSynth{(\DecomposeOp{\TEnv_1}{\TEnv_3})} M A
    {(\DecomposeOp{\TEnv_2}{\TEnv_3})}$.
  \item\label{item:JAlgTypeCheck} If
    $\JAlgTypeCheck{\TEnv_1} M A {\TEnv_2}$, then
    $\JAlgTypeCheck{(\DecomposeOp{\TEnv_2}{\TEnv_3})} M A
    {(\DecomposeOp{\TEnv_2}{\TEnv_3})}$.
  \end{enumerate}
\end{restatable}

\begin{restatable}[Algorithmic Linear Strengthening]{lemma}{AlgorithmicLinearStrengthening}
  \label{lem:AlgorithmicLinearStrengthening}
  Suppose that $\UNR{\TEnv_1} \vdash A:\KindLin$.
  \begin{enumerate}
  \item If
    $\JAlgTypeSynth{\TEnv_1, x : A} M {\TEnv_2, x :
      A}$, then $\JAlgTypeSynth{\TEnv_1} M B {\TEnv_2}$.
  \item If
    $\JAlgTypeCheck{\TEnv_1, x : A} M {\TEnv_2, x :
      A}$, then $\JAlgTypeCheck{\TEnv_1} M B {\TEnv_2}$.
  \end{enumerate}
\end{restatable}

%

The rest of this section is dedicated to the soundness and
completeness results for the various relations in algorithmic type
checking. Proofs are by mutual rule induction, even if we present the
results separately, for ease of understanding. Proofs can be found in
\finalconcrete{Appendix~\ref{sec:metatheory}}.

\begin{restatable}[Soundness of Unfolding]{lemma}{AlgorithmicUnfoldingSoundness}
  \label{lem: AlgorithmicUnfoldingSoundness}
  Suppose that
  $\TEnv \vdash A : \Kind$ and $\JAlgTypeUnfold\TEnv A B$.
  Then $B$ is not a case and $\TEnv \vdash \Convertible A B : \Kind$.
\end{restatable}

\begin{restatable}[Completeness of Unfolding]{lemma}{AlgorithmicUnfoldingCompleteness}
  \label{lem: AlgorithmicUnfoldingCompleteness}
  Suppose that   $\TEnv \vdash A : \Kind$  and there exists some $
  \ECN[P] \in \{ L, \Pi_\Multm 
  (y:A)\Hole, \Sigma_\Multm (y:A)\Hole, {!(y:A)}\Hole, {?(y:A)}\Hole
  \} $ such that  $\TEnv \vdash \Convertible A {\Context[P]B} : \Kind$.
  Then $\JAlgTypeUnfold\TEnv A {\Context[P]{B'}}$ where $\TEnv \vdash
    \Convertible B {B'} : \Kind$.
\end{restatable}

\begin{restatable}[Algorithmic Subtyping Soundness]{lemma}{AlgorithmicSubtypingSoundness}
  \label{lem:AlgorithmicSubtypingSoundness}~
  \begin{enumerate}\item
    If $\JAlgSubtypeSynth\TEnv B A \Kind$, then
    $\TEnv \vdash B \Subtype A : \Kind$.
  \item
    If $\JAlgSubtypeCheck\TEnv B A \Kind$, then
    $\TEnv \vdash B \Subtype A : \Kind$.
  \end{enumerate}
\end{restatable}

\begin{restatable}[Algorithmic Subtyping Completeness]{lemma}{AlgorithmicSubtypingCompleteness}
  \label{lem:AlgorithmicSubtypingCompleteness}
   Let $\TEnv \vdash B \Subtype A : \Kind$. Then,
  \begin{enumerate}
  \item $\JAlgSubtypeSynth\TEnv A B \Kind'$ with $\Kind\Subkind\Kind'$.
  \item $\JAlgSubtypeCheck\TEnv A B \Kind$.
  \end{enumerate}
\end{restatable}

\begin{restatable}[Algorithmic Kinding Soundness]{lemma}{AlgorithmicKindingSoundness}
  \label{lem:AlgorithmicKindingSoundness}
  \
  \begin{enumerate}
  \item If $\JAlgKindSynth\TEnv M \Kind$, then
    $\TEnv \vdash M : \Kind$.
  \item If $\JAlgKindCheck \TEnv M \Kind$, then
    $\TEnv \vdash M : \Kind$.
  \end{enumerate}
\end{restatable}

\begin{restatable}[Algorithmic Kinding Completeness]{lemma}{AlgorithmicKindingCompleteness}
  \label{lem:AlgorithmicKindingCompleteness}
  If $\TEnv \vdash A : \Kind$, then
  \begin{enumerate}
  \item 
    $\JAlgKindSynth\TEnv A{ \Kind'}$ with $\Kind' \Subkind \Kind$ and
  \item 
    $\JAlgKindCheck\TEnv A \Kind$.
  \end{enumerate}
\end{restatable}

\begin{restatable}[Algorithmic soundness]{theorem}{AlgorithmicSoundness}
  \label{thm:algorithmic-soundness}
  Suppose that $\AEnvForm \MultUn {\TOut}$.
  \begin{enumerate}
  \item If $\JAlgTypeSynth {\TEnv} M A {\TOut}$, then $\TEnv \vdash M : A$.
  \item If $\JAlgTypeCheck \TEnv M A \TOut$, then $\TEnv \vdash M : A$.
  \end{enumerate}
\end{restatable}

\begin{restatable}[Algorithmic Completeness]{theorem}{AlgorithmicCompleteness}
  \label{thm:AlgorithmicCompleteness}
  If $\TEnv \vdash M : A$, then
  \begin{enumerate}
  \item 
    $\JAlgTypeSynth{\TEnv} M {B}{ \UNR\TEnv}$ with $\UNR\TEnv \vdash B \Subtype A : \Kind$.
  \item
    $\JAlgTypeCheck{\TEnv} M A{ \UNR\TEnv}$.
  \end{enumerate}
\end{restatable}

The development in this section does not cover the extension to
natural numbers from Section~\ref{sec:naturals}. However, we present
the necessary rules in \finalconcrete{the Appendix~\ref{sec:compl-typing-rules}} and
we believe the technical results extend straightforwardly.


%% file: fig-ldgv-algorithmic-unfolding.tex
\begin{figure*}[tp]
  Algorithmic value conversion
  \hfill
  \fbox{$\JAlgConv\TEnv V \ell L$}
  \begin{mathpar}
    \RuleACRefl

    \RuleACAss
  \end{mathpar}
  Algorithmic value unfolding
  \hfill
  \fbox{$\JAlgTypeUnfold\TEnv A B$}
  \begin{mathpar}
    \RuleAUnfoldType

    \RuleAUnfoldCase

    \RuleAUnfoldCaseA

    \RuleAUnfoldCaseB
  \end{mathpar}
  \caption{Algorithmic value conversion and unfolding}
  \label{fig:ldgv-algorithmic-unfolding}
\end{figure*}


%% file: fig-ldgv-algorithmic-subtyping-excerpt.tex
\begin{figure*}[tp]
  Algorithmic subtyping (synthesis)
  \hfill
  \fbox{$\JAlgSubtypeSynth\TEnv A B \Kind$}
  \begin{mathpar}
    \RuleASCaseLeftA
    
    \RuleASCaseLeftB
  \end{mathpar}
  \caption{Algorithmic subtyping in \LDGV{} (excerpt)}
  \label{fig:ldgv-algorithmic-subtyping-excerpt}
\end{figure*}


%% file: fig-ldgv-typing-alg.tex
\begin{figure*}
  Algorithmic type checking for expressions (synthesize)
  \hfill\fbox{$\JAlgTypeSynth\TEnv M A \TOut$}
  \begin{mathpar}
    \RuleAName

    \RuleAUnitI

    \RuleALabI

    \RuleALabEA

    \RuleALabEB

    \RuleAPiI

    \RuleAPiE

    \RuleASigmaI

    \RuleASigmaE

    \RuleASigmaG

    \RuleAFork
 
    \RuleASsnI

    \RuleASsnSendE

    \RuleASsnRecvE
  \end{mathpar}

  Algorithmic type checking for expressions (check against)
  \hfill\fbox{$\JAlgTypeCheck\TEnv M A \TOut$}
  \begin{mathpar}
    \RuleASubType
  \end{mathpar}

  \caption{Algorithmic typing for \LDGV{}}
  \label{fig:ldgv-algorithmic-typing}
\end{figure*}


%% file: embedding.tex
\section{Embedding {\LSST} into {\LDGV}}
\label{sec:embedding-gv-into}

\input{fig-gv-ldgv-translation}

The translation in Figure~\ref{fig:gv-ldgv-translation} maps \LSST's
types, environments, and typing derivations to \LDGV. It extends
homomorphically over all types, expressions, and processes that are
not mentioned explicitly. The translation of typing environments
annotates each binding with the multiplicity derived from the type.
As expected, internal (external) choice maps to sending (receiving) a
label followed by a case distinction on that label. Actively
(passively) ending a connection maps to sending (receiving) a
distinguished $\EOFLabel$ token and dropping the channel.

The translation is a conservative embedding as it preserves subtyping and
typing. We establish a simulation and a co-simulation between the original \LSST{} expression
and its image in \LDGV{}.  In the
simulation, each step gives rise to one or more steps in the image of the translation.
In co-simulation, one step in the image may yield an expression that is still
related to the same preimage.

%
%





\begin{restatable}[Typing Preservation]{theorem}{SimulationTypingPreservation}~\\[-\baselineskip]
  \begin{enumerate}
  \item If $\TEnv \vdash_\LSST M : A$, then
    $\DGVEnvForm \MultLin {\Embed{\TEnv}}$ and
    $\Embed\TEnv \vdash_\DGV \Embed{\TEnv \vdash_\LSST M : A} : \Embed
    A$.
  \item If $\TEnv \vdash_\LSST P$, then $\Embed\TEnv \vdash_\DGV
    \Embed{\TEnv \vdash_\LSST P}$.
  \end{enumerate}
\end{restatable}

\begin{theorem}[Simulation]~\\[-\baselineskip]
  \begin{enumerate}
  \item If $\TEnv \vdash_\LSST M : A$ and $M \ReducesTo_\LSST N$, then
    $\Embed{\TEnv\vdash_\LSST M:A} \ReducesTo^+_\DGV
    \Embed{\TEnv\vdash_\LSST N:A}$.
  \item If $\TEnv \vdash_\LSST P$ and $P \ReducesTo_\LSST Q$, then
    $\Embed{\TEnv \vdash_\LSST P} \ReducesTo^+_\DGV \Embed{\TEnv
      \vdash_\LSST Q}$.
  \end{enumerate}
\end{theorem}

\begin{restatable}[Co-Simulation]{theorem}{CoSimulation}~\\[-\baselineskip]
  \begin{enumerate}
  \item If $\TEnv \vdash_\LSST M : A$ and
    $\Embed{\TEnv\vdash_\LSST M:A} \ReducesTo_\DGV N$, then
    $M \ReducesTo_\LSST M'$ and
    $N \ReducesTo^*_\DGV\Embed{\TEnv \vdash_\LSST M':A}$.
  \item If $\TEnv \vdash_\LSST P$ and
    $\Embed{\TEnv\vdash_\LSST P} \ReducesTo_\DGV Q$, then
    $P \ReducesTo_\LSST P'$ and
    $Q \ReducesTo^*_\DGV \Embed{\TEnv \vdash_\LSST P'}$.
  \end{enumerate}
\end{restatable}


%% file: fig-gv-ldgv-translation.tex
\begin{figure}[tp]
  Type translation
  \hfill\fbox{$\Embed{A} = B$}
\begin{align*}
  \Embed{{!A}.S} =&\;{!(x : \Embed{A})} \Embed{S}
  &&&
  \Embed{{?A}.S} =&\;{?(x : \Embed{A})} \Embed{S}
  \\
  \Embed{\oplus \{ \overline{\ell: S_\ell}\} } =&\;
    {!(x : L)} \CASE x {\overline{\ell: \Embed{S_\ell}}}
  &&&
  \Embed{\&\{\overline{\ell: S_\ell}\}} =&\;
    {?(x : L)} \CASE x {\overline{\ell: \Embed{S_\ell}}}
  \\
  \Embed{\ENDS} =&\;{!(x: \Setof\EOFLabel)}\TEnd
  &&&
  \Embed{\ENDR} =&\; {?(x: \Setof\EOFLabel)}\TEnd
\end{align*}
  Environment translation
  \hfill\fbox{$\Embed{\TEnv} = \SEnv$}
  \begin{mathpar}
    \inferrule{}{\Embed\EmptyEnv = \EmptyEnv}

    \inferrule{}{
      \Embed{\TEnv,x: A} = \Embed{\TEnv}, x: \Embed{A}
    }
  \end{mathpar}
  Expression translation
  \hfill\fbox{$\Embed{\TEnv \vdash_\LSST M : A} = N$}
\begin{align*}
  \Embed{\TEnv \vdash_\LSST \SEND M : A \to_\KindLin S}
  &=
  \SEND{\Embed{\TEnv \vdash_\LSST M : {!A.S}}}
  \\
  \Embed{\TEnv \vdash_\LSST \RECV M : A \times S}
  &=
  \RECV{\Embed{\TEnv \vdash_\LSST M : {?A.S}}}
  \\
  \Embed{\TEnv \vdash_\LSST \SELECT\ell' : \oplus \{\overline{\ell:
  S_\ell}^{\ell\in L}\}}
  &=
  \lambda_\KindLin(x: {!(y: L)} \CASE y {\overline{\ell : \Embed{S_\ell}}}). \APP{\SEND x} \ell'
  \\
  \Embed{\DecomposeOp{\TEnv_1}{\TEnv_2} \vdash_\LSST \RCASE M {\overline{\ell : y.N_\ell}} : A}
  &=
  \ELet{\EPair[]{x}{y}} {\RECV{\Embed{\TEnv_1\vdash_\LSST
    M:\&\{\overline{\ell: S_\ell}\}}}}{}
  \\
  \CASEK\,x\,\syntax{of}\,&\{\overline{\ell: \Embed{\TEnv_2,y:S_\ell\vdash_\LSST N_\ell : A}}\}
    \quad \text{where }\TEnv_1 \vdash M : \&\{\overline{\ell: S_\ell}\}
  \\
  \Embed{\TEnv \vdash_\LSST \CLOSE M : \TUnit}
  &=
  \SEND \Embed{\TEnv\vdash_\LSST M : \ENDS}\, \EOFLabel
  \\
  \Embed{\TEnv \vdash_\LSST \WAIT M : \TUnit}
  &=
  \ELet{\EPair[]{x}{y}}{\RECV{\Embed{\TEnv\vdash_\LSST M : \ENDR}}} \EUnit
\end{align*}
 
\caption{Translation from {\LSST} to {\LDGV}}
\label{fig:gv-ldgv-translation}
\end{figure}


%% file: implementation.tex
\section{Implementation}
\label{sec:implementation}

We implemented a frontend consisting of a parser and a type checker
for the \LDGV{} calculus, which is available in a GitHub
repository\footnote{Available at
  \href{https://github.com/proglang/ldgv}{https://github.com/proglang/ldgv}}.
The parser implements an OCaml-inspired 
syntax which deviates slightly from the
Haskell-inspired syntax used in Section~\ref{sec:motivation}.

The type checker implements exactly the algorithmic rules from
Section~\ref{sec:algorithmic-typing} including subtyping as well as
additional algorithmic unfolding, subtyping, and synthesis rules
dealing with natural numbers and their recursor. The type checker
supports a coinductive reading of the typing and subtyping rules
so that types and session types can be equirecursive. The
implementation requires caching of the weakened judgments modulo
alpha conversion. This complication arises because types may contain
free variables of label type. A weakened judgment contains just the
bindings for these free variables; comparing modulo alpha conversion
means that the names of the free variables do not matter: $x:L \vdash
\CASE x {\dots}$ is equal to $y:L \vdash \CASE y {\dots}$ (if the
$\dots$ match). Caching modulo alpha conversion is needed to make the
type checker terminate.


%% file: related.tex
\section{Related Work}
\label{sec:related-work}

\paragraph{Linear and Dependent Types}


\citet{DBLP:conf/lics/CervesatoP96} developed the first logical
framework supporting linear type theory and dependent types.
\citet{DBLP:journals/scp/ShiX13} propose using linear types on top of
ATS, their dependently typed language for developing provably correct
code.

F$^*$ \cite{DBLP:journals/jfp/SwamyCFSBY13} is a language that
includes linear types and value dependent types. The authors use
affine environments to control the use of linear values and
distinguish between value application and standard application to
properly deal with dependency. F$^*$ has further developed into a
verification system with full-fledged dependent types
\cite{DBLP:journals/pacmpl/AhmanFHMRS18}.

Trellys \cite{DBLP:conf/popl/CasinghinoSW14} combines a general computation
language with a specification language via dependent types. While
Trellys has no support for linear types, it has been inspiring in
finding a replacement for value dependency and in its treatment of
equations.

Idris \cite{DBLP:journals/jfp/Brady13} is a dependently typed language
with uniqueness types. While linear types avoid duplication
and dropping of values, a value with unique type is referenced at most
once at run time. \citet{DBLP:journals/aghcs/Brady17} shows how to use
this feature combination to develop concurrent systems.

\citet{DBLP:journals/corr/abs-1104-0193}
introduce a lambda calculus
with linear dependent types and full higher-order recursion. It relies
on a decoration of PCF with first-order index expressions. Under certain
assumptions, their type system is complete, i.e., all operational
behavior can be captured by typing.
\citet{DBLP:conf/ppdp/LagoP12}
also consider a sound and complete
linear dependent type system. Their emphasis is on complexity analysis
for higher-order functional programs.

\citet{DBLP:conf/popl/KrishnaswamiPB15} propose a full-spectrum
language that integrates linear and dependent types. It
is based on the observation that intuitionistic linear logic can be
modeled with an adjunction. The resulting syntactic theory
consists of an intuitionistic and a linear lambda calculus combined
via two modal operators corresponding to the adjunction.


Our work stays in the tradition that keeps linear and unrestricted
resources apart. Computations and processes are allowed to
depend on unrestricted index values, but dependencies on linear
resources are ruled out. Unlike the cited work, our calculus supports
dependent subtyping \cite{AspinallCompagnoni2001}.

McBride's and Atkey's works combine linear and
dependent types in Quantitative Type Theory (QTT)
\cite{DBLP:conf/birthday/McBride16,DBLP:conf/lics/Atkey18}.
In QTT types may depend on linear resources, whereas
types in our system can only depend on unrestricted values.
%
%

Linear Haskell \cite{DBLP:journals/pacmpl/BernardyBNJS18} 
is a proposal to integrate linear types with stock functional
programming.  It does not have dependent types and it manages bindings
using a semiring.

\paragraph{Session Types}

\citet{CairesPfenning2010} developed logical foundations for session
types building on intuitionistic linear logic. Their approach enables
viewing $\pi$-calculus reductions as proof transformations in the
logic. \citet{Wadler2012} proposed a foundation based on classical
linear logic.

Dependent session types have been proposed first by
\citet{ToninhoCairesPfenning2011} for the $\pi$-calculus with value
passing. The calculus is aimed at specification and verification, and
features a rich logic structure with correspondingly rich proof terms.
\citet{DBLP:journals/corr/WuX17} encode session types in their wide
spectrum language ATS, which includes DML-style type dependency. Indexed types with unpolarized quantification
are used to represent channel types. Types for channel ends are
obtained by interpreting the quantifiers. While ATS provides all
features for verification, \LDGV{} is a minimalist dependent calculus
geared towards practical applications.
\citet{DBLP:conf/fossacs/ToninhoY18} develop a language with dependent
session types that integrates processes and functional computation via
a monadic embedding.  Processes may thus depend on expressions as well
as expressions may depend on monadic process values.

Compared to our work, their theory encompasses type-level functions
with type and value dependent kinds and monads, whereas type-level
computation in \LDGV{} is restricted to label introduction and
elimination. Their work strictly separates linear and unrestricted
assumptions, which leads to further duplication, and it has no notion
of subtyping.  Our setup formalizes large elimination for labels,
which is needed in practical applications, but not considered in their
work. Moreover, the point of our calculus is to showcase an economic
operational semantics with just one communication reduction at the
process level. We expect that \LDGV{} can be extended with further
index types and type-level computation without complicating the
operational semantics.

Lolliproc \cite{DBLP:conf/icfp/MazurakZ10} is a core calculus for
concurrent functional programming. Its primitives are derived from a
Curry-Howard interpretation of classical linear logic. Some form of
session types can be expressed in Lolliproc, but it does not support
unrestricted values nor dependency.


\citet{DBLP:journals/corr/abs-1211-4099} introduce a notion of
session types with refinements over linear resources specified by
uninterpreted predicates.  Even if linear, the dependency is not on
expressions of the programming language, thus greatly simplifying the
underlying theory.
\citet{DBLP:journals/entcs/BonelliCG04} study a simpler
extension for session types whereby assume/assert labels present in
expressions make their way into types to represent starting and ending
points in protocols.


\citet{DBLP:journals/mscs/GotoJJPR16} consider a polymorphic session
typing system for a $\pi$-calculus which replaces branching and choice
by matching and mismatching tests. These tests compare tokens, akin to
our labels, and introduce (in)equational constraints in the type
system.

\paragraph{Others.}

\citet{Nishimura1998} considers a calculus for objects where messages
(a method name and parameters) are first-class constructs. Each such
message is typed as the set of method names that may be invoked by the
message, formalized in a second order polymorphic type system.
\citet{DBLP:conf/isotas/VasconcelosT93} and
\citet{DBLP:journals/iandc/Sangiorgi98} pursue a similar idea in
the context of the $\pi$-calculus that allow the transmission of
variant values, say a label, together with an integer value. We follow
a different approach, by exchanging values only, while labels appear
as a particular case. Neither these works use (label) dependent types
to classify messages.

ROSE \cite{DBLP:journals/pacmpl/MorrisM19} is a versatile theory of
row typing that could be applied to session types among other
applications. Strikingly, a row is a mapping from labels to
types. Hence, a row type could be expressed by a $\Pi$ type in our
system using a case for the label dispatch. ROSE has a fixed set of
constraints for combining rows and requires labels to be compile-time
constants. \LDGV{} labels are first-class objects and combinations are
expressed with user-defined functions.


%% file: conclusions.tex
\section{Conclusions}
\label{sec:conclusions}

\LDGV{} is a minimalist calculus that combines dependent types and
session types from the point of economy of expression: a single pair
of communication primitives is sufficient. It faithfully extends
existing systems while retaining wire compatibility with
them. Building the calculus on dependent types liberates the structure
of session-typed programs from mimicking the type structure.

%
\LDGV{} supports encodings of algebraic datatypes with subtyping by
modeling tagged data with $\Sigma$-types. The same approach may be
used to simulate session calculi based on sending and receiving tagged
data.
It further incorporates natural numbers and primitive recursion at the
type level.

We are currently working on a few extensions for \LDGV.
\begin{enumerate}
\item Our implementation already supports recursive session types and
  we expect that the properties of algorithmic typing also extend to
  this setting.
\item We plan to address subtyping for the type recursor in a coinductive manner.
\item 
  It would be interesting to add further kinds of predicates beyond
  equality as well as type dependency (as supported by previous work
  \cite{ToninhoCairesPfenning2011,DBLP:conf/fossacs/ToninhoY18}).
\end{enumerate}


%% file: appendix.tex
\section{Rules for Process Typing and for Algorithmic Subtyping}
\label{sec:compl-typing-rules}
\begin{itemize}
\item Process typing (Figure~\ref{fig:ldgv-process-typing})
\item Algorithmic kinding (Figure~\ref{fig:ldgv-algorithmic-kind-inference})
\item Algorithmic subtyping (Figure~\ref{fig:ldgv-algorithmic-subtyping})
\item Extensions to algorithmic typing to support naturals (Figure~\ref{fig:nats-alg})
\end{itemize}
\input{fig-ldgv-process-typing}
\input{fig-ldgv-algorithmic-subtyping}
\input{fig-nat-alg}


\input{sec-example}
\input{ldgv-meta}
\input{proofs-alg-type-check}
\input{proofs-embedding-gv}


%% file: fig-ldgv-process-typing.tex
\begin{figure}[tp]
  \begin{mathpar}
    \RuleProcExpr

    \RuleProcChannel

    \RuleProcPar
\end{mathpar}
  \caption{\LDGV{} process typing}
  \label{fig:ldgv-process-typing}
\end{figure}

%% file: fig-ldgv-algorithmic-subtyping.tex
\begin{figure*}[tp]
  Algorithmic kind synthesis and checking
  \hfill\fbox{$\JAlgKindSynth\TEnv A \Kind$}\quad\fbox{$\JAlgKindCheck\TEnv A \Kind$}

  \begin{mathpar}
    \RuleAEqualityF

    \RuleAUnitF

    \RuleAEndF

    \RuleALabF

    \RuleALabET

    \RuleAPiF

    \RuleASigmaF

    \RuleASsnOutF

    \RuleASsnInF
    \RuleASubKind
  \end{mathpar}
  \caption{Algorithmic kinding}
  \label{fig:ldgv-algorithmic-kind-inference}
\end{figure*}
\begin{figure*}[tp]
  Algorithmic subtyping (synthesis)
  \hfill
  \fbox{$\JAlgSubtypeSynth\TEnv A B \Kind$}
  \begin{mathpar}
    \RuleASUnit

    \RuleASLabel
  
    \RuleASPi

    \RuleASSigma

    \RuleASSend

    \RuleASRecv

    \RuleASCaseLeftA
    
    \RuleASCaseLeftB

    \RuleASCaseRightA

    \RuleASCaseRightB
  \end{mathpar}
  Algorithmic subtyping (check against)
  \hfill
  \fbox{$\JAlgSubtypeCheck\TEnv A B \Kind$}
  \begin{mathpar}
    \RuleASCheck    
  \end{mathpar}
  \caption{Algorithmic subtyping in \LDGV{}}
  \label{fig:ldgv-algorithmic-subtyping}
\end{figure*}


%% file: fig-nat-alg.tex
\begin{figure*}[t]
  Algorithmic Value Conversion
  \hfill\fbox{$    \JAlgConv\TEnv V W{A}$}
  \begin{mathpar}
    \RuleACReflNat

    \RuleACAssNat
  \end{mathpar}
  Algorithmic Value Unfolding
  \hfill\fbox{$\JAlgTypeUnfold\TEnv {A}{B} $}
  \begin{mathpar}
    \RuleAUnfoldRecZ

    \RuleAUnfoldRecS

    \RuleAUnfoldRecX
  \end{mathpar}
  Algorithmic kind synthesis
  \hfill\fbox{$\JAlgKindSynth\TEnv{A}{\Kind}$}
  \begin{mathpar}
    \RuleARecET

    \RuleARecAlpha
  \end{mathpar}
  Algorithmic subtyping
  \hfill\fbox{$\JAlgSubtypeSynthExt\TEnv A B \Kind \Subst$}
  \begin{mathpar}
    \RuleASRecLeftZ

    \RuleASRecLeftS

    \RuleASRecLeftX

    \RuleASTVarLeft
  \end{mathpar}
  Algorithmic type synthesis\hfill\fbox{$\JAlgTypeSynth\TEnv M A \TOut$}
  \begin{mathpar}
    \RuleANatE
  \end{mathpar}
  \caption{Extensions of algorithmic typing for natural numbers and recursor}
  \label{fig:nats-alg}
\end{figure*}


%% file: sec-example.tex
\section{Example of a Typing Derivation}
\label{sec:example-type-deriv}

 An an example, we present the type derivation for the function
\lstinline{sendNode} from Section~\ref{sec:furth-opport},
Listing~\ref{lst:motivation:sending-receiving-nodes}. For brevity we
write $E$ and $N$ for the labels Empty and Node.
\begin{align*}
  \textit{sendNode}
  &=
    \lambda_\MultUn (n:\textit{Node})
    \lambda_\MultUn (c:\textit{NodeC})
    \ELet{\EPair[] t v} n \SEND{(\SEND c\,t)}\, v
  \\
  \TEnv_1
  &= n:\textit{Node}, c:\textit{NodeC}
  \\
  \TEnv_0 = \UNR{\TEnv_1}
  &= n:\textit{Node}
  \\
  \textit{EN}
  &= \{E, N\}
  \\
  \TEnv_{tv}
  &=  t :\textit{EN}, v : \CASE t {E:\TUnit, N:\TInt}
\end{align*}
The typing environment $\TEnv_1$ shows up in the premises after
processing the two lambdas. We have 
$\Decompose{\TEnv_1}{\TEnv_1}{\TEnv_0}$ and $\TEnv_0 = \UNR{\TEnv_1}$.
\begin{mathpar}
  \inferrule*
  {\inferrule*
    {\TEnv_0 \vdash n : \textit{Node} \\
      {D} \\
      D_0
    }
    {\TEnv_1 \vdash \ELet{\EPair[] t v} n \SEND{(\SEND c\,t)}\, v : \TUnit}}
  {\EmptyEnv \vdash \textit{sendNode} : \Pi_\MultUn(n:\textit{Node})\Pi_\MultUn(c:\textit{NodeC})\TUnit}
\end{mathpar}
As $t$ has a label type, we apply the $\TirName{Sigma-G}$ rule at the
top.  The subderivation $D$ establishes formation for the type of $v$.
\begin{mathpar}
  \inferrule*
  {
    \inferrule*
    {\EnvForm \MultUn {\UNR{\TEnv_1}, t : \textit{EN}, x_E : t=E} }
    {\UNR{\TEnv_1}, t : \textit{EN}, x_E : t=E \vdash \TUnit : \KindUn}
    \\
    \inferrule*
    {\EnvForm \MultUn {\UNR{\TEnv_1}, t : \textit{EN}, x_N : t=N} }
    {\UNR{\TEnv_1}, t : \textit{EN}, x_N : t=N \vdash \TInt : \KindUn}
  }
  {\UNR{\TEnv_1}, t : \textit{EN} \vdash \CASE t {E:\TUnit, N:\TInt} : \KindUn}
\end{mathpar}
The subderivation $D_0$ checks the expression $\CASE t {E: \SEND{(\SEND c\,t)}\, v,
  N: \SEND{(\SEND c\,t)}\, v}$ which amounts to checking $\SEND{(\SEND
  c\,t)}\, v$ once with $t=E$ and once with $t=N$.
\begin{mathpar}
  D_0 = 
  \inferrule*
  {
    \inferrule*
    {
    D_1
    \\
      \inferrule*
      {\TEnv_1,\TEnv_{tv}, x_E : t=E \vdash v : \CASE t {E:\TUnit, N:\TInt}}
      {\TEnv_1,\TEnv_{tv}, x_E : t=E \vdash v : \TUnit}
    }
    {\TEnv_1,\TEnv_{tv}, x_E : t=E \vdash \SEND{(\SEND c\,t)}\, v : \TUnit}
    \\
    D_N
  }
  {\TEnv_1,\TEnv_{tv} \vdash \CASE t {E:\SEND{(\SEND c\,t)}\, v, N: \SEND{(\SEND c\,t)}\, v} : \TUnit}
\end{mathpar}
The omitted subderivation $D_N$ for $t=N$ is analogous. 
The subderivation $D_1$ analyses the application of the send operations.
\begin{mathpar}
  \inferrule*{
  \inferrule*
  {
    \inferrule*
    {\TEnv_1,\TEnv_{tv},x_E : t=E \vdash \SEND c : \Pi_\MultLin (t:\textit{EN})\CASE t{E:!\TUnit.\TUnit, N:!\TInt.\TUnit
      } \\
      \TEnv_0, \TEnv_{tv},x_E : t=E \vdash t : \textit{EN}
    }
    {
      \TEnv_1,\TEnv_{tv},x_E : t=E \vdash \SEND c\,t : \CASE t{E:!\TUnit.\TUnit, N:!\TInt.\TUnit
      }
    }
  }
  {
    \TEnv_1,\TEnv_{tv},x_E : t=E \vdash \SEND c\,t : !\TUnit.\TUnit
  }
}{
    \TEnv_1,\TEnv_{tv},x_E : t=E \vdash \SEND (\SEND c\,t) : \Pi_\MultLin(\TUnit)\TUnit
  }
\end{mathpar}


%% file: ldgv-meta.tex
\section{Proofs for {\LDGV}}
\label{sec:metatheory}

This section collects the standard metatheoretical results for {\LDGV}
culminating in typing preservation and progress.

\begin{restatable}[Strengthening]{lemma}{Strengthening}
  \label{lem:strengthening}
    If $\TEnv,x\colon B \vdash A : \Kind$ and $x\notin\FV(A)$, then $\TEnv\vdash A :  \Kind$.
\end{restatable}
\begin{proof}
  By rule induction.
\end{proof}

Weakening can only be established for additional unrestricted
bindings, that is, bindings for types $A:\KindUn$.

\begin{restatable}[Weakening]{lemma}{Weakening}
  \label{lemma:weakening}
    If $\TEnv \vdash A : \Kind$ and $\TEnv \vdash B : \Kindn$, then
    $\DemoteExtend\TEnv xB \vdash A : \Kind$.
\end{restatable}
\begin{proof}
  By rule induction.
\end{proof}

The below proof includes the cases of the rules for natural numbers.

\begin{restatable}[Agreement]{lemma}{Agreement}
  \label{lem:agreement}
  \
  \begin{enumerate}
  \item If $\EnvForm\KindUn{\TEnv}$ and
    $\Decompose{\TEnv}{\TEnv_1}{\TEnv_2}$, then
    $\EnvForm\KindUn{\TEnv_1}$ and $\EnvForm\KindUn{\TEnv_2}$.
  \item If $\TEnv \vdash A : \Kind$, then $\EnvForm\KindUn{\TEnv}$. 
  \item If $\TEnv \vdash \Convertible A B : \Kind$, then 
    $\TEnv \vdash A : \Kind$ and $\TEnv \vdash B : \Kind$. 
  \item If $\TEnv \vdash A \Subtype B : \Kind$, then 
    $\TEnv \vdash A : \Kind$ and $\TEnv \vdash B : \Kind$. 
  \item If $\TEnv \vdash M : A$, then $\UNR\TEnv \vdash A: \Kind$,
    for some $\Kind$.
  \end{enumerate}
\end{restatable}
\begin{proof}
  By mutual rule induction on the various hypotheses.
  \par
  (3) Use Weakening (Lemma~\ref{lemma:weakening}) in the case of rule
  \TirName{Conv-Beta};
  Properties of equality type (Lemma~\ref{lemma:equality-type}) in the
  case of rule \TirName{Conv-Subst};
  Substitution (Lemma~\ref{lemma:type-substitution}) in the case of
  \TirName{Conv-S}.

  (4) Use Context Subtyping (Lemma~\ref{lemma:context-subtyping}) in
  the case of rule \TirName{Sub-Pi}, \TirName{Sub-Send},
  \TirName{Sub-Rev}, and \TirName{Sub-Case}.

  (5) Use Strengthening (Lemma~\ref{lem:strengthening}) in the case of
  rules \TirName{Lab-E}, \TirName{Sigma-E}, \TirName{Lab-G}, and
  \TirName{Nat-E};
  Substitution (Lemma~\ref{lemma:type-substitution}) in the case of
  rule \TirName{Pi-E} ;
  Properties of context split (Lemma~\ref{lemma:context-split}) in
  the case of rules \TirName{Sigma-I}, \TirName{Sigma-E}, and
  \TirName{Sigma-G};
  Weakening (Lemma~\ref{lemma:weakening}) in the case of rule
  \TirName{Ssn-I};
  Kinding duality (Lemma~\ref{lemma:kinding-duality}) in the case of
  rule \TirName{Ssn-I}.
\end{proof}

The following lemma introduces basic properties of the context split
operations, used in mostly other results.

\begin{lemma}[Properties of context split]
  \label{lemma:decomposition-preserves-wellformedness}
  \label{lemma:decomposition-unrestricted}
  \label{lemma:properties-un}
  \label{lemma:context-split}
  Suppose that $\Decompose{\TEnv}{\TEnv_1}{\TEnv_2}$.
  \begin{enumerate}
  \item If $\EnvForm \Multm \TEnv$,
  then $\EnvForm \Multm {\TEnv_1}$ and $\EnvForm \Multm{ \TEnv_2}$.
  \item If $\EnvForm \MultUn \TEnv$, then $\TEnv = \TEnv_1 = \TEnv_2$.
  \item\label{item:uniquely-det} If $\EnvForm \MultUn{ \TEnv_1}$, then $\TEnv = \TEnv_2$ and $\TEnv_1$
    is uniquely determined by $\TEnv$.
  \item\label{item:split-omega} If $\EnvForm \MultUn{ \TEnv_1}$ and
    $\EnvForm \MultUn{ \TEnv_2}$, then $\EnvForm \MultUn \TEnv$.
  \item\label{item:split-commutative} $\Decompose{\TEnv}{\TEnv_2}{\TEnv_1}$.
  \item\label{item:split-demote} $\DecomposeOp{\UNR\TEnv_1}{\TEnv_2}$ is defined.
  \item $\Decompose\TEnv\TEnv{\UNR\TEnv}$.
  \item $\UNR\TEnv = \UNR{\TEnv_1} = \UNR{\TEnv_2}$.
  \end{enumerate}
\end{lemma}
\begin{proof} By rule induction on the various hypotheses. 
\end{proof}


\begin{lemma}[Properties to equality type]
  \label{lemma:equality-type}
    Let $\TEnv \vdash \Equation VW : m$.
    \begin{enumerate}
    \item If $\TEnv \vdash V : A$, then  $\TEnv \vdash W : A$.
    \item If $\TEnv x : \Equation VU \vdash A : m$, then
      $\TEnv x : \Equation WU \vdash A : m$.
    \end{enumerate}
\end{lemma}
\begin{proof}
  \vv{TODO}
\end{proof}

\begin{restatable}{lemma}{LDGVUnrestrictedValuesEnvironments}\label{lemma:unrestricted-values-have-unrestricted-environments}
  If $\TEnv \vdash V : A$ and $\UNR\TEnv \vdash A : \KindUn$, then
  $\EnvForm \MultUn \TEnv$.
\end{restatable}
\begin{proof}
  By rule induction on the first hypothesis.
  
  \textbf{Cases }\TirName{Unit-I}, \TirName{Lab-I}, and
  \TirName{Z-I}. The conclusion is one of the premises to the rule.

  \textbf{Case }\TirName{Lab-E}. From premise $\UNR\TEnv \vdash V : L$
  and Agreement (Lemma~\ref{lem:agreement}) we know that
  $\UNR\TEnv \vdash L : \Multm$. The result follows from the premise of
  rule \TirName{Lab-F}, the only rule that applies.

  \textbf{Case }\TirName{Name}. From the premise and the second
  hypothesis to the lemma.

  \textbf{Case }\TirName{Pi-I}. From the rule we read that
  $C = \Pi_\Multm (x:A)B$ and $\EnvForm \Multm \TEnv$.  By hypothesis
  $\UNR{\TEnv} \vdash C:\KindUn$ it must be that $\Multm = \MultUn$,
  which proves the claim.

  \textbf{Case }\TirName{Sigma-I}. We have $C= \Sigma (x:A)B$ as
  well as the following premise $\UNR\TEnv \vdash C : \MultUn$. The
  result follows from Agreement.

  \textbf{Case }\TirName{Sssn-Send-E}. We find that
  $C = \Pi_\MultLin (x:A)B$ which contradicts the assumption
  $\UNR{\TEnv} \vdash C : \MultUn$. This contradiction establishes the
  claim.
  
  \textbf{Case }\TirName{Sub-Type}. Premises are $\TEnv \vdash V:C$
  and $\UNR\TEnv \vdash C \Subtype A : \Kindm$. From hypothesis
  $\UNR\TEnv \vdash A: \MultUn$ and Agreement we have
  $\UNR\TEnv \vdash C : \MultUn$. By induction $\TEnv \vdash \MultUn$.
  
  \textbf{Case }\TirName{S-I}. By induction.
\end{proof}

\begin{lemma}[Context subtyping]
  \label{lemma:context-subtyping}
  If $\UNR\TEnv \vdash A \Subtype B$ and
  $\DemoteExtend\TEnv xA \vdash C \colon \Kind$, then
  $\DemoteExtend\TEnv xB \vdash C \colon \Kind$.
\end{lemma}
\begin{proof}
  \vv{TODO}
\end{proof}

\begin{lemma}[Subtyping duality]\label{lemma:subtyping-duality}
  If $\TEnv \vdash R \Subtype S : \Kind$,
  then
  $\TEnv \vdash \DUAL S \Subtype \DUAL R : \Kind$.
\end{lemma}
\begin{proof}
  Rule induction on the hypothesis.
\end{proof}

It follows that, if $\DUAL A$ and $\DUAL B$ are both defined, then
$\TEnv \vdash A \Subtype B : \Kind$ iff
$\TEnv \vdash \DUAL B \Subtype \DUAL A : \Kind$.

\begin{lemma}[Kinding duality]\label{lemma:kinding-duality}
  If $\TEnv \vdash S : \Kind$, then $\TEnv \vdash \DUAL S:\Kind$.
\end{lemma}




\begin{restatable}[Substitution for context formation]{lemma}{LDGVSubstitutionFormation}\label{lemma:substitution-formation}
  If
  $\EnvForm \MultUn {\TEnv_1, x: A, \TEnv_2}$ and
  $\TEnv_1 \vdash V : A$, then
  $\EnvForm \MultUn {\TEnv_1, \TEnv_2[V/x]}$.
\end{restatable}
\begin{proof}
  By induction on $\TEnv_2$.

  \textbf{Case }$\EmptyEnv$. Immediate.

  \textbf{Case }$\TEnv_2, y: B$ with $x\ne y$ and $\TEnv_2 \vdash B:\KindUn$
  by assumption.
  By induction $\EnvForm \MultUn{ \TEnv_1, \TEnv_2[W/x]}$.
  By induction (on typing) $\TEnv_1, \TEnv_2[W/x] \vdash B[W/x]:
  \Kind$.
  Hence $\EnvForm \MultUn{ \TEnv_1, \TEnv_2[W/x], y: B[W/x]}$.
\end{proof}

\begin{restatable}[Substitution for Types]{lemma}{LDGVTypeSubstitution}\label{lemma:type-substitution}
  If $\DemoteExtend{\TEnv_1} x B \vdash A : \Kind$ and
  $\TEnv_2 \vdash N : B$ and $\Decompose\TEnv{\TEnv_1}{\TEnv_2}$, then
  $\TEnv \vdash A[N/x] : \Kind$.
\end{restatable}
\begin{proof}
  \vv{The proof is by mutual rule induction on the first
    hypothesis. All the substitution lemmas should come together in a
    big lemma with many cases. Checked equality, unit, end, and
    lab. Did not have a look at the remaining cases}
\end{proof}

\begin{restatable}[Substitution for convertibility]{lemma}{LDGVSubstitutionConvertibility}\label{lemma:substitution-convertibility}
  Suppose that
  $\TEnv \vdash V : A$ and
  $\UNR\TEnv \vdash C : \KindUn$. If
  $\UNR\TEnv, x : C, \SEnv \vdash \Convertible A B :
  \Kind$, then
  $\TEnv, \SEnv[V/x] \vdash \Convertible {A[V/x]} {B[V/x]} : \Kind$.
\end{restatable}
\begin{proof}
  The proof is by induction on the derivation of $\TEnv_1, x : C, \TEnv_3
  \vdash \Equation A B : \Kind$.

  \textbf{Case }$\inferrule[Conv-Refl]{\TEnv \vdash A : \Kind}{\TEnv \vdash \Convertible A A :
    \Kind}$: By induction (on typing), we obtain
  \begin{gather}
    \label{eq:65}
    \TEnv, \TEnv_3[W/x] \vdash {A[W/x]} : \Kind
    \intertext{and we apply \TirName{Conv-Refl} to yield}
    \TEnv, \TEnv_3[W/x] \vdash \Convertible{A[W/x]}{A[W/x]} : \Kind
m  \end{gather}

  \textbf{Case }$\inferrule[Conv-Sym]{\TEnv \vdash \Convertible A B : \Kind}{\TEnv \vdash
    \Convertible B A : \Kind}$: immediate by induction.


  \textbf{Case }$
  \inferrule[Conv-Subst]{
    \TEnv, z:\LabelSet \vdash B : \Kind \\
    \TEnv \vdash A : \Equation M N
  }{
    \TEnv \vdash
    \Convertible {B[M/z]} {B[N/z]}
    : \Kind
  }
  $: By induction (on typing) and observing that $z$ is chosen such
  that $z\ne x$, we obtain 
  \begin{gather}
    \label{eq:66}
    \TEnv, z:\LabelSet, \TEnv_3[W/x] \vdash B[W/x] : \Kind \\
    \label{eq:120}
    \TEnv, \TEnv_3[W/x] \vdash A[W/x] : \Equation {M[W/x]} {N[W/x]} 
  \end{gather}
  Applying rule \TirName{Conv-Subst} yields the desired result.
  \begin{gather}
    \TEnv, \TEnv_3[W/x] \vdash
    \Convertible {B[M/z][W/x]} {B[N/z][W/x]}
    : \Kind
  \end{gather}

  \textbf{Case }$\inferrule[Conv-Red]
  {\TEnv \vdash \Context[T]A: \Kind}
  {\TEnv \vdash \Convertible {\Context[T]A} {\Context[T]B} : \Kind}
  \quad{\text{if}\ A \ReducesTo B}$: By induction (on typing), we obtain
  \begin{gather}
    \label{eq:69}
    \TEnv, \TEnv_3[W/x] \vdash {\Context[T]A}[W/x] : \Kind
  \end{gather}
  Furthermore, reduction is closed under substitution of values
  hence $A[W/x] \ReducesTo B[W/x]$ and 
  we can conclude with rule \TirName{Conv-Red}:
  \begin{gather}
    \label{eq:70}
    \TEnv, \TEnv_3[W/x] \vdash {\Context[T]B}[W/x] : \Kind
  \end{gather}

  This case concludes the proof.
\end{proof}

\begin{restatable}[Substitution for subtyping]{lemma}{LDGVSubstitutionSubtyping}\label{lemma:substitution-subtyping}
  Suppose that
  $\UNR\TEnv \vdash C : \Kind$. If
  $\UNR\TEnv, x: C, \SEnv \vdash A \Subtype B : \Kind'$ and
  $\TEnv \vdash V : C$, then
  $\TEnv, \SEnv[V/x] \vdash A[V/x] \Subtype B[V/x] : \Kind'$.
\end{restatable}
\begin{proof}
  By induction on the derivation of the subtyping judgment.

  \textbf{Case }$  \inferrule[Sub-Conv]{\TEnv \vdash \Convertible A B
    : \Kind}{\TEnv \vdash A \Subtype B : \Kind} 
  $: immediate by IH through
  Lemma~\ref{lemma:substitution-convertibility}.

  \textbf{Case }$  \inferrule[Sub-Lab]
  { \EnvForm \MultUn \TEnv \\ \LabelSet \subseteq \LabelSet' }
  { \TEnv \vdash \LabelSet \Subtype \LabelSet' : \KindUn }
  $: immediate.

  \textbf{Case }$  \inferrule[Sub-Trans]{
    \TEnv \vdash  A \Subtype B : \Kind \\
    \TEnv \vdash  B \Subtype C : \Kind}{
    \TEnv \vdash  A \Subtype C : \Kind}
  $: immedidate by IHs.

  \textbf{Case }$ \inferrule[Sub-Pi]
  { \TEnv \vdash A' \Subtype A : \Kind_A^\Multl
    \\ \TEnv, z:^{\Demote\Multl}A' \vdash B \Subtype B' : \Kind_B
    \\ \Multm \MultLE \Multn }
  {\TEnv \vdash \Pi_\Multm (z:A) B \Subtype \Pi_\Multn (z:A')B' : \Multn}
  $:
  Induction on the first subgoal yields
  \begin{gather}
    \TEnv, \TEnv_3[W/x] \vdash A'[W/x] \Subtype A[W/x] : \Kind_A^\Multl
  \end{gather}
  Induction on the second subgoal yields
  \begin{gather}
    \TEnv, (\TEnv_3, z:^{\Demote\Multl}A')[W/x] \vdash B[W/x] \Subtype B'[W/x] : \Kind_B
  \end{gather}
  Hence the claim
  \begin{gather}
    {\TEnv, \TEnv_3[W/x] \vdash (\Pi_\Multm (z:A) B \Subtype
      \Pi_\Multn (z:A')B')[W/x] : \Multn}
  \end{gather}

  \textbf{Case }$\RuleSubSigma$:\\
  Induction on the first subgoal yields
  \begin{gather}
    \TEnv, \TEnv_3[W/x] \vdash A[W/x] \Subtype A'[W/x] : \Kind
  \end{gather}
  Induction on the second subgoal yields
  \begin{gather}
    \TEnv, (\DemoteExtend{\TEnv_3} zA)[W/x] \vdash B[W/x]
    \Subtype B'[W/x] : \Kind
  \end{gather}
  Putting those two together yields the claim
  \begin{gather}
    {\TEnv, \TEnv_3[W/x] \vdash (\Sigma (z:A) B \Subtype \Sigma (z:A')B')[W/x]
      : \Multn}  
  \end{gather}

  \textbf{Case }$\RuleSubSend$:\\
  Analogous to \TirName[Sub-Pi].

  \textbf{Case }$\RuleSubRecv$:\\
  Analogous to \TirName{Sub-Sigma}.

  \textbf{Case }$\RuleSubCase$:\\
  Induction on the first subgoal (for typing) yields
  \begin{gather}
    \TEnv, \TEnv_3[W/x] \vdash N[W/x] : \LabelSet\cap \LabelSet'
  \end{gather}
  Induction on the second family of subgoals yields, for each $\ell\in \LabelSet\cap \LabelSet'$,
  \begin{gather}
    \TEnv, (\TEnv_3, z: \Equation N \ell)[W/x] \vdash
    B_\ell[W/x] \Subtype B'_\ell[W/x] : \Kind
    \intertext{which is the same as}
    \TEnv, \TEnv_3[W/x], z: \Equation {N[W/x]} \ell \vdash
    B_\ell[W/x] \Subtype B'_\ell[W/x] : \Kind
  \end{gather}
  Hence, we can conclude with
  \begin{gather}
    { \TEnv, \TEnv_3[W/x] \vdash
      (\CASE N {\overline{\ell: {B_\ell}}^{\ell\in \LabelSet}}
      \Subtype 
      \CASE N {\overline{\ell: {B'_\ell}}^{\ell\in \LabelSet'}})[W/x] : \Kind
    }  
  \end{gather}
  This case concludes the proof.
\end{proof}

\begin{restatable}[Substitution for typing]{lemma}{LDGVSubstitutionTyping}\label{lemma:substitution}
  Suppose that $\Decompose\TEnv{\TEnv_1}{\TEnv_2}$ and
  $\UNR\TEnv \vdash A : \Kind$ \vv{This follows from hypothesis
    $\TEnv_2 \vdash V : A$ by agreement}. If
  $\DemoteExtend{\TEnv_1} x A \vdash M : B$ and
  $\TEnv_2 \vdash V : A$, then $\TEnv \vdash M[V/x] : B[V/x]$.
  \vv{This lemma must be generalised to $\TEnv_2 \vdash N :
    A$. Furthermore all substitutions statements must be under a big
    lemma. And the proof is by mutual rule induction.}
\end{restatable}
\begin{proof}
  We need to generalize the statement to account for the possibility
  that $x: C$ occurs somewhere in the middle of the
  typing environment: for all $\TEnv_3$, if
  $\DemoteExtend{\TEnv_1} x C, \TEnv_3 \vdash M : B$ and
  $\TEnv_2 \vdash W : C$, then
  $\TEnv, \TEnv_3[W/x] \vdash M[W/x] : B[W/x]$.
  \vv{This should be the statement of the lemma.}

  The proof is by induction on the derivation of
  $\DemoteExtend{\TEnv_1} x C, \TEnv_3 \vdash M : B$.

  \textbf{Case }$\RuleSubKind$: immediate by IH and because kinds are atomic.

  \textbf{Case }$\RuleSubType$: From
  \begin{gather}
    \PlainExtend{\TEnv_1} x C, \TEnv_3 \vdash M : B
  \end{gather}
  we obtain by inversion
  \begin{gather}
    \label{eq:121}
    \PlainExtend{\TEnv_1} x C, \TEnv_3 \vdash M : A
    \\
    \label{eq:122}
    \UNR{(\TEnv_1 , x: C, \TEnv_3)} \vdash A \Subtype B
    : \Kind
  \end{gather}
  Induction for (\ref{eq:121}) yields
  \begin{gather}
    \TEnv, \TEnv_3[W/x] \vdash M[W/x] : A[W/x]
  \end{gather}
  Induction for (\ref{eq:122}) yields
  \begin{gather}
    \UNR{(\TEnv_1, \TEnv_3[W/x])} \vdash A[W/x] \Subtype B[W/x]
    : \Kind
  \end{gather}
  Applying \TirName{Sub-Type} yields
  \begin{gather}
    \TEnv, \TEnv_3[W/x] \vdash M[W/x] : B[W/x]
  \end{gather}

  \textbf{Case }$\inferrule[Unit-F]{\EnvForm \MultUn \TEnv}{\TEnv \vdash
    \TUnit : \KindUnSession}$: by Lemma~\ref{lemma:substitution-formation} if
  $\Multn = \MultUn$. If $\Multn=\MultLin$, the implication is void.

  \textbf{Case }$\inferrule[Unit-I]{\EnvForm \MultUn \TEnv}{\TEnv \vdash \EUnit :  \TUnit}$:
  by Lemma~\ref{lemma:substitution-formation} as in the previous case.

  \textbf{Case }$\inferrule[Lab-F]{\EnvForm \MultUn \TEnv \\
    \LabelSet \subseteq_\Fin \LabelUniv
  }{\TEnv \vdash \LabelSet : \KindUn}$: immediate with Lemma~\ref{lemma:substitution-formation}.

  \textbf{Case }$\inferrule[Lab-I]{
    \ell \in \LabelSet \\
    \TEnv \vdash \LabelSet : \KindUn
  }{\TEnv \vdash \ell : \LabelSet}$: immediate by induction.

  \textbf{Case }$
  \inferrule[Lab-E]{
    \Decompose{\TEnv}{\TEnv_1}{\TEnv_2} \\
    \EnvForm \MultUn \TEnv_1 \\
    \TEnv_1 \vdash M : \LabelSet \\
    (\forall\ell\in \LabelSet)~\TEnv_2, z : \Equation M \ell \vdash N_\ell : B
  }{\TEnv \vdash \CASE M {\overline{\ell: N}^{\ell\in L}} : B
  }$: the starting point is
  \begin{gather}
    \label{eq:74}
    \TEnv_1 , x: C, \TEnv_3 \vdash \CASE M {\overline{\ell: N}^{\ell\in L}} : B
  \end{gather}
  Inversion yields
  \begin{gather}
    \label{eq:75}
    \Decompose{\TEnv_1 , x: C, \TEnv_3}{\TEnv_u}{\TEnv_r}
    \\
    \label{eq:76}
    \EnvForm \MultUn{ \TEnv_u}
    \\
    \label{eq:77}
    \TEnv_u \vdash M : \LabelSet
    \\
    \label{eq:78}
    (\forall\ell\in L)~ \TEnv_r, z: \Equation M \ell \vdash N_\ell : B
  \end{gather}
  By Lemma~\ref{lemma:decomposition-unrestricted}
  \begin{gather}
    \label{eq:79}
    \TEnv_r = \TEnv_1 , x: C, \TEnv_3 \\
    \TEnv_u = \DemoteExtend{\UNR{\TEnv_1}} x C, \UNR{\TEnv_3} \\
    \Decompose{\TEnv_1}{\TEnv_1}{\UNR{\TEnv_1}} \\
    \Decompose{\TEnv_3}{\TEnv_3}{\UNR{\TEnv_3}}
  \end{gather}
  Induction on \eqref{eq:78} yields
  \begin{gather}
    \label{eq:80}
    (\forall\ell\in L)~ \TEnv_1, \TEnv_3[W/x], \Equation {M[W/x]} \ell \vdash N_\ell[W/x] : B[W/x]
  \end{gather}
  Induction on \eqref{eq:77} yields
  \begin{gather}
    \label{eq:81}
    \UNR{\TEnv_1},\UNR{ \TEnv_3}[W/x] \vdash M[W/x] : \LabelSet
  \end{gather}
  Lemma~\ref{lemma:substitution-formation} applied to \eqref{eq:76} yields
  \begin{gather}
    \label{eq:82}
    \EnvForm \MultUn{ \UNR{\TEnv_1}, \UNR{\TEnv_3}[W/x]}
  \end{gather}
  Finally
  \begin{gather}
    \label{eq:83}
    \Decompose{\TEnv_1, \TEnv_3[W/x]}{\TEnv_1, \TEnv_3[W/x]}{\UNR{\TEnv_1}, \UNR{\TEnv_3}[W/x]}
  \end{gather}
  Applying \TirName{Lab-E} to (\ref{eq:83}), (\ref{eq:82}),
  (\ref{eq:81}), and~(\ref{eq:80}) yields the desired result 
  \begin{gather}
    \label{eq:84}
    \TEnv_1, \TEnv_3[W/x] \vdash
    \CASE {M[W/x]} {\overline{\ell: N_\ell[W/x]}^{\ell\in L}} : B[W/x]
  \end{gather}

  \textbf{Case }$\RulePiF$: immediate by induction.

  \textbf{Case }$\RuleName$: There are two cases.

  \textbf{Subcase }$x=z$: we are looking at
  \begin{gather}
    \label{eq:85}
    \TEnv_1 , z: C \vdash z : C
  \end{gather}
  Inversion yields
  \begin{gather}
    \label{eq:86}
    \EnvForm \MultUn {\TEnv_1}
  \end{gather}
  By assumption $\TEnv_2 \vdash W:C$, by Lemma~\ref{lemma:decomposition-unrestricted}
  $\Decompose{\TEnv_2}{\TEnv_2}{\TEnv_1}$, $W = x[W/x]$, and weakening (Lemma~\ref{lemma:weakening})
  we obtain the result 
  \begin{gather}
    \label{eq:88}
    \TEnv_2 \vdash W : C
  \end{gather}

  \textbf{Subcase }$x\ne y$: by induction, considering that $x$ may
  occur before or after $y$ in the environment. In the first case, $x$
  may appear in $A$, in the latter, it does not. Both are straightforward.

  \textbf{Case }$\RulePiI$:

  Assuming that $x\ne y$, our starting point is
  \begin{gather}
    \TEnv_1 , x: C, \TEnv_3 \vdash
    \lambda_\Multm (y: A).M : \Pi_\Multm (y:  A)B
   \end{gather}
   Inversion yields
   \begin{gather}
     \EnvForm \Multm {\TEnv_1 , x: C, \TEnv_3}
     \\
     \DemoteExtend{\TEnv_1, x:C, \TEnv_3} y A \vdash M : B
   \end{gather}
   Induction for all inverted judgments yields
   \begin{gather}
     \EnvForm \Multm{ \TEnv_1, \TEnv_3[W/x]}
     \\
     \DemoteExtend{\TEnv_1, \TEnv_3[W/x]} y {A[W/x]} \vdash M[W/x] : B[W/x]
   \end{gather}
   Applying \TirName{Pi-I} yields
   \begin{gather}
     \TEnv_1, \TEnv_3[W/x] \vdash
     (\lambda_\Multm (y: A).M)[W/x] : (\Pi_\Multm (y:  A)B)[W/x]
   \end{gather}

   \textbf{Case }$\RulePiE$.

  The starting point is
  \begin{gather}
    \Decompose\TEnv{\TEnv_1}{\TEnv_2}
    \\
    \TEnv_1, x:C, \TEnv_3 \vdash M\,N : B[N/y]
    \\
    \TEnv_2 \vdash W : C
  \end{gather}
  \textbf{Subcase.} Assuming that $\UNR{\TEnv_2} \vdash C : \KindUn$, inversion yields
  \begin{gather}
    \Decompose
    {(\TEnv_1, x:C, \TEnv_3)}
    {(\TEnv'_1, x: C, \TEnv'_3)}
    {(\TEnv''_1, x: C, \TEnv''_3)}
    \\
    \label{eq:60}
    \TEnv'_1, x: C, \TEnv'_3 \vdash M : \Pi_\Multm (y: A)B
    \\
    \label{eq:63}
    \TEnv''_1, x: C, \TEnv''_3 \vdash N : A
    \\
    \label{eq:64}
    \UNR{(\TEnv_1, x:C, \TEnv_3)} \vdash B[N/y] : \Kind
  \end{gather}

  The induction hypothesis for~(\ref{eq:60}),~(\ref{eq:63}), and~(\ref{eq:64}) yields
  \begin{gather}
    \label{eq:95}
    (\DecomposeOp{\TEnv'_1}{ \TEnv_2}), \TEnv'_3[W/x] \vdash M[W/x] : (\Pi_\Multm (y: A)B)[W/x]
    \\
    \label{eq:96}
    (\DecomposeOp{\TEnv''_1}{ \TEnv_2}), \TEnv''_3[W/x] \vdash N[W/x] : A[W/x]
    \\
    \label{eq:97}
    \UNR{(\TEnv_1, \TEnv_3[W/x])} \vdash B[N/y][W/x] : \Kind
  \end{gather}
  By
  Lemma~\ref{lemma:unrestricted-values-have-unrestricted-environments},
  we know that $\EnvForm\KindUn{\TEnv_2}$. Hence,
  $\Decompose{\TEnv_1'}{\TEnv_1'}{\TEnv_2}$ and
  $\Decompose{\TEnv_1''}{\TEnv_1''}{\TEnv_2}$ so that
  \begin{gather}
    \label{eq:94}
    \Decompose{(\TEnv_1, \TEnv_3[W/x])}{(\TEnv_1', \TEnv_3'[W/x])}{(\TEnv_1'', \TEnv_3''[W/x])}
  \end{gather}
  It remains to apply \TirName{Pi-E} to~(\ref{eq:94}), (\ref{eq:95}), (\ref{eq:96}), and~(\ref{eq:97}) to prove the judgment.

  \textbf{Subcase.} Assuming that $\UNR{\TEnv_2} \vdash C : \KindLin$
  and that $x\in\FV(N)$, inversion yields 
  \begin{gather}
    \Decompose
    {(\TEnv_1, x:C, \TEnv_3)}
    {(\TEnv'_1, \TEnv'_3)}
    {(\TEnv''_1, x: C, \TEnv''_3)}
    \\
    \label{eq:98}
    \TEnv'_1, \TEnv'_3 \vdash M : \Pi_\Multm (y: A)B
    \\
    \label{eq:99}
    \TEnv''_1, x: C, \TEnv''_3 \vdash N : A
    \\
    \label{eq:100}
    \UNR{(\TEnv_1, x: C, \TEnv_3)} \vdash B[N/y] : \Kind
  \end{gather}
  Hence $x$ 
  does not appear in $\TEnv'_3$, $M$, and $\Pi_\Multm (y: A)B$ so that
  they are indifferent to substitution:
  \begin{gather}
    \label{eq:67}
    \TEnv'_1, \TEnv'_3[W/x]
    \vdash M[W/x] : (\Pi_\Multm (y: A)B)[W/x]
  \end{gather}
  Moreover,   induction applied to~\eqref{eq:99} yields
  \begin{gather}\label{eq:68}
    (\DecomposeOp{\TEnv''_1}{ \TEnv_2}), \TEnv''_3[W/x] \vdash N[W/x] : A[W/x]
  \end{gather}
  Furthermore, $\UNR{(\TEnv_1, x:C, \TEnv_3)} = \UNR{(\TEnv_1, \TEnv_3)}$ contains
  no binding for $x$ so that $x$ does not appear on the
  right side of the kinding judgment~(\ref{eq:64}), which yields
  \begin{gather}\label{eq:71}
    \UNR{(\TEnv_1, \TEnv_3[W/x])} \vdash B[N/y][W/x] : \Kind
  \end{gather}
  Hence, we can apply rule \TirName{Pi-E} to (\ref{eq:67}),
  (\ref{eq:68}), and (\ref{eq:71}) to obtain
  \begin{gather}
    \DecomposeOp{(\TEnv'_1, \TEnv'_3[W/x])}{ ((\TEnv''_1 + \TEnv_2), \TEnv''_3[W/x])}
    \vdash (M\, N)[W/x] : B[N/y][W/x]
    \intertext{and resolving the decompositions yields}
    \TEnv, \TEnv_3[W/x] \vdash (M\, N)[W/x] : B[N/y][W/x]
  \end{gather}

  \textbf{Subcase.} The case where $x \in \FV (M)$ is analogous.

  \textbf{Case }$\RuleSigmaF$.
  Immediate by induction.

  \textbf{Case }$\RuleSigmaI$.
  Immediate by induction.


  \textbf{Case }$\RuleSigmaE$.

  Immediate by induction.

  \textbf{Case }$\RuleSsnOutF$.

  Immediate by induction.

  \textbf{Case}$\RuleSsnInF$.

  Immediate by induction.

  \textbf{Case}$\RuleSsnSendE$.  Immediate by induction.
  
  \textbf{Case}$\RuleSsnRecvE$.

  Immediate by induction.
\end{proof}

\begin{lemma}[Substitution for Types]
  \label{lemma:type-substitution}
  If $\TEnv, \alpha : \Kind \vdash A : \Kind$ and $\TEnv \vdash B :
  \Kind$, then $\TEnv \vdash A[B/\alpha] : \Kind$.
\end{lemma}
\begin{proof}
  \vv{TODO}
\end{proof}

{\LDGVSubjectReductionExpressions*}
\begin{proof}
  The proof is by cases on the reduction relation. According to
  a canonical derivation lemma, every typing derivation ends 
  with exactly one application of the \TirName{Sub-Type} rule on top of
  a structural rule. 

  \textbf{Case }$ \CASE {\ell_j} {\overline{\ell_i: N_i}^{1\le i\le n}}
   \ReducesTo
    N_j$ if $1\le j\le n$. 
  Suppose that
  \begin{gather}
    \label{eq:1}
    \TEnv \vdash \CASE {\ell_j} {\overline{\ell_i: N_i}^{1\le i \le n}} : A
  \end{gather}
  Inversion of subtyping yields
  \begin{gather}
    \label{eq:48}
    \Decompose{\TEnv}{\TEnv_1}{\TEnv_2} \\
    \label{eq:49}
    \EnvForm \MultUn {\TEnv_2} \\
    \label{eq:50}
    \TEnv_1 \vdash \CASE {\ell_j} {\overline{\ell_i: N_i}^{1\le i \le n}} : B\\
    \label{eq:51}
    \TEnv_2 \vdash B \Subtype A :\Kind
  \end{gather}
  Inversion of \eqref{eq:50} yields
  \begin{gather}
    \label{eq:2}
    \Decompose{\TEnv_1}{\TEnv_{11}}{\TEnv_{12}}
    \\
    \label{eq:52}
    \EnvForm \MultUn{ \TEnv_{11}} \\
    \label{eq:3}
    \TEnv_{11} \vdash \ell_j : \{\ell_1, \dots, \ell_n\} \\
    \label{eq:4}
    (\forall1\le i \le n)~\TEnv_{12}, z : \Equation{ \ell_j}{\ell_i} \vdash N_i : B
  \end{gather}
  Inversion of~\eqref{eq:3} yields
  \begin{gather}
    \label{eq:6}
    1 \le j \le n
    \\
    \TEnv_{11} \vdash \{\ell_1, \dots, \ell_n\} : \KindUn
  \end{gather}
  Now consider~\eqref{eq:4} for $i=j$.
  As $\EnvForm \MultUn{ \TEnv_{11}}$ and  $\EnvForm \MultUn{ \TEnv_2}$ it must be that
  $\TEnv_{12} = \TEnv_1$. Furthermore,
  the assumption $\Equation{\ell_j}{\ell_j}$ can be omitted. Thus, we have
  \begin{gather}
    \label{eq:7}
    \TEnv_1 \vdash N_j : B
  \end{gather}
  We apply subtyping to assumptions \eqref{eq:48}, \eqref{eq:49}, and
  \eqref{eq:51} to obtain
  \begin{gather}
    \label{eq:5}
    \TEnv_1 \vdash N_j : A
  \end{gather}

  \textbf{Case }$(\lambda_\Multm x.M)\, V  \ReducesTo M[V/x]$. Suppose
  that
  \begin{gather}
    \label{eq:8}
    \TEnv \vdash (\lambda_\Multm x.M)\, V : A
  \end{gather}
  Inversion of the top-level subtyping yields
  \begin{gather}
    \label{eq:54}
    \TEnv \vdash (\lambda_\Multm x.M)\, V : B[V/x] \\
    \label{eq:25}
    \UNR{\TEnv} \vdash  {B[V/x]} \Subtype A :\Kind'
  \end{gather}
  By inversion of \eqref{eq:54} using \TirName{Pi-E}
  \begin{gather}
    \label{eq:9}
    \Decompose{\TEnv}{\TEnv_{1}}{\TEnv_{2}} \\
    \label{eq:10}
    \TEnv_{1} \vdash (\lambda_\Multm x.M) : \Pi_\Multn (x:C)B \\
    \label{eq:11}
    \TEnv_{2} \vdash V : C \\
    \label{eq:21}
    \UNR{\TEnv} \vdash B[V/x] : \Kind'
  \end{gather}
  By a lemma of canonical derivations, there is a 
  \TirName{Sub-Type} rule on top of the derivation
  for~\eqref{eq:10}. Its inversion yields
  \begin{gather}
    \label{eq:22}
    \TEnv_{1} \vdash (\lambda_\Multm x.M) : \Pi_\Multm (x:C')B' \\
    \label{eq:23}
    \UNR{\TEnv_{1}} \vdash \Pi_\Multm (x:C')B' \Subtype \Pi_\Multn (x:C)B  :\Kind
  \end{gather}
  Further inversion of (\ref{eq:22}) yields
  \begin{gather}
    \label{eq:16}
    \EnvForm \Multm{ \TEnv_{1}} \\
    \label{eq:13}
    \PlainExtend{\TEnv_{1}} x {C'} \vdash M : B' \\
  \end{gather}
  Inversion of subtyping~(\ref{eq:23}) yields
  \begin{gather}
    \label{eq:12}
    \UNR{\TEnv_{1}} \vdash C \Subtype C' : \Kind_C
    \\
    \label{eq:15}
    \DemoteExtend{\UNR{\TEnv_{1}}} xC \vdash B' \Subtype B  :\Kind_B
    \\
    \label{eq:18}
    \Multm \MultLE \Multn
  \end{gather}
  Taking (\ref{eq:11}) and (\ref{eq:12}) together with
  Lemma~\ref{lemma:properties-un}, we find 
  \begin{gather}
    \label{eq:24}
    \TEnv_{2} \vdash V : C'
  \end{gather}
  In this situation, we apply the substitution
  Lemma~\ref{lemma:substitution} to~\eqref{eq:13} and~\eqref{eq:24}
  and Lemma~\ref{lemma:substitution-subtyping} to~(\ref{eq:15}) and~\eqref{eq:24} to obtain 
  \begin{gather}
    \TEnv \vdash M[V/x] : B'[V/x] \\
    \label{eq:26}
    \UNR{\TEnv} \vdash  B'[V/x] \Subtype  B[V/x] :\Kind_B
  \end{gather}
  Applying \TirName{Sub-Type} twice with ~\eqref{eq:26}
  and~\eqref{eq:25} yields the desired 
  \begin{gather}
    \TEnv \vdash M[V/x] : B[V/x]
    \\
    \TEnv \vdash M[V/x] : A
  \end{gather}


  \textbf{Case }$\ELet{\EPair[]xy}{\EPair V W}{N}
   \ReducesTo
   N[V/x][W/y]$.
   Suppose that 
   \begin{gather}
     \label{eq:19}
     \TEnv \vdash \ELet{\EPair[]xy}{\EPair V W}{N} : C 
   \end{gather}
   Treating the outermost subtyping is trivial, so we directly invert
   \TirName{Sigma-E}:
   \begin{gather}
     \label{eq:30}
     \Decompose\TEnv{\TEnv_1}{\TEnv_2}\\
     \label{eq:32}
     \TEnv_1 \vdash \EPair V W : \Sigma (x:A)B \\
     \label{eq:35}
     \TEnv_2, x: A, y: B \vdash N : C
     \\
     \label{eq:59}
     \UNR{\TEnv} \vdash C : \Kind_C
   \end{gather}
   Inversion of the subtyping on~\eqref{eq:32} yields
   \begin{gather}
     \label{eq:57}
     \TEnv_1 \vdash \EPair V W : \Sigma (x:A')B' \\
     \label{eq:58}
     \UNR{\TEnv_1} \vdash \Sigma (x:A')B' \Subtype
     \Sigma (x:A)B : \Kind
   \end{gather}
  Inversion of~\eqref{eq:57} using \TirName{Sigma-I} yields
  \begin{gather}
    \label{eq:36}
    \Decompose{\TEnv_1}{\TEnv_{11}}{\TEnv_{12}}\\
    \label{eq:37}
    \TEnv_{11} \vdash  V : A' \\
    \label{eq:39}
    \TEnv_{12} \vdash W : B' [V/x] \\
    \label{eq:40}
    \UNR{\TEnv_1} \vdash \Sigma (x:A') B' : \Kind
  \end{gather}
  Inversion of subtyping~(\ref{eq:58}) yields
  \begin{gather}
    \label{eq:31}
    \UNR{\TEnv_1} \vdash A' \Subtype A : \Kind_A \\
    \label{eq:44}
    \DemoteExtend{\UNR{\TEnv_1}} x{A'} \vdash B' \Subtype B :
    \Kind_B
  \end{gather}
  Using (\ref{eq:37}),~(\ref{eq:31}) and Lemma~\ref{lemma:properties-un} we obtain
  \begin{gather}
    \label{eq:38}
    \TEnv_{11} \vdash  V : A
  \end{gather}
  Applying substitution (Lemma~\ref{lemma:substitution}) to
  \eqref{eq:35} and~\eqref{eq:38}  yields (recall the $x\notin\FV (C)$
  by (\ref{eq:59})) 
  \begin{gather}
    \label{eq:41}
    (\DecomposeOp{\TEnv_{11}}{\TEnv_2}), y B[V/x] \vdash N[V/x] : C
  \end{gather}
  Substitution for~(\ref{eq:44}) with~(\ref{eq:38}) yields
  \begin{gather}
    \UNR{\TEnv_1} \vdash B'[V/x] \Subtype B[V/x] :
    \Kind_B
  \end{gather}
  which can be used with~(\ref{eq:39}) to yield
  \begin{gather}
    \label{eq:45}
    \TEnv_{12} \vdash W : B [V/x] 
  \end{gather}
  Applying substitution to~\eqref{eq:41} and~\eqref{eq:45} yields
  \begin{gather}
    \label{eq:42}
    (DecomposeOp{\TEnv_{11}}{\DecomposeOp{\TEnv_2}{\TEnv_{12}}}) \vdash N[V/x][W/y] : C
  \end{gather}
  as $C$ neither contains $x$ nor $y$, combining the environments
  yields  the desired
  \begin{gather}
    \label{eq:43}
    \TEnv \vdash N[V/x][W/y] : C
  \end{gather}
\end{proof}

To prove typing preservation for processes, we adapt the following two lemmas from
\citet{GayVasconcelos2010-jfp}.

\begin{restatable}[Subderivation introduction]{lemma}{SubderivationIntroduction}\label{lemma:subderivation-intro}
  If $\mathcal D$ is a derivation of $\TEnv \vdash \Context{M} : A$
  with $\FV (M) \subseteq \Dom (\TEnv)$,  
  then there are $\TEnv_1$, $\TEnv_2$, and $B$ such that
  $\Decompose\TEnv{\TEnv_1}{\TEnv_2}$, $\mathcal D$ has a
  subderivation $\mathcal D'$ concluding $\TEnv_2 \vdash M:B$, and the
  position of $\mathcal D'$ in $\mathcal D$ corresponds to the
  position of the hole in $\ECN{}$.
\end{restatable}
\begin{proof}
  By induction on  $\ECN$. (Two illustrative cases.)

  \textbf{Case }$\Hole$.
  In this case $\mathcal D' = \mathcal D$, $B=A$, $\TEnv_1 = \UNR{\TEnv}$, and $\TEnv_2 = \TEnv$.

  \textbf{Case }$\ECN\,N$. In this case, inversion on $\mathcal D$ yields
  $\Decompose\TEnv{\TEnv_1}{\TEnv_2}$, and a derivation $\mathcal D_1$
  of $\TEnv_1 \vdash \Context{M} : \Pi_\Multm (x:A')B'$ with
  $\FV (M) \subseteq\Dom (\TEnv_1)$.  Induction yields $\TEnv_1'$,
  $\TEnv_2'$, and $B$ such that
  $\Decompose{\TEnv_1}{\TEnv_1'}{\TEnv_2'}$ and $\mathcal D_1$ has a
  subderivation $\mathcal D'$ concluding $\TEnv_2' \vdash M : B$, and
  the position of $\mathcal D'$ in $\mathcal D_1$ corresponds to the
  position of the hole in $\ECN$.

  We can reassociate the decomposition to
  $\Decompose\TEnv{\TEnv'}{\TEnv_2'}$ and
  $\Decompose{\TEnv'}{\TEnv_1'}{\TEnv_2}$ as it is commutative and associative.
  The claim follows with the $B$ obtained by induction.

  \textbf{Remaining cases.} They all work analogously if decomposition
  is involved. Otherwise, they are straightforward.
\end{proof}

\begin{restatable}[Subderivation elimination]{lemma}{SubderivationElimination}\label{lemma:subderivation-elimination}
  Suppose that
  \begin{enumerate}
  \item\label{item:3} $\Decompose{\TEnv}{\TEnv_1}{\TEnv_2}$,
  \item\label{item:4} $\mathcal D$ is a derivation of $\TEnv \vdash \Context M : A$
    with $\FV (M) \subseteq \Dom (\TEnv)$,
  \item $\mathcal D'$ is a subderivation of $\mathcal D$ concluding
    $\TEnv_2 \vdash M : B$,
  \item the position of $\mathcal D'$ in $\mathcal D$ corresponds to
    the position of the hole in $\ECN$,
  \item $\TEnv_3 \vdash N : B$,
  \item\label{item:5} $\Decompose{\TEnv'}{\TEnv_1}{\TEnv_3}$,
  \end{enumerate}
  then $\TEnv' \vdash \Context N : A$.  
\end{restatable}
\begin{proof}
  By induction on $\ECN$. (Two illustrative cases.)

  \textbf{Case }$\Hole$. Here, $A=B$, $\TEnv_1 = \UNR{\TEnv}$, and $\TEnv_2 = \TEnv$.
  Hence, $\TEnv_3 = \TEnv'$ so that $\TEnv' \vdash N : B$ holds trivially.

  \textbf{Case }$\ECN\,N'$. Inversion on $\mathcal D$ yields
  $\Decompose\TEnv{\TEnv_1'}{\TEnv_2'}$, $A = B'[N'/x]$ and a derivation
  $\mathcal D_1$ for
  $\TEnv_1' \vdash \Context{M} : \Pi_\Multm (x:A')B'$ with
  $\FV (M) \subseteq \Dom (\TEnv_1')$ (item~\ref{item:4}).  From
  $\Decompose\TEnv{\TEnv_1}{\TEnv_2}$, we obtain some $\TEnv_0$ with
  $\Decompose{\TEnv_1'}{\TEnv_0}{\TEnv_2}$ for item~\ref{item:3} and from
  $\Decompose{\TEnv'}{\TEnv_1}{\TEnv_3}$ we obtain
  $\Decompose{\TEnv''}{\TEnv_0}{\TEnv_3}$ for item~\ref{item:5}.

  Induction yields $\TEnv'' \vdash \Context{N}: \Pi_\Multm (x:A')B'$
  so that
  $\Decompose{\TEnv'}{\TEnv''}{\TEnv_2'}$ constructed by applying
  rule \TirName{Pi-E}.
  The result is $\TEnv' \vdash \Context{N}\,N' : B'[N'/x]$ as required. 

  \textbf{Remaining cases.} Similar.
\end{proof}



{\TypingPreservationProcesses*}
\begin{proof}
  \textbf{Case }$\Proc{\Context\NEW} \ReducesTo \NUC\Chanc\Chand
  \Proc{\Context{\EPair[\MultLin]cd}}$.
  Suppose that
  \begin{gather}
    \TEnv \vdash \Proc{\Context\NEW}
  \end{gather}
  Inversion yields
  \begin{gather}
    \TEnv \vdash \Context\NEW : \TUnit
  \end{gather}
  Lemma~\ref{lemma:subderivation-intro} yields
  \begin{gather}
    \Decompose{\TEnv}{\TEnv_1}{\TEnv_2} \\
    \label{eq:46}
    \TEnv_2 \vdash \NEW : S \times \DUAL S 
  \end{gather}
  so that $\EnvForm \MultUn{ \TEnv_2}$.

  Setting 
  \begin{align}
    \TEnv_3 &= \TEnv_2, c : S, d : \DUAL S \\
    \TEnv' & = \TEnv, c : S, d : \DUAL S
  \end{align}
  we apply Lemma~\ref{lemma:subderivation-elimination} and the
  \TirName{Proc-Channel} rule to obtain
  \begin{gather}
    \label{eq:47}
    \TEnv_3 \vdash \EPair[\MultLin]cd : S \times \DUAL S
    \\
    \TEnv' \vdash \Proc{\Context{ \EPair[\MultLin]cd}}
    \\
    \TEnv \vdash \NUC\Chanc\Chand~ \Proc{\Context{ \EPair[\MultLin]\Chanc\Chand}}
  \end{gather}

  \textbf{Case }$
  \Proc{\Context{\FORK M}}n
  \ReducesTo
  \Proc{\Context{\EUnit}} \PAR \Proc{M}
  $. Suppose that 
  \begin{gather}
    \TEnv \vdash \Proc{\Context{\FORK M}}
  \end{gather}
  Inversion yields
  \begin{gather}
    \TEnv \vdash {\Context{\FORK M}} : \TUnit
  \end{gather}
  Lemma~\ref{lemma:subderivation-intro} yields
  \begin{gather}
    \Decompose{\TEnv}{\TEnv_1}{\TEnv_2} \\
    \TEnv_2 \vdash \FORK M : \TUnit
  \end{gather}
  and inversion yields
  \begin{gather}
    \TEnv_2 \vdash M : \TUnit
  \end{gather}
  Hence
  \begin{gather}
    \TEnv_1 \vdash \Context{\EUnit} : \TUnit \\
    \TEnv_1 \vdash \Proc{\Context{\EUnit}}
    \intertext{and}
    \TEnv_2 \vdash \Proc M 
    \intertext{result in (by rule \TirName{Proc-Par})}
    \TEnv \vdash \Proc{\Context{\EUnit}} \PAR \Proc M
  \end{gather}

  \todo[inline]{PJT continue here; check this again; still think that
    some subtyping inversions are ignored}
  \textbf{Case }$
  \begin{array}[t]{ll}
    & \NUC\Chanc\Chand~\Proc{\Context{\SEND\Chanc\,V}} \PAR
    \Proc{\Context[F]{\RECV\Chand}}
    \\
    \ReducesTo
    & \NUC\Chanc\Chand~\Proc{\Context{\Chanc}} \PAR
    \Proc{\Context[F]{\EPair[\MultLin]V\Chand}}
    \text.
  \end{array}
  $

  Suppose that
  \begin{gather}
    \TEnv \vdash
    \NUC\Chanc\Chand~\Proc{\Context{\SEND\Chanc\,V}} \PAR
    \Proc{\Context[F]{\RECV\Chand}}
  \end{gather}
  By inversion of \TirName{Proc-Channel}
  \begin{gather}
    \UNR{\TEnv} \vdash S : \Kind \\
    \TEnv, \Chanc: S, \Chand:\DUAL S
    \vdash \Proc{\Context{\SEND\Chanc\,V}} \PAR
    \Proc{\Context[F]{\RECV\Chand}}
  \end{gather}
  By inversion of \TirName{Proc-Par}
  \begin{gather}
    \Decompose
    {(\TEnv, \Chanc: S, \Chand:\DUAL S)}
    {(\TEnv_1, \Chanc: S)}
    {(\TEnv_2, \Chand: \DUAL S)}\\
    \TEnv_1, \Chanc: S \vdash \Proc{\Context{\SEND\Chanc\,V}} \\
    \TEnv_2, \Chand:\DUAL S
    \vdash \Proc{\Context[F]{\RECV\Chand}} 
  \end{gather}
  By subderivation introduction (Lemma~\ref{lemma:subderivation-intro})
  \begin{gather}
    \Decompose{(\TEnv_1, \Chanc: S)}
    {\TEnv_{11}}{(\TEnv_{12}, \Chanc: S)} \\
    \TEnv_{12}, \Chanc: S \vdash \SEND\Chanc\,V : S'
  \end{gather}
  so that by further inversion
  \begin{gather}
    \label{eq:62}
    \UNR{(\TEnv_{12}, \Chanc: S)} \vdash S \Subtype
    {!(x:A)S'} : \KindSession \\
    B' =S'[V/x] \\
    \UNR{ \TEnv_{12}}, \Chanc: S \vdash \Chanc : {!(x:A)S'} \\
    \TEnv_{12} \vdash V : A
  \end{gather}

  Similarly, by subderivation introduction (Lemma~\ref{lemma:subderivation-intro})
  \begin{gather}
    \Decompose{(\TEnv_2, \Chand:\DUAL S)}
    {\TEnv_{21}}{(\TEnv_{22},  \Chand:\DUAL S)} \\
    (\TEnv_{22},  \Chand:\DUAL S) \vdash \RECV\Chand
    : \Sigma (x:A)\Dual{S'}
  \end{gather}
  so that
  \begin{gather}
    \EnvForm \MultUn{ \TEnv_{22}} \\
    \label{eq:61}
    \TEnv_{22} \vdash \DUAL S
    \Subtype {?(x:A)}\Dual{ S'} : \KindSession \\
    (\TEnv_{22},  \Chand: \DUAL S) \vdash \Chand :
    {?(x:A)}\Dual {S'}
  \end{gather}
  By Lemma~\ref{lemma:subtyping-duality} applied to~(\ref{eq:61})
  and~(\ref{eq:62}), it must be that 
  \begin{gather}
    S = {!(x:A)S'}
  \end{gather}
  
  For the reduct of $\SENDK$, we need
  \begin{gather}
    \UNR{ (\DemoteExtend{\TEnv_{12}, \Chanc: S} x A)} \vdash S' :
    \Kind_C
    \\
    \TEnv_{12}, \Chanc : S'[V/x] \vdash \Chanc : S'[V/x]
  \end{gather}

  For the reduct of $\RECVK$, we need
  \begin{gather}
    (\DecomposeOp{\TEnv_{12}}{\TEnv_{22}}),  \Chand:\Dual {S'[V/x]} \vdash
    \EPair[\MultLin] V \Chand
    : \Sigma (x:A)\DUAL C
  \end{gather}

  Unrolling derivation elimination yields
  \begin{gather}
    \TEnv_1,
    \Chanc: S'
    \vdash \Proc{\Context{\Chanc}}
    \\
    \TEnv_2,
    \Chand:\Dual{S'}
    \vdash \Proc{\Context[F]{\EPair[\MultLin] V \Chand}}
    \\
    \Decompose\TEnv{\TEnv_1}{\TEnv_2}
    \\
    \TEnv,
    \Chanc: S',
    \Chand: \Dual{S'}
    \vdash  \Proc{\Context{\Chanc}} \PAR
    \Proc{\Context[F]{\EPair[\MultLin] V \Chand}}
    \\
    \TEnv,
    \vdash \NUC\Chanc\Chand~( \Proc{\Context{\Chanc}} \PAR
    \Proc{\Context[F]{\EPair[\MultLin] V \Chand}})
  \end{gather}
\end{proof}

{\AbsenceRuntimeErrors*}
\begin{proof}
  A simple case analysis on the processes that constitute errors.
\end{proof}


%% file: proofs-alg-type-check.tex
\section{Proofs for algorithmic type checking}
\label{sec:proofs-alg}

{\AlgorithmicWeakening*}
\begin{proof}
  Mutual rule induction on the various hypotheses. We show the cases
  of items~\ref{item:JAlgTypeSynth} and~\ref{item:JAlgTypeCheck}.

\textbf{Case }$\RuleAName$. We have distinguish two cases, which both
use \TirName{A-Name}.

If $\UNR{\TEnv_1} \vdash A : \Un$, then 
  \begin{equation*}
    \JAlgTypeSynth
    {(\DecomposeOp{(\TEnv_1,z: A,\TEnv_2)}{(\TEnv_1',z: A, \TEnv_2')})}
    M A
    {(\DecomposeOp{(\TEnv_1,z: A, \TEnv_2)}{(\TEnv_1',z: A, \TEnv_2')})}
  \end{equation*}

If $\UNR{\TEnv_1} \vdash A : \Lin$, then 
  \begin{equation*}
    \JAlgTypeSynth
    {(\DecomposeOp{(\TEnv_1,z: A,\TEnv_2)}{(\TEnv_1', \TEnv_2')})}
    M A
    {(\DecomposeOp{(\TEnv_1, \TEnv_2)}{(\TEnv_1', \TEnv_2')})}
  \end{equation*}

  \textbf{Case }$\RuleAUnitI$. Immediate.

  \textbf{Case }$\RuleALabI$. Immediate.

  \textbf{Case }$\RuleALabEA$.

  Properties of context split,
  Lemma~\ref{lemma:decomposition-unrestricted}(\ref{item:split-demote}),
  and
  induction.

  \textbf{Case }$\RuleALabEB$.

  By induction we know that
  \begin{equation*}
    \JAlgTypeSynth
    {(\DecomposeOp{(\TEnv_1, y : \Equation[L] x \ell)}{(\TEnv_2, y : \Equation[L] x \ell)})}
    { N_\ell}{ A_\ell}
    {(\DecomposeOp{(\TOut_\ell, y : \Equation[L] x \ell)}{(\TEnv_2, y : \Equation[L] x \ell)})}
  \end{equation*}
  We can easily show that $\TOut = \TOut_\ell$ implies
  $\DecomposeOp\TOut{\TEnv_2} =
  \DecomposeOp{\TOut_\ell}{\TEnv_2}$. Complete with induction on the
  two remaining premises to the rule, followed by rule \TirName{A-Lab-E2}.
  
  \textbf{Case }$\RuleAPiI$.

  Lemma~\ref{lemma:decomposition-unrestricted}(\ref{item:split-demote})
  and induction gives
  $\JAlgKindSynth{\UNR{(\DecomposeOp{\TEnv_1}{\TEnv_3})}} A{ \Multn}$.
  There are two cases.
  
  If $\Multn = \OccInfty$, then we have
  \begin{equation*}
    \JAlgTypeSynth
    {(\DecomposeOp{(\TEnv_1, x : A)}{(\TEnv_3, x : A)})}
    M  B
    {(\DecomposeOp{(\TEnv_2, x : A)}{(\TEnv_3, x : A)})}
  \end{equation*}

  If $\Multn = \OccOne$, then $x$ is only added to $\TEnv_1$:
  \begin{equation*}
    \JAlgTypeSynth
    {(\DecomposeOp{(\TEnv_1, x : A)}{(\TEnv_3)})}
    M  B
    {(\DecomposeOp{(\TEnv_2, x : A)}{(\TEnv_3)})}
  \end{equation*}
  In both cases, appeal to induction and conclude with rule \textsc{A-Pi-I}.

  \textbf{Case }$\RuleAPiE$.

  Properties of context split,
  Lemma~\ref{lemma:decomposition-unrestricted}(\ref{item:split-demote}),
  and  induction.

  \textbf{Case }$\RuleASigmaI$.

  There are two cases depending on $\Kind$.

  If $\Kind=\Un$, then consider that conditional extension is just
  plain extension and let $\TEnv_4 = \TEnv_4', x : A, z :  \Equation[A]x V$. By induction we have
  \begin{equation*}
    \JAlgTypeSynth
    {(\DecomposeOp{(\TEnv_2, x : A, z : \Equation[A]x V)}{\TEnv_4})}
    {N}{ B}
    {(\DecomposeOp{(\TEnv_3, x : A, z : \Equation[A]x V)}{\TEnv_4})}
  \end{equation*}

  If $\Kind=\Lin$, then conditional extension drops the
  binding and the equality and we choose $\TEnv_4 = \TEnv_4'$:
  \begin{equation*}
    \JAlgTypeSynth
    {(\DecomposeOp{(\TEnv_2, x : A)}{\TEnv_4})}
    {N}{ B}
    {(\DecomposeOp{\TEnv_3}{\TEnv_4})}
  \end{equation*}
  that is
  \begin{equation*}
    \JAlgTypeSynth
    {(\DecomposeOp{\TEnv_2}{\TEnv_4'}, x : A)}
    {N}{ B}
    {(\DecomposeOp{\TEnv_3}{\TEnv_4'})}
  \end{equation*}
  Conclude by straightforward induction on the remaining premises and rule \TirName{A-Sigma-I}.

  \textbf{Case }$\RuleASigmaE$.

  As for rule \textsc{A-Pi-E}.

  \textbf{Case }$\RuleASigmaG$.

  As for rules \textsc{A-Lab-E2} and \textsc{A-Pi-I}.

  The remaining for cases in item~\ref{item:JAlgTypeSynth} follow by a
  straightforward induction.

  \textbf{Case }$\RuleASubType$.
  
  Properties of context split,
  Lemma~\ref{lemma:decomposition-unrestricted}(\ref{item:split-demote}),
  induction, and subtyping weakening,
  Lemma~\ref{lem:AlgorithmicWeakening}.
\end{proof}

{\AlgorithmicLinearStrengthening*}
\begin{proof}
  By mutual rule induction on the hypotheses.
\end{proof}

{\AlgorithmicSubtypingSoundness*}
\begin{proof}
  By mutual rule induction on $\JAlgSubtypeSynth\TEnv B A \Kind$
  and $\JAlgSubtypeCheck\TEnv B A \Kind$.

  \textbf{Case }$\RuleASUnit$.

  The claim follows by this derivation using \TirName{Conv-Refl} and \TirName{Sub-Conv}.
  \begin{mathpar}
    \inferrule*
    {\inferrule*
      {\TEnv \vdash \TUnit : \KindUnSession}
      {\TEnv \vdash \Convertible\TUnit \TUnit : \KindUnSession}
    }
    {\TEnv \vdash \TUnit \Subtype \TUnit : \KindUnSession}
  \end{mathpar}

  \textbf{Case }$\RuleASLabel$.

    Immediate by rule \TirName{Sub-Lab}.

    \textbf{Case }$\RuleASPi$.

    Immediate by induction and then rule \TirName{Sub-Pi}.

    \textbf{Case }$\RuleASSigma$.

  Immediate by induction and then rule \TirName{Sub-Sigma}.

  \textbf{Case }\TirName{AS-Send} and \TirName{AS-Recv}. Analogous to
  the previous cases.

  \textbf{Case }$\RuleASCaseLeftA$.

  There are two cases. From inverting
  $\JAlgConv\TEnv V \ell { L'}$ we obtain that either
  $V=\ell$ (rule \TirName{AC-Refl}) or $V=x$ and $\Equation x \ell$ is
  in $\TEnv$ (rule \TirName{AC-Assoc}).
  
  If $V=\ell$, then we apply \TirName{Conv-Beta} to convert the left
  hand side to $A_\ell$. We conclude by transitivity
  \TirName{Sub-Trans} and the inductive hypothesis for
  $ \TEnv \Mapsto A_\ell \Subtype B : \Kind$.

  If $V=x$ and $\Equation x \ell$ in $\TEnv$, then we need to apply
  \TirName{Conv-Subst}, then transitivity and can conclude as in the
  first case.

  \textbf{Case }$\RuleASCaseLeftB$.

  In this case, we apply eta conversion \TirName{Conv-Eta} to
  introduce a \CASEK{} at the top leve of the right hand side. Then
  \TirName{Sub-Case} becomes applicable to the inductive hypotheses
  arising from
  $ \JAlgSubtypeSynth{\TEnv, y:^\OccInfty \Equation[L] x \ell}{ A_\ell}{ B}{ \Kind_\ell} $.

  \textbf{Case }\TirName{AS-Case-Right1} and \TirName{AS-Case-Right2}
  are analogous.

  \textbf{Case }$\RuleASCheck$. Immediate by the IH and subkinding.
\end{proof}

{\AlgorithmicSubtypingCompleteness*}
\begin{proof}
  All proofs are by rule induction on their hypotheses.
  
  \textbf{1. Kinding}

  Most cases are straightforward, but the \CASEK{} rules merit some
  attention.
  Suppose that $A = \CASE {V} {\overline{\ell: A_\ell}^{\ell\in L}}$.
  There are two cases for $V$, either $V=\ell$ or $V=x$.
  If $V=\ell$, then rules \TirName{AS-Case-Left1} and
  \TirName{AS-Case-Right1} get us to the inductive hypothesis for
  $A_\ell$.
  The same applies if $V=x$ and $\Equation x \ell$ is in $\TEnv$.
  Otherwise, \TirName{AS-Case-Left2} and \TirName{AS-Case-Right2} with
  several subsequent uses of \TirName{AS-Case-Right1} get us to the
  inductive hypotheses for all $A_\ell$.

  \textbf{2. Conversion}

  \textbf{Case }$\RuleConvSym$.

  Given that we have an algorithmic derivation
  $\JAlgSubtype\TEnv A B \Kind$, we would have to convert that
  to a derivation for $\JAlgSubtype\TEnv B A \Kind$.
  Instead, we show for each of the following cases that the subtyping
  judgment is derivable in both directions.

  \textbf{Case  }
  $\RuleConvRefl $.

  Follows by part 1 of this lemma.

  \textbf{Case }$\RuleConvSubst$. By canonical forms it must be $V=\ell' \in L$.

  \textbf{Subcase }left-to-right. Inverting rule
  \TirName{AS-Case-Left1} yields
  $\TEnv \Mapsto A_{\ell'} \le     {\CASE{\ell'}{\overline{\ell:A_\ell}}}
  : \Kind$. Further inverting rule \TirName{AS-Case-Right1} yields
  $\TEnv \Mapsto A_{\ell'} \le A_{\ell'} : \Kind$, which holds by reflexivity.

  \textbf{Subcase }right-to-left. Analogous.

  \textbf{Case }$\RuleConvEta$.

  \textbf{Subcase }left-to-right. Suppose that $\TEnv \Mapsto
  \Convertible x\ell : L$. In this case, we can establish the claim
  by \TirName{AS-Case-Right1} and reflexivity.

  If  $\TEnv \not\Mapsto \Convertible x\ell : L$. Then the claim
  follows by \TirName{AS-Case-Right2} and reflexivity.

  \textbf{Subcase }right-to-left. Analogous, but with the respective
  \TirName{AS-Case-Left} rules.

  \textbf{Case }$\RuleConvBeta$.

  \textbf{Subcase }left-to-right. Immediate by the
  \TirName{AS-Case-Left1} rule.

  \textbf{Subcase }right-to-left. By rule \TirName{AS-Case-Right1}.

  \textbf{3. Subtyping}

  \textbf{Case }$\RuleSubConv$. Immediate by part 2.

  \textbf{Case }$\RuleSubLab$. Immediate by \TirName{AS-Label}.

  \textbf{Case }$\RuleSubTrans$.

  By the IH, there are derivations
  $\Delta_1$ for $\JAlgSubtype\TEnv  A B  \Kind$ and
  $\Delta_2$ for $ \JAlgSubtype \TEnv  B C \Kind$.
  We construct a derivation for $\JAlgSubtype\TEnv A C \Kind$ by lexicographic induction on
  $(\Delta_1, \Delta_2)$. The cases are by the final rules applied in $\Delta_1$ and $\Delta_2$.

  \textbf{Subcase} (\TirName{AS-unit}, R). Use $\Delta_2$ ending in $R$.

  \textbf{Subcase} (\TirName{AS-Label} with $L\subseteq L'$, \TirName{AS-Label} with $L'\subseteq
  L''$). Use \TirName{AS-Label} with $L \subseteq L''$.

  \textbf{Subcase}  (\TirName{AS-Label}, \TirName{AS-Case-Right1}),
  $  \inferrule[AS-Case-Right1]
  { \TEnv \Mapsto \Convertible V \ell : L_V \\
    \TEnv \Mapsto B \Subtype C_\ell : \Kind
  }
  { \TEnv \Mapsto
    B 
    \Subtype 
    \CASE V {\overline{\ell: {C_\ell}}^{\ell\in L}}
    : \Kind
  }
  $.

  By induction, we can find a derivation combining $\JAlgSubtype\TEnv AB\Kind$ and
  $\JAlgSubtype\TEnv B{C_\ell}\Kind$. We apply \TirName{AS-Case-Right1} to conclude.

  \textbf{Subcase} (\TirName{AS-Label}, \TirName{AS-Case-Right2}).

  By induction and weakening (Lemma~\ref{lem:AlgorithmicWeakening}), we can find
  derivations combining $\JAlgSubtype{\TEnv'} AB\Kind$ and 
  $\JAlgSubtype{\TEnv'}B{C_\ell}\Kind$ where $\TEnv' = \TEnv, y:^\OccInfty \Equation x \ell$, for
  each $\ell\in L$. We apply \TirName{AS-Case-Right2} to conclude.

  \textbf{Subcase} (\TirName{AS-Pi}, \TirName{AS-Pi}). By induction for domain and range types,
  transitivity of $\MultLE$, and then reapplying \TirName{AS-Pi}.

  \textbf{Subcase} (\TirName{AS-Pi}, \TirName{AS-Case-RightX}). Similar as with \TirName{AS-Label}.

  \textbf{Subcase} (\TirName{AS-Sigma}, R). Analogous to \TirName{AS-Pi}.

  \textbf{Subcase} (\TirName{AS-Send}, R). Analogous to \TirName{AS-Pi}.

  \textbf{Subcase} (\TirName{AS-Recv}, R). Analogous to \TirName{AS-Pi}.

  \textbf{Subcase} (\TirName{AS-Case-LeftX}, R). Analogous to dealing with \TirName{AS-Case-RightX}
  on the right, except for the case analyzed next where \TirName{AS-Case-Left2} and
  \TirName{AS-Case-Right2} interfere through a case distinction on the same variable. 

  \textbf{Subcase} $  \inferrule[AS-Case-Left2]
  { \TEnv \vdash x : L \\
    \TEnv \not\Mapsto \Convertible x \ell : L' \\
    (\forall\ell\in L)~
    \TEnv, y:^\OccInfty \Equation[L] x \ell \Mapsto A_\ell \Subtype B : \Kind_\ell
  }
  { \TEnv \Mapsto
    \CASE x {\overline{\ell: {A_\ell}}^{\ell\in L''}}
    \Subtype 
    B : \bigsqcup_{\ell\in L}\Kind_\ell
  }
  $

  and
  $  \inferrule[AS-Case-Right2]
  { \TEnv \vdash x : L \\
    \TEnv \not\Mapsto \Convertible x \ell : L' \\
    (\forall\ell\in L)~
    \TEnv, y:^\OccInfty \Equation[L] x \ell \Mapsto B \Subtype C_\ell : \Kind_\ell
  }
  { \TEnv \Mapsto
    B
    \Subtype
    \CASE x {\overline{\ell: {C_\ell}}^{\ell\in L''}}
    : \bigsqcup_{\ell\in L}\Kind_\ell
  }
  $.

  By the IH, we obtain a derivation for $\JAlgSubtype{\TEnv, y:^\OccInfty \Equation[L] x \ell}
  {A_\ell}{C_\ell}{\Kind_\ell}$ and we can construct a derivation for $A \Subtype C$ by applying
  \TirName{AS-Case-Right1} (because $\Equation x \ell$ is known) and then
  \TirName{AS-Case-Left2}. (Or equivalently \TirName{AS-Case-Left1} first and then
  \TirName{AS-Case-Right2}.)

  \textbf{Case }$\RuleSubSub$. Immediate by IH and transitivity of subkinding.

  \textbf{Case }$\RuleSubPi$.

  Immediate by induction.

  \textbf{Case }\TirName{Sub-Sigma} analogous.

  \textbf{Case }$\RuleSubSend$.

  By the IH, we obtain
  \begin{enumerate}\item 
    $\JAlgSubtype\TEnv{ A'}{ A}{\Kind_A^{\Multm_A}}$ with $\Kind_A \Subkind \Kind^\Multm $ and
  \item 
    $\JAlgSubtype{\DemoteExtend\TEnv x{A'}} B{ B'}{\Kind_B}$ with
    $\Kind_B \Subkind \KindSession$.
  \end{enumerate}
  If $\Multm_A = \Multm$, then applying \TirName{AS-Check} to IH(2) implies
  $\JAlgSubtypeCheck{\DemoteExtend\TEnv x{A'}} B{ B'}{\KindSession}$ and we conclude by applying
  \TirName{AS-Send}.
  If $\Multm_A \MultLT \Multm$, then $\Multm_A = \MultUn$ and $\Multm = \MultLin$ and we first have
  to apply weakening Lemma~\ref{lem:AlgorithmicWeakening} to conclude in the same way. 

  \textbf{Case }$\RuleSubRecv$. Analogously.

  \textbf{Case }$\RuleSubCase$.

  By IH we have $\JAlgSubtype {\TEnv, y:^\OccInfty \Equation[L\cap L'] x \ell}{ A_\ell}{ A'_\ell} {
    \Kind_\ell}$ with $\Kind_\ell \Subkind \Kind$, for all $\ell \in L \cap L'$.

  Applying \TirName{AS-Case-Right1} yields
  $\JAlgSubtype {\TEnv, y:^\OccInfty \Equation[L\cap L'] x \ell}{ A_\ell}{ \CASE x {\overline{\ell :
        A'_\ell}^{\ell \in L'}}} { \Kind_\ell}$, for all $\ell \in L \cap L'$.

  Applying \TirName{AS-Case-Left1} yields the desired
  $\JAlgSubtype {\TEnv}{\CASE x{\overline{\ell :  A_\ell}^{\ell \in L}}}{ \CASE x {\overline{\ell :
        A'_\ell}^{\ell \in L \cap L'}}} { \bigsqcup_\ell \Kind_\ell}$ where clearly
  $\bigsqcup_\ell\Kind_\ell \Subkind \Kind$.
\end{proof}

{\AlgorithmicUnfoldingSoundness*}
\begin{proof}
  By rule induction on $\JAlgTypeUnfold\TEnv A B$.

  \textbf{Case }$\RuleAUnfoldType$. Immediate by reflexivity.

  \textbf{Case }$\RuleAUnfoldCaseA$.

  By induction, $\TEnv \vdash
  \Convertible{B_\ell}{B} : \Kind$. Conclude by \TirName{Conv-Beta}
  and transitivity.

  \textbf{Case }$\RuleAUnfoldCaseB$.

  By induction $\TEnv \vdash \Convertible {A_\ell}{\Context[P]{B_\ell}}
  : \Kind$.

  By \TirName{SubCase},
  $\TEnv \vdash \Convertible{
    \CASE x {\overline{\ell : A_\ell}^{\ell \in L}}
  }{
    \CASE x {\overline{\ell : \Context[P]{B_\ell}}^{\ell\in L}}
  } : \Kind$.

  By \TirName{ConvEta}
  $ \TEnv \vdash \Convertible   {\Context[P]{\CASE x {\overline{\ell:
          {B_\ell}}^{\ell\in L}}} }
  {\CASE x {\overline{\ell: \Context[P]{\CASE x {\overline{\ell:
              {B_\ell}}^{\ell\in L}}}}^{\ell\in L}}} : \Kind$.

  Repeated use (i.e., transitivity) of \TirName{ConvSubst} and \TirName{ConvBeta} on the
  inner cases yields

 $    \CASE x {\overline{\ell : \Context[P]{B_\ell}}^{\ell\in L}}$.
  Conclude by symmetry and transitivity.
\end{proof}

{\AlgorithmicUnfoldingCompleteness*}
\begin{proof}
  We consider the case where $\ECN[P] = {?(y:A)}\Hole$n.

  By completeness of subtyping
  (Lemma~\ref{lem:AlgorithmicSubtypingCompleteness}), we know that  
  \begin{enumerate}\item\label{item:1} 
    $\JAlgSubtypeSynth \TEnv A {\Context[P]B} \Kind'$ and
  \item\label{item:2} 
    $\JAlgSubtypeSynth \TEnv {\Context[P]B} A \Kind''$.
  \end{enumerate}
  We argue by rule induction on the derivations of these judgments.
  
  For the judgment~\ref{item:1}, the only applicable rules are
  \TirName{AS-Recv}, \TirName{AS-Case-Left1}, or
  \TirName{AS-Case-Left2}.

  In case of \TirName{AS-Recv}, (only) the same rule is applicable to
  judgment~\ref{item:2}, so we obtain $\TEnv, y:^\Occp A \vdash
  \Convertible{B'} B : \Kind'''$, for some $B'$, which proves the claim.

  In case of \TirName{AS-Case-Left1}, (only) the dual rule
  \TirName{AS-Case-Right1} is applicable to judgment~\ref{item:2}, so
  the claim holds by induction using the selected case branch and the
  rule \TirName{A-Unfold-Case1}.

  In case of \TirName{AS-Case-Left2}, (only) the dual rule
  \TirName{AS-Case-Right2} is applicable to judgment~\ref{item:2}, so
  the induction hypotheses for all labels support rule
  \TirName{A-Unfold-Case2}, which establishes the claim.
\end{proof}

{\AlgorithmicKindingSoundness*}
\begin{proof}
  By mutual induction.

  \textbf{Case }$\RuleAUnitF$. Immediate.

  \textbf{Case }$\RuleALabF$. Immediate.

  \textbf{Case }$\RuleALabET$.

  By the IH for algorithmic typing soundness, we obtain that $\TEnv \vdash V : L$.
  By the IH for algorithmic kinding synthesis, we obtain, for each $\ell\in L$, that
  $\TEnv, x :^\OccInfty \Equation[L] V \ell \vdash A_\ell : K_\ell$ and by \TirName{Sub-Kind} that
  $\TEnv, x :^\OccInfty \Equation[L] V \ell \vdash A_\ell \shortrightarrow \bigsqcup K_\ell$.
  We conclude by \TirName{Lab-E'}.

  \textbf{Case }$\RuleAPiF$.

  By the IH for algorithmic kinding synthesis, we obtain
  $\TEnv \vdash A : \Kind^\Multn $ and
  $\DemoteExtend\TEnv x A \vdash B : \Kind'$.
  We conclude by \TirName{Pi-F}.

  \textbf{Case }$\RuleASigmaF$.

  By the IH for algorithmic kinding synthesis, we obtain
  $\TEnv \vdash A : \Kind^\Multn $ and
  $\DemoteExtend\TEnv x A \vdash B : {\Kind'}^{\Multn'}$.
  We conclude by \TirName{Sigma-F}.

  \textbf{Case }$\RuleASsnOutF$.

  By the IH for algorithmic kinding synthesis and checking, we obtain
  $    \TEnv \vdash A : \Kind^\Multn $ and
  $\DemoteExtend\TEnv x A \vdash B : \KindSession$. We conclude by \TirName{Ssn-Out-F}.

  \textbf{Case }$\RuleASsnInF$.

  By the IH for algorithmic kinding synthesis and checking, we obtain
  $    \TEnv \vdash A : \Kind^\Multn $ and
  $    \DemoteExtend\TEnv x A \vdash B : \KindSession$.
  We conclude by \TirName{Ssn-In-F}.

  \textbf{Case (Checking) }$\RuleASubKind$.

  By the IH for algorithmic kinding synthesis, we obtain
  $  {\TEnv \vdash A :  \Kind }$.
  We conclude by \TirName{Sub-Kind} using $ \Kind \Subkind \Kind'$.
\end{proof}

{\AlgorithmicKindingCompleteness*}
\begin{proof}
  By rule induction on the derivation of  $\TEnv \vdash A : \Kind$. In each case, we just consider
  synthesis; the claim for
  checking follows from the result for synthesis by application of the \TirName{A-Sub-Kind} rule.

  \textbf{Case }$\RuleUnitF$.  Immediate.

  \textbf{Case }$\RuleLabF$. Immediate.

  \textbf{Case }$\RuleLabET$.
  
  By induction, $(\forall\ell\in L)~\TEnv, x :^\OccInfty \Equation[L] V \ell \shortrightarrow A_\ell
  : K_\ell$ with $K_\ell \Subkind K$. Hence $\bigsqcup K_\ell \Subkind K$, so the claim holds by
  \TirName{A-Lab-E'}.

  \textbf{Case }$\RulePiF$. 

  By the IH for the first premise, we obtain $\JAlgKindSynth\TEnv A{\Kind_A}$ with $\Kind_A^{\Multn_A}
  \Subkind \Kind^\Multn$.

  If $\Multn_A=\Multn$, then we obtain 
  $\JAlgKindSynth{\DemoteExtend\TEnv  x A} B {\Kind_B^{\Multn_B}}$
  with $\Kind_B \Subkind  {\Kind'}^{\Multn'}$ from the second premise and conclude by the synthesis rule
  \TirName{A-Pi-F}. The checking part holds by reflexivity of subkinding.
  
  Otherwise, if $\Multn_A \MultLT \Multn$, then it must be that $\Multn_A = \MultUn$ and
  $\Multn = \MultLin$.  In this case, conditional extension ensures that $x\notin \FV (B)$ cannot be
  used in the derivation of $B$. Hence, by ``weakening'', there is also a derivation for
  $\TEnv, x: A \vdash B : {\Kind'}$, for which induction yields a corresponding
  algorithmic synthesis, with which we can conclude again with \TirName{A-Pi-F} and reflexivity
  of subkinding.

  \textbf{Case }$\RuleSigmaF$.

  By the IH for the first premise, we obtain $\JAlgKindSynth\TEnv A{\Kind_A^{\Multn_A}}$ with $\Kind_A
  \Subkind \Kind^\Multn$.

  By the IH for the second premise, we obtain
  $\JAlgKindSynth{\DemoteExtend\TEnv  x A} B {\Kind_B^{\Multn_B}}$
  with $\Kind_B \Subkind  {\Kind'}^{\Multn'}$.

  If $\Multn_A = \Multn$, then we conclude immediately with \TirName{A-Sigma-F} where transitivity
  of subkinding guarantees $\Multn_A\MultLE \Multn \MultLE \Multm$ and
  $\Multn_B \MultLE \Multn' \MultLE \Multm$.

  If $\Multn_A \MultLT \Multn$, we reason analogously to the case for \TirName{Pi-F}.

  \textbf{Case }$\RuleSsnOutF$.

  By the IH, we have that
  $   \JAlgKindSynth \TEnv A \Kind_A$ where $\Kind_A^{\Multn_A} \Subkind \Kind^\Multn $ and
  $    \JAlgKindCheck{\DemoteExtend\TEnv x A} B \KindSession $.
  The remaining reasoning is analogous to the case for \TirName{Pi-F}.
  
  \textbf{Case }$\RuleSsnInF$. Analogous to the case for \TirName{Ssn-Out-F}.

  \textbf{Case }$\RuleSubKind$.

  By the IH for synthesis, we have $\JAlgKindSynth\TEnv A {\Kind_A}$ with $\Kind_A \Subkind
  \Kind$. By transitivity of subkinding, $\Kind_A \Subkind \Kind'$, which proves the claim for synthesis.
\end{proof}

{\AlgorithmicSoundness*}
\begin{proof}
  To get a viable induction hypothesis we need to generalize the two
  statements to
  \begin{enumerate}
  \item For all $\TOut$,
    if $\JAlgTypeSynth {\DecomposeOp\TEnv\TOut} M A {\TOut}$, then $\TEnv \vdash M : A$.
  \item For all $\TOut$,
    if $\JAlgTypeCheck {\DecomposeOp\TEnv\TOut} M A \TOut$, then $\TEnv \vdash M : A$. 
  \end{enumerate}
  
  The proof of the two parts is by mutual rule induction on the
  sequents in the hypotheses. Most cases are straightforward.
  When the derivation ends with rule \textsc{A-Lab-E} we use the
  Agreement properties (lemma~\ref{lem:agreement}).
  When the derivation ends with rule \textsc{A-Pi-E} or with rule
  \textsc{A-Sigma-I} we use Algorithmic Kinding Soundness and
  Algorithmic Linear Strengthening
  (lemmas~\ref{lem:AlgorithmicKindingSoundness}
  and~\ref{lem:AlgorithmicLinearStrengthening}).

  We detail the case for rule \TirName{A-Sigma-E}.
  \begin{gather}
    \RuleASigmaE
  \end{gather}
  We need to show that for a splitting $\TEnv_1 =
  \DecomposeOp{\TEnv}{\TEnv_3} $ we finally obtain $ {\TEnv} \vdash {
    \ELet{\EPair[]xy}{M}{N}} : C$.
  For the first premise, induction yields that $\TEnv_1 =
  \DecomposeOp{\TEnv_1'}{\TEnv_2}$ and $\TEnv_1' \vdash M : D$.
  The soundness theorem for unfolding yields that $\UNR{\TEnv_2}
  \vdash \Convertible D {{\Sigma(x: A)}B}$. By
  Lemma~\ref{lemma:decomposition-preserves-wellformedness} we find
  that $\UNR{\TEnv_1'} = \UNR{\TEnv_2}$, so that $\UNR{\TEnv_1'}
  \vdash \Convertible D {{\Sigma(x: A)}B}$. Thus by conversion, 
  $\TEnv_1' \vdash M : {{\Sigma(x: A)}B}$.

  To pick one of four possible cases as an example, assume that
  $A:\Un$ and $B:\Lin$, which means that the binding for $y$ is used
  up in $N$. Hence, considering the premise for $N$ yields $\TEnv_2 ,
  x:A, y:B = \DecomposeOp{\TEnv_2'}{\TEnv_3}, x:A, y:B =
  \DecomposeOp{(\TEnv_2', x:A, y:B)}{(\TEnv_3, x:A)}$.
  The IH now yields $\TEnv_2', x:A, y:B \vdash N : C$.

  Putting the splittings together, we obtain that $\TEnv_1 =
  \DecomposeOp{(\DecomposeOp{\TEnv_1'}{\TEnv_2'})}{\TEnv_3}$ so that
  we can assume  $\TEnv = \DecomposeOp{\TEnv_1'}{\TEnv_2'}$.

  It remains to apply rule \TirName{Sigma-E} using $\TEnv_1'$ and
  $\TEnv_2'$ in place of $\TEnv_1$ and $\TEnv_2$ in the rule:
  \begin{gather}
    \label{eq:20}
    \RuleSigmaE
  \end{gather}

  When the derivation ends with rule \textsc{A-Sub-Type} use
  Algorithmic Subtyping Soundness
  (Lemma~\ref{lem:AlgorithmicSubtypingSoundness}), part (1) of this
  theorem, and rule \textsc{Sub-Type}.
\end{proof}

{\AlgorithmicCompleteness*}
\begin{proof}
  By mutual rule induction on the hypothesis. In each case, the second claim follows from the first
  by rule \TirName{A-Sub-Type}.

  \textbf{Case }$\RuleSubType$.

  Immediate by induction and transitivity of subtyping.

  \textbf{Case }$\RuleName$. 
  Immediate.

  \textbf{Case }$\RuleUnitI$.
  Immediate.

  \textbf{Case }$\RuleLabI$. Immediate.

  \textbf{Case }$\RuleLabE$.

  Induction yields
  \begin{enumerate}
  \item 
    $ \JAlgTypeCheck{\UNR\TEnv}{V}{L}{\UNR\TEnv}$
  \item 
    $\JAlgTypeSynth{\TEnv, y :^\OccInfty \Equation[L]{V}{\ell}}{N_\ell}{B_\ell}{\UNR\TEnv, y :^\OccInfty \Equation[L]{V}{\ell}}$ with
    $\UNR\TEnv, y :^\OccInfty \Equation[L]{V}{\ell} \vdash B_\ell \Subtype B : \Kind$
  \end{enumerate}

  If $\JAlgConv{\UNR\TEnv} V \ell{ L'}$, then 
  $\JAlgTypeSynth{\TEnv}{N_\ell}{B_\ell}{\UNR\TEnv}$ with
  $\UNR\TEnv \vdash B_\ell \Subtype B : \Kind$ is also derivable.
  Applying \TirName{A-Lab-E1} yields
  $\JAlgTypeSynth{\TEnv}{\CASE V {\overline{\ell:N_\ell}}}{B_\ell}{\UNR\TEnv}$, with establishes the
  claim.

  Otherwise, if $\not\exists\ell$ such that  $\JAlgConv{\UNR\TEnv} V \ell{ L'}$,
  then $V=x$ and \TirName{A-Lab-E2} is applicable, but it derives the judgment
  $$\JAlgTypeSynth{\TEnv}{\CASE x {\overline{\ell:N_\ell}^{\ell\in L}}}{\CASE x
    {\overline{\ell:B_\ell}^{\ell \in L}}}{\UNR\TEnv}$$
  The computed type is a subtype of $B$, by \TirName{Conv-Eta} and \TirName{Sub-Case} and the
  returned type environment is $\UNR\TEnv$, as required.

  \textbf{Case }$\RulePiI \quad \Big(\RuleAPiI\Big)$.

  By induction $\JAlgTypeSynth{\TEnv, x : A} M {B'}{\UNR {(\TEnv, x:
      A)}}$ with 
  $\UNR {(\TEnv, x : A)} \vdash B' \Subtype B : \Kind_B$.
  We have that $\UNR{(\TEnv, x:A)} = \DemoteExtend{\UNR\TEnv} x A$.
  
  If $\Kind=\Un$, then $\UNR\TEnv = \TEnv$. Hence, we can apply \TirName{A-Pi-I} to obtain
  $\JAlgTypeSynth{\TEnv}{\lambda_\Multm (x: A).M}{ \Pi_\Multm
    (x:  A)B'}{\UNR \TEnv}$ where $\UNR\TEnv \vdash \Pi_\Multm
  (x:  A)B'  \Subtype \Pi_\Multm    (x:  A)B : \Kind$.

  \textbf{Case }$\RulePiE$.

  By induction $\JAlgTypeSynth{\TEnv_1} MC{\UNR\TEnv_1}$
  and $\UNR\TEnv_1 \vdash C \Subtype \Pi_\Multm (x: A)B : \Kind$.

  By induction $\JAlgTypeCheck{\TEnv_2} N {A}{\UNR\TEnv_2}$.

  To conclude, we use Algorithmic Weakening (Lemma~\ref{lem:AlgorithmicWeakening}), and properties
  \ref{item:split-omega} and \ref{item:split-commutative} of context
  split (Lemma~\ref{lemma:decomposition-unrestricted}) to apply
  \TirName{A-Pi-E}.


  \textbf{Case }$\RuleSigmaI$.

  By induction $\JAlgTypeSynth{\TEnv_1}{  V}{ A'}{ \UNR{\TEnv_1}} $
  and $\UNR{\TEnv_1} \vdash A' \Subtype A : \Kind$.

  By induction $\JAlgTypeSynth{\DemoteExtend{\PlainExtend{\TEnv_2} x
      {A'}} z {\Equation[A]x V}} N {B'} {
    \DemoteExtend{\DemoteExtend{\TEnv_3} x {A'}} z { \Equation[A]x V}
  }$
  and $\UNR{(\DemoteExtend{\DemoteExtend{\TEnv_3} x {A'}} z {
      \Equation[A]x V})}
  = \DemoteExtend{\DemoteExtend{\UNR{\TEnv_3}} x{A'}} z{\Equation[A]xV}
  \vdash B' \Subtype B : \Kindn$.

  Again using Algorithmic Weakening and properties of context split,
  we can apply \TirName{A-Sigma-I}.

  
  \textbf{Case }$\RuleSigmaE$.

  Follows by reasoning analogous to \TirName{Pi-E}.


  The remaining cases are straightforward variations of the
  demonstrated techniques.
\end{proof}

The following lemmas document the assumptions on the inputs of 
various algorithmic judgments. For algorithmic subtyping
$\JAlgSubtypeSynth \TEnv A B \Kind$, we assume 
that $\EnvForm\Un\TEnv$, $\TEnv \vdash A : \Kind_A$, and $\TEnv \vdash B
: \Kind_B$. 
\begin{lemma}
  If $\TEnv \vdash A : \Kind$
  and $\JAlgTypeUnfold \TEnv A B$,
  then $\TEnv \vdash B : \Kind$.
\end{lemma}
\begin{proof}
  By rule induction on the unfolding judgment.
\end{proof}
\begin{lemma} ~
  \begin{enumerate}\item 
    If $\EnvForm\Lin\TEnv$ and $\JAlgTypeSynth \TEnv M A \TOut$, then
    $\UNR\TEnv \vdash A : \Kind$ and $\EnvForm\Lin\TOut$.
  \item 
    If $\EnvForm\Lin\TEnv$
    and $\UNR\TEnv \vdash A : \Kind$
    and $\JAlgTypeCheck \TEnv M A \TOut$, then
    $\EnvForm\Lin\TOut$.
  \end{enumerate}
\end{lemma}


%% file: proofs-embedding-gv.tex
\section{Proofs for Embedding {\LSST} into {\LDGV}}
\label{sec:proofs-embedding-gv}

\begin{restatable}[Subtyping Preservation]{lemma}{SimulationSubtypingPreservation}
  If $\vdash_\LSST A \le B$ and $\EnvForm \Multm {_{GV}A}$, then
  $\vdash_\LDGV \Embed A \le \Embed B : \Kind$.
\end{restatable}
\begin{proof}
  The proof is by mutual induction on the derivations of {\LSST} subtyping judgments for
  types $T$ and for session types $S$. The main cases to check are the ones for branch and
  choice types for session types.
  
  \textbf{Case} choice types: Suppose that $\vdash_\LSST \oplus \{\overline{\ell: S_\ell}^{\ell\in L}\} \Subtype
  \oplus \{\overline{\ell: S'_\ell}^{\ell\in L'}\}
  $. By inversion we obtain
  $L' \subseteq L$ and
  $(\forall\ell\in L')~ \vdash_\LSST S_\ell \Subtype S'_\ell$.
  Induction yields $\EmptyEnv \vdash_\LDGV \Embed{S_\ell} \Subtype \Embed{S'_\ell} : \KindSession$. 

  To obtain the goal
  $\EmptyEnv \vdash_\LDGV
  {!(x : \LabelSet)}. \CASE x {\overline{\ell: \Embed{S_\ell}}
  }
  \Subtype
  {!(x : \LabelSet')}. \CASE x {\overline{\ell: \Embed{S'_\ell}}
  }
  $ we first need   to obtain the communication part, that is, we need that $\EmptyEnv
  \vdash_\LDGV \LabelSet'\Subtype \LabelSet : \Kind^\Multl$ which holds because $L'
  \subseteq L$.
  It remains to establish that
  $x:^{\OccInfty} \LabelSet'\vdash_\LDGV
  \CASE x {\overline{\ell: \Embed{S_\ell}}
  }
  \Subtype
  \CASE x {\overline{\ell: \Embed{S'_\ell}}
  } : \KindSession
  $ which holds by rule \TirName{Sub-Case}.

  \textbf{Case} branch types: suppose that
  $\vdash_\LSST \& \{\overline{\ell: S_\ell}^{\ell\in L}\} \Subtype
  \& \{\overline{\ell: S'_\ell}^{\ell\in L'}\}
  $. By inversion we obtain $L\subseteq L'$ and $(\forall\ell\in L)~
  \vdash_\LSST S_\ell \Subtype S'_\ell$. Induction yields $\EmptyEnv \vdash_\DGV \Embed{S_\ell} \Subtype \Embed{S'_\ell} : \KindSession$. 
  By weakening, we obtain $x:^{\OccInfty} \LabelSet, z:^{\OccInfty} \Equation x \ell \vdash_\LDGV \Embed{S_\ell} \Subtype \Embed{S'_\ell} : \KindSession$, for each $\ell\in L$, and hence by rule \TirName{Sub-Case}
  $x:^{\OccInfty} \LabelSet \vdash_\DGV
  \CASE x {\overline{\ell: \Embed{S_\ell}}
  }
  \Subtype
  \CASE x {\overline{\ell: \Embed{S'_\ell}}
  } : \KindSession
  $. From $L \subseteq L'$ we obtain
  $\EmptyEnv \vdash_\LDGV \LabelSet \Subtype \LabelSet': \Kind^\Multl$ and then by
  \TirName{Sub-Recv} that 
  $\EmptyEnv \vdash_\LDGV
  {?(x:\LabelSet).}
  \CASE x {\overline{\ell: \Embed{S_\ell}}
  }
  \Subtype
  {?(x:\LabelSet).}
  \CASE x {\overline{\ell: \Embed{S'_\ell}}
  } : \KindSession$, which is the desired result.
\end{proof}

\begin{lemma}
  If $\TEnv \vdash_\LSST M : S$, then $\UNR{\Embed\Gamma} \vdash_\DGV
  \Embed S : \KindSession$.
\end{lemma}

{\SimulationTypingPreservation*}

\begin{proof}
  The proof is by induction on the respective {\LSST} typing judgments. The interesting cases are the ones for $\SELECTK$, $\RCASEK$, $\CLOSEK$, and $\WAITK$.

  \textbf{Case \SELECTK}. Suppose that $\TEnv \vdash_\LSST \SELECT \ell : T$.
  Inversion yields that $\EnvForm{ \MultUn}{_\LSST \TEnv}$ and $T =\oplus \{ \overline{\ell: S_\ell}
  \} \to_\MultUn S_\ell$ so that

  $\Embed T = {!(x : \LabelSet)}.
  \CASE x {\overline{\ell: \Embed{S_\ell}}
  } \to_\MultUn \Embed{S_\ell} $.
  For the embedding, we consider $\Embed\TEnv \vdash_\LDGV \LAM z \APP{\SEND z} \ell : A$ and find by inversion the requirement
  \begin{gather}
    \begin{array}[t]{l}
    \Embed\TEnv, z : {!(x : \LabelSet)}.
  \CASE x {\overline{\ell: \Embed{S_\ell}}
      }
      \\
      \vdash_\LDGV \APP{\SEND z} \ell : \Embed{S_\ell}
    \end{array}
  \end{gather}
  By conversion, we obtain
  \begin{gather}
    \begin{array}[t]{l}
    \Embed\TEnv, z : {!(x : \LabelSet)}.
  \CASE x {\overline{\ell: \Embed{S_\ell}}
      }
      \\
      \vdash_\LDGV \APP{\SEND z} \ell :
      \CASE \ell {\overline{\ell: \Embed{S_\ell}}
      }
    \end{array}
  \end{gather}
  Further inversion yields
  \begin{gather}
    \begin{array}[t]{l}
      \Embed\TEnv, z : {!(x : \LabelSet)}.
      \CASE x {\overline{\ell: \Embed{S_\ell}}
      }
      \\
      \vdash_\LDGV {\SEND z} : \Pi_\MultLin (x : \LabelSet).
      \CASE x {\overline{\ell: \Embed{S_\ell}}
      }
    \end{array}
  \end{gather}
  This concludes the typing derivation.

  \textbf{Case \RCASEK}.
  Suppose $\TEnv \vdash_\LSST \RCASE e {\overline{\ell : y.f_\ell}^{\ell\in L}} : T$.
  Inversion yields
  \begin{gather}
    \Decompose\TEnv{\TEnv_1}{\TEnv_2}
    \\
    \TEnv_1 \vdash_\LSST e :  \&\{\overline{\ell: S_\ell}^{\ell\in L}\}
    \\
    (\forall \ell\in L)~ \TEnv_2, y:S_\ell \vdash_\LSST f_\ell : T
  \end{gather}
  By induction
  \begin{gather}
    \Embed{\TEnv_1} \vdash_\LDGV \Embed e :  \Embed{ \&\{\overline{\ell: S_\ell}^{\ell\in L}\} }
    \\
    \Embed{\TEnv_2}, y:\Embed{S_\ell} \vdash_\LDGV \Embed{f_\ell} : \Embed T
  \end{gather}
  Hence
  \begin{gather}
    \Embed{\TEnv_1} \vdash_\LDGV \Embed e :  {?(x : \LabelSet)}. \CASE x {\overline{\ell: \Embed{S_\ell}}^{\ell\in L}
    }
    \\
    \begin{array}[t]{l}
    \Embed{\TEnv_2}, x:\LabelSet, y:\CASE x {\overline{\ell: \Embed{S_\ell}}^{\ell\in L}
      }, z : \Equation x \ell
      \vdash_\LDGV \Embed{f_\ell} : \Embed T
    \end{array}
  \end{gather}
  by the subderivation using the conversion rule
  \begin{mathpar}
    \inferrule
    {
      \Embed{\TEnv_2}, x:\LabelSet, y:\CASE x {\overline{\ell: \Embed{S_\ell}}
      }, z : \Equation x \ell
      \vdash_\LDGV y : \CASE x {\overline{\ell: \Embed{S_\ell}}^{\ell\in L}
      }
    }
    {
      \Embed{\TEnv_2}, x:\LabelSet, y:\CASE x {\overline{\ell: \Embed{S_\ell}}^{\ell\in L}
      }, z : \Equation x \ell
      \vdash_\LDGV y : \Embed{S_\ell}
    }
  \end{mathpar}
  Hence
  \begin{gather}
    \begin{array}[t]{l}
      \Embed{\TEnv_1}
      \vdash_\LDGV
      \RECV{ \Embed e} :  {\Sigma(x : \LabelSet)}. \CASE x {\overline{\ell: \Embed{S_\ell}}
    }
    \end{array}
    \\
    \begin{array}[t]{l}
      \Embed{\TEnv_2}, x:\LabelSet, y:\CASE x {\overline{\ell: \Embed{S_\ell}}^{\ell\in L}
      }, z:\Equation x \ell 
      \vdash_\LDGV
       \Embed{f_\ell} :
       \Embed T
    \end{array}
    \\
    \begin{array}[t]{l}
      \Embed{\TEnv_2}, x:\LabelSet, y:\CASE x {\overline{\ell: \Embed{S_\ell}}^{\ell\in L}
      }
      \vdash_\LDGV
      \CASE x { \overline{ \ell: \Embed{f_\ell}}} :
      \Embed T
    \end{array}
    \\
    \Embed\TEnv \vdash_\LDGV
    \begin{array}[t]{@{}l}
      \ELet{\EPair[]{x}{y}} {\RECV{\Embed{e}}}
      {\CASE x {\overline{\ell:
      \Embed{f_\ell}}}}
      : \Embed T
    \end{array}
  \end{gather}
  which is the desired embedding.
\end{proof}

{\CoSimulation*}

\begin{proof}
  First, we observe that evaluation contexts are compatible with the translation. That is, if $\ECN$ is an evaluation context in {\LSST}, then $\Embed\ECN$ is an evaluation context in {\LDGV}. The proof is by induction on evaluation contexts.

  Hence, it is sufficient to consider the images of the redexes under the translation. Of these, only the image of $\RCASEK$ communicating with $\SELECTK$ is interesting. 

  \begin{gather*}
    \NUC\Chanc\Chand(
    \Context{
      \Embed{\RCASE \Chanc {\overline{\ell : y.f_\ell}}}}
    \PAR
    \Context[F]{
      \Embed{\SELECT\ell d}})
    \\
    =
    \\
    \NUC\Chanc\Chand(
    \Context{
      \begin{array}[t]{l}
        \ELet{\EPair[]{x}{y}} {\RECV{\Chanc}}
        {\CASE x {\overline{\ell:
        \Embed{f_\ell}}}}
      \end{array}
    }
    \PAR
    \Context[F]{
      {\SEND\ell d}})
    \\
    \ReducesTo
    \\
    \NUC\Chanc\Chand(
    \Context{
      \begin{array}[t]{l}
        \ELet{\EPair[]{x}{y}} {\EPair[\MultLin]\ell{\Chanc}}
        {\CASE x {\overline{\ell:
        \Embed{f_\ell}}}}
      \end{array}
    }
    \PAR
    \Context[F]{d})
    \\
    \ReducesTo
    \\
    \NUC\Chanc\Chand(
    \Context{
        {\CASE \ell {\overline{\ell:
              \Embed{f_\ell[\Chanc/y]}}}}
    }
    \PAR
    \Context[F]{d})
    \\
    \ReducesTo
    \\
    \NUC\Chanc\Chand(
    \Context{
      \Embed{f_\ell[\Chanc/y]}
    }
    \PAR
    \Context[F]{d})
  \end{gather*}
  The last line corresponds to the reduct of the {\RCASEK} redex in \LSST.

  The communication between $\CLOSEK$ and $\WAITK$ translates in a similar way. 
\end{proof}


%% file: popl19-dgv.bbl

\begin{thebibliography}{62}


\ifx \showCODEN    \undefined \def \showCODEN     #1{\unskip}     \fi
\ifx \showDOI      \undefined \def \showDOI       #1{#1}\fi
\ifx \showISBNx    \undefined \def \showISBNx     #1{\unskip}     \fi
\ifx \showISBNxiii \undefined \def \showISBNxiii  #1{\unskip}     \fi
\ifx \showISSN     \undefined \def \showISSN      #1{\unskip}     \fi
\ifx \showLCCN     \undefined \def \showLCCN      #1{\unskip}     \fi
\ifx \shownote     \undefined \def \shownote      #1{#1}          \fi
\ifx \showarticletitle \undefined \def \showarticletitle #1{#1}   \fi
\ifx \showURL      \undefined \def \showURL       {\relax}        \fi
\providecommand\bibfield[2]{#2}
\providecommand\bibinfo[2]{#2}
\providecommand\natexlab[1]{#1}
\providecommand\showeprint[2][]{arXiv:#2}

\bibitem[\protect\citeauthoryear{Ahman, Fournet, Hritcu, Maillard, Rastogi, and
  Swamy}{Ahman et~al\mbox{.}}{2018}]%
        {DBLP:journals/pacmpl/AhmanFHMRS18}
\bibfield{author}{\bibinfo{person}{Danel Ahman}, \bibinfo{person}{C{\'{e}}dric
  Fournet}, \bibinfo{person}{Catalin Hritcu}, \bibinfo{person}{Kenji Maillard},
  \bibinfo{person}{Aseem Rastogi}, {and} \bibinfo{person}{Nikhil Swamy}.}
  \bibinfo{year}{2018}\natexlab{}.
\newblock \showarticletitle{Recalling a Witness: Foundations and Applications
  of Monotonic State}.
\newblock \bibinfo{journal}{\emph{{PACMPL}}} \bibinfo{volume}{2},
  \bibinfo{number}{{POPL}} (\bibinfo{year}{2018}),
  \bibinfo{pages}{65:1--65:30}.
\newblock


\bibitem[\protect\citeauthoryear{Aspinall and Compagnoni}{Aspinall and
  Compagnoni}{2001}]%
        {AspinallCompagnoni2001}
\bibfield{author}{\bibinfo{person}{David Aspinall} {and}
  \bibinfo{person}{Adriana~B. Compagnoni}.} \bibinfo{year}{2001}\natexlab{}.
\newblock \showarticletitle{Subtyping Dependent Types}.
\newblock \bibinfo{journal}{\emph{Theoretical Computer Science}}
  \bibinfo{volume}{266}, \bibinfo{number}{1-2} (\bibinfo{year}{2001}),
  \bibinfo{pages}{273--309}.
\newblock
\urldef\tempurl%
\url{https://doi.org/10.1016/S0304-3975(00)00175-4}
\showDOI{\tempurl}


\bibitem[\protect\citeauthoryear{Atkey}{Atkey}{2018}]%
        {DBLP:conf/lics/Atkey18}
\bibfield{author}{\bibinfo{person}{Robert Atkey}.}
  \bibinfo{year}{2018}\natexlab{}.
\newblock \showarticletitle{Syntax and Semantics of Quantitative Type Theory}.
  In \bibinfo{booktitle}{\emph{{LICS}}}. \bibinfo{publisher}{{ACM}},
  \bibinfo{pages}{56--65}.
\newblock


\bibitem[\protect\citeauthoryear{Baltazar, Mostrous, and Vasconcelos}{Baltazar
  et~al\mbox{.}}{2012}]%
        {DBLP:journals/corr/abs-1211-4099}
\bibfield{author}{\bibinfo{person}{Pedro Baltazar}, \bibinfo{person}{Dimitris
  Mostrous}, {and} \bibinfo{person}{Vasco~Thudichum Vasconcelos}.}
  \bibinfo{year}{2012}\natexlab{}.
\newblock \showarticletitle{Linearly Refined Session Types}. In
  \bibinfo{booktitle}{\emph{{LINEARITY}}} \emph{(\bibinfo{series}{{EPTCS}})},
  Vol.~\bibinfo{volume}{101}. \bibinfo{pages}{38--49}.
\newblock


\bibitem[\protect\citeauthoryear{Bernardy, Boespflug, Newton, {Peyton Jones},
  and Spiwack}{Bernardy et~al\mbox{.}}{2018}]%
        {DBLP:journals/pacmpl/BernardyBNJS18}
\bibfield{author}{\bibinfo{person}{Jean{-}Philippe Bernardy},
  \bibinfo{person}{Mathieu Boespflug}, \bibinfo{person}{Ryan~R. Newton},
  \bibinfo{person}{Simon {Peyton Jones}}, {and} \bibinfo{person}{Arnaud
  Spiwack}.} \bibinfo{year}{2018}\natexlab{}.
\newblock \showarticletitle{Linear {Haskell}: Practical Linearity in a
  Higher-Order Polymorphic Language}.
\newblock \bibinfo{journal}{\emph{{PACMPL}}} \bibinfo{volume}{2},
  \bibinfo{number}{{POPL}} (\bibinfo{year}{2018}), \bibinfo{pages}{5:1--5:29}.
\newblock


\bibitem[\protect\citeauthoryear{Bonelli, Compagnoni, and Gunter}{Bonelli
  et~al\mbox{.}}{2004}]%
        {DBLP:journals/entcs/BonelliCG04}
\bibfield{author}{\bibinfo{person}{Eduardo Bonelli},
  \bibinfo{person}{Adriana~B. Compagnoni}, {and} \bibinfo{person}{Elsa~L.
  Gunter}.} \bibinfo{year}{2004}\natexlab{}.
\newblock \showarticletitle{Correspondence Assertions for Process
  Synchronization in Concurrent Communications}.
\newblock \bibinfo{journal}{\emph{Electr. Notes Theor. Comput. Sci.}}
  \bibinfo{volume}{97} (\bibinfo{year}{2004}), \bibinfo{pages}{175--195}.
\newblock
\urldef\tempurl%
\url{https://doi.org/10.1016/j.entcs.2004.04.036}
\showDOI{\tempurl}


\bibitem[\protect\citeauthoryear{Brady}{Brady}{2013}]%
        {DBLP:journals/jfp/Brady13}
\bibfield{author}{\bibinfo{person}{Edwin Brady}.}
  \bibinfo{year}{2013}\natexlab{}.
\newblock \showarticletitle{Idris, A General-Purpose Dependently Typed
  Programming Language: Design and Implementation}.
\newblock \bibinfo{journal}{\emph{J. Funct. Program.}} \bibinfo{volume}{23},
  \bibinfo{number}{5} (\bibinfo{year}{2013}), \bibinfo{pages}{552--593}.
\newblock
\urldef\tempurl%
\url{https://doi.org/10.1017/S095679681300018X}
\showDOI{\tempurl}


\bibitem[\protect\citeauthoryear{Brady}{Brady}{2017}]%
        {DBLP:journals/aghcs/Brady17}
\bibfield{author}{\bibinfo{person}{Edwin Brady}.}
  \bibinfo{year}{2017}\natexlab{}.
\newblock \showarticletitle{Type-driven Development of Concurrent Communicating
  Systems}.
\newblock \bibinfo{journal}{\emph{Computer Science {(AGH)}}}
  \bibinfo{volume}{18}, \bibinfo{number}{3} (\bibinfo{year}{2017}).
\newblock
\urldef\tempurl%
\url{https://doi.org/10.7494/csci.2017.18.3.1413}
\showDOI{\tempurl}


\bibitem[\protect\citeauthoryear{Caires and Pfenning}{Caires and
  Pfenning}{2010}]%
        {CairesPfenning2010}
\bibfield{author}{\bibinfo{person}{Lu\'{i}s Caires} {and}
  \bibinfo{person}{Frank Pfenning}.} \bibinfo{year}{2010}\natexlab{}.
\newblock \showarticletitle{Session Types as Intuitionistic Linear
  Propositions}. In \bibinfo{booktitle}{\emph{CONCUR}}
  \emph{(\bibinfo{series}{LNCS})}, Vol.~\bibinfo{volume}{6269}.
  \bibinfo{publisher}{Springer}, \bibinfo{address}{Paris, France},
  \bibinfo{pages}{222--236}.
\newblock
\showISBNx{978-3-642-15374-7}


\bibitem[\protect\citeauthoryear{Caires, Pfenning, and Toninho}{Caires
  et~al\mbox{.}}{2016}]%
        {DBLP:journals/mscs/CairesPT16}
\bibfield{author}{\bibinfo{person}{Lu{\'{\i}}s Caires}, \bibinfo{person}{Frank
  Pfenning}, {and} \bibinfo{person}{Bernardo Toninho}.}
  \bibinfo{year}{2016}\natexlab{}.
\newblock \showarticletitle{Linear logic propositions as session types}.
\newblock \bibinfo{journal}{\emph{Mathematical Structures in Computer Science}}
  \bibinfo{volume}{26}, \bibinfo{number}{3} (\bibinfo{year}{2016}),
  \bibinfo{pages}{367--423}.
\newblock
\urldef\tempurl%
\url{https://doi.org/10.1017/S0960129514000218}
\showDOI{\tempurl}


\bibitem[\protect\citeauthoryear{Casinghino, Sj{\"{o}}berg, and
  Weirich}{Casinghino et~al\mbox{.}}{2014}]%
        {DBLP:conf/popl/CasinghinoSW14}
\bibfield{author}{\bibinfo{person}{Chris Casinghino}, \bibinfo{person}{Vilhelm
  Sj{\"{o}}berg}, {and} \bibinfo{person}{Stephanie Weirich}.}
  \bibinfo{year}{2014}\natexlab{}.
\newblock \showarticletitle{Combining Proofs and Programs in a Dependently
  Typed Language}. In \bibinfo{booktitle}{\emph{POPL}},
  \bibfield{editor}{\bibinfo{person}{Suresh Jagannathan} {and}
  \bibinfo{person}{Peter Sewell}} (Eds.). \bibinfo{publisher}{ACM},
  \bibinfo{pages}{33--46}.
\newblock
\showISBNx{978-1-4503-2544-8}
\urldef\tempurl%
\url{https://doi.org/10.1145/2535838.2535883}
\showDOI{\tempurl}


\bibitem[\protect\citeauthoryear{Castagna, Dezani-Ciancaglini, Giachino, and
  Padovani}{Castagna et~al\mbox{.}}{2009}]%
        {CastagnaDezaniCiancagliniGiachinoPadovani2009}
\bibfield{author}{\bibinfo{person}{Giuseppe Castagna},
  \bibinfo{person}{Mariangiola Dezani-Ciancaglini}, \bibinfo{person}{Elena
  Giachino}, {and} \bibinfo{person}{Luca Padovani}.}
  \bibinfo{year}{2009}\natexlab{}.
\newblock \showarticletitle{Foundations of Session Types}. In
  \bibinfo{booktitle}{\emph{Principles and Practice of Declarative Programming,
  {PPDP} 2009}}, \bibfield{editor}{\bibinfo{person}{Ant{\'o}nio Porto} {and}
  \bibinfo{person}{Francisco~J. L{\'o}pez-Fraguas}} (Eds.).
  \bibinfo{publisher}{ACM}, \bibinfo{address}{Coimbra, Portugal},
  \bibinfo{pages}{219--230}.
\newblock
\showISBNx{978-1-60558-568-0}


\bibitem[\protect\citeauthoryear{Castagna, Petrucciani, and Nguyen}{Castagna
  et~al\mbox{.}}{2016}]%
        {DBLP:conf/icfp/CastagnaP016}
\bibfield{author}{\bibinfo{person}{Giuseppe Castagna}, \bibinfo{person}{Tommaso
  Petrucciani}, {and} \bibinfo{person}{Kim Nguyen}.}
  \bibinfo{year}{2016}\natexlab{}.
\newblock \showarticletitle{Set-Theoretic Types for Polymorphic Variants}. In
  \bibinfo{booktitle}{\emph{{ICFP}}}. \bibinfo{publisher}{{ACM}},
  \bibinfo{pages}{378--391}.
\newblock


\bibitem[\protect\citeauthoryear{Cervesato and Pfenning}{Cervesato and
  Pfenning}{1996}]%
        {DBLP:conf/lics/CervesatoP96}
\bibfield{author}{\bibinfo{person}{Iliano Cervesato} {and}
  \bibinfo{person}{Frank Pfenning}.} \bibinfo{year}{1996}\natexlab{}.
\newblock \showarticletitle{A Linear Logical Framework}. In
  \bibinfo{booktitle}{\emph{{LICS}}}. \bibinfo{publisher}{{IEEE} Computer
  Society}, \bibinfo{pages}{264--275}.
\newblock


\bibitem[\protect\citeauthoryear{Chen, Dezani{-}Ciancaglini, Scalas, and
  Yoshida}{Chen et~al\mbox{.}}{2017}]%
        {DBLP:journals/lmcs/ChenDSY17}
\bibfield{author}{\bibinfo{person}{Tzu{-}Chun Chen},
  \bibinfo{person}{Mariangiola Dezani{-}Ciancaglini}, \bibinfo{person}{Alceste
  Scalas}, {and} \bibinfo{person}{Nobuko Yoshida}.}
  \bibinfo{year}{2017}\natexlab{}.
\newblock \showarticletitle{On the Preciseness of Subtyping in Session Types}.
\newblock \bibinfo{journal}{\emph{Logical Methods in Computer Science}}
  \bibinfo{volume}{13}, \bibinfo{number}{2} (\bibinfo{year}{2017}).
\newblock


\bibitem[\protect\citeauthoryear{{Dal Lago} and Gaboardi}{{Dal Lago} and
  Gaboardi}{2011}]%
        {DBLP:journals/corr/abs-1104-0193}
\bibfield{author}{\bibinfo{person}{Ugo {Dal Lago}} {and} \bibinfo{person}{Marco
  Gaboardi}.} \bibinfo{year}{2011}\natexlab{}.
\newblock \showarticletitle{Linear Dependent Types and Relative Completeness}.
\newblock \bibinfo{journal}{\emph{Logical Methods in Computer Science}}
  \bibinfo{volume}{8}, \bibinfo{number}{4} (\bibinfo{year}{2011}).
\newblock


\bibitem[\protect\citeauthoryear{{Dal Lago} and Petit}{{Dal Lago} and
  Petit}{2012}]%
        {DBLP:conf/ppdp/LagoP12}
\bibfield{author}{\bibinfo{person}{Ugo {Dal Lago}} {and}
  \bibinfo{person}{Barbara Petit}.} \bibinfo{year}{2012}\natexlab{}.
\newblock \showarticletitle{Linear Dependent Types in a Call-By-Value
  Scenario}. In \bibinfo{booktitle}{\emph{{PPDP}}}. \bibinfo{publisher}{{ACM}},
  \bibinfo{pages}{115--126}.
\newblock


\bibitem[\protect\citeauthoryear{Dardha, Giachino, and Sangiorgi}{Dardha
  et~al\mbox{.}}{2012}]%
        {DBLP:conf/ppdp/DardhaGS12}
\bibfield{author}{\bibinfo{person}{Ornela Dardha}, \bibinfo{person}{Elena
  Giachino}, {and} \bibinfo{person}{Davide Sangiorgi}.}
  \bibinfo{year}{2012}\natexlab{}.
\newblock \showarticletitle{Session Types Revisited}. In
  \bibinfo{booktitle}{\emph{{PPDP}}}. \bibinfo{publisher}{{ACM}},
  \bibinfo{pages}{139--150}.
\newblock


\bibitem[\protect\citeauthoryear{Dezani-Ciancaglini, Drossopoulou, Mostrous,
  and Yoshida}{Dezani-Ciancaglini et~al\mbox{.}}{2009}]%
        {DezaniCiancagliniDrossopoulouMostrousYoshida2009}
\bibfield{author}{\bibinfo{person}{Mariangiola Dezani-Ciancaglini},
  \bibinfo{person}{Sophia Drossopoulou}, \bibinfo{person}{Dimitris Mostrous},
  {and} \bibinfo{person}{Nobuko Yoshida}.} \bibinfo{year}{2009}\natexlab{}.
\newblock \showarticletitle{Objects and Session Types}.
\newblock \bibinfo{journal}{\emph{Information and Computation}}
  \bibinfo{volume}{207}, \bibinfo{number}{5} (\bibinfo{year}{2009}),
  \bibinfo{pages}{595--641}.
\newblock


\bibitem[\protect\citeauthoryear{Dunfield and Krishnaswami}{Dunfield and
  Krishnaswami}{2013}]%
        {DBLP:conf/icfp/DunfieldK13}
\bibfield{author}{\bibinfo{person}{Joshua Dunfield} {and}
  \bibinfo{person}{Neelakantan~R. Krishnaswami}.}
  \bibinfo{year}{2013}\natexlab{}.
\newblock \showarticletitle{Complete and easy bidirectional typechecking for
  higher-rank polymorphism}. In \bibinfo{booktitle}{\emph{{ICFP}}}.
  \bibinfo{publisher}{{ACM}}, \bibinfo{pages}{429--442}.
\newblock


\bibitem[\protect\citeauthoryear{Ferreira and Pientka}{Ferreira and
  Pientka}{2014}]%
        {DBLP:conf/ppdp/FerreiraP14}
\bibfield{author}{\bibinfo{person}{Francisco Ferreira} {and}
  \bibinfo{person}{Brigitte Pientka}.} \bibinfo{year}{2014}\natexlab{}.
\newblock \showarticletitle{Bidirectional Elaboration of Dependently Typed
  Programs}. In \bibinfo{booktitle}{\emph{{PPDP}}}. \bibinfo{publisher}{{ACM}},
  \bibinfo{pages}{161--174}.
\newblock


\bibitem[\protect\citeauthoryear{Garrigue}{Garrigue}{1998}]%
        {Garrigue1998}
\bibfield{author}{\bibinfo{person}{Jacques Garrigue}.}
  \bibinfo{year}{1998}\natexlab{}.
\newblock \showarticletitle{Programming with Polymorphic Variants}. In
  \bibinfo{booktitle}{\emph{In ACM Workshop on ML}}.
\newblock


\bibitem[\protect\citeauthoryear{Gay and Hole}{Gay and Hole}{2005}]%
        {GayHole2005}
\bibfield{author}{\bibinfo{person}{Simon~J. Gay} {and} \bibinfo{person}{Malcolm
  Hole}.} \bibinfo{year}{2005}\natexlab{}.
\newblock \showarticletitle{Subtyping for Session Types in the Pi Calculus}.
\newblock \bibinfo{journal}{\emph{Acta Informatica}} \bibinfo{volume}{42},
  \bibinfo{number}{2-3} (\bibinfo{year}{2005}), \bibinfo{pages}{191--225}.
\newblock


\bibitem[\protect\citeauthoryear{Gay and Vasconcelos}{Gay and
  Vasconcelos}{2010}]%
        {GayVasconcelos2010-jfp}
\bibfield{author}{\bibinfo{person}{Simon~J. Gay} {and}
  \bibinfo{person}{Vasco~Thudichum Vasconcelos}.}
  \bibinfo{year}{2010}\natexlab{}.
\newblock \showarticletitle{Linear Type Theory for Asynchronous Session Types}.
\newblock \bibinfo{journal}{\emph{J. Funct. Program.}} \bibinfo{volume}{20},
  \bibinfo{number}{1} (\bibinfo{year}{2010}), \bibinfo{pages}{19--50}.
\newblock


\bibitem[\protect\citeauthoryear{Gay, Vasconcelos, Ravara, Gesbert, and
  Caldeira}{Gay et~al\mbox{.}}{2010}]%
        {GayVasconcelosRavaraGesbertCaldeira2010}
\bibfield{author}{\bibinfo{person}{Simon~J. Gay}, \bibinfo{person}{Vasco~T.
  Vasconcelos}, \bibinfo{person}{Ant\'{o}nio Ravara}, \bibinfo{person}{Nils
  Gesbert}, {and} \bibinfo{person}{Alexandre~Z. Caldeira}.}
  \bibinfo{year}{2010}\natexlab{}.
\newblock \showarticletitle{Modular Session Types for Distributed
  Object-Oriented Programming}, See \citeN{DBLP:conf/popl/2010},
  \bibinfo{pages}{299--312}.
\newblock
\showISBNx{978-1-60558-479-9}
\urldef\tempurl%
\url{https://doi.org/10.1145/1706299.1706335}
\showDOI{\tempurl}


\bibitem[\protect\citeauthoryear{Goto, Jagadeesan, Jeffrey, Pitcher, and
  Riely}{Goto et~al\mbox{.}}{2016}]%
        {DBLP:journals/mscs/GotoJJPR16}
\bibfield{author}{\bibinfo{person}{Matthew~A. Goto}, \bibinfo{person}{Radha
  Jagadeesan}, \bibinfo{person}{Alan Jeffrey}, \bibinfo{person}{Corin Pitcher},
  {and} \bibinfo{person}{James Riely}.} \bibinfo{year}{2016}\natexlab{}.
\newblock \showarticletitle{An Extensible Approach to Session Polymorphism}.
\newblock \bibinfo{journal}{\emph{Mathematical Structures in Computer Science}}
  \bibinfo{volume}{26}, \bibinfo{number}{3} (\bibinfo{year}{2016}),
  \bibinfo{pages}{465--509}.
\newblock


\bibitem[\protect\citeauthoryear{Harper}{Harper}{2016}]%
        {Harper2016-book}
\bibfield{author}{\bibinfo{person}{Robert Harper}.}
  \bibinfo{year}{2016}\natexlab{}.
\newblock \bibinfo{booktitle}{\emph{Practical Foundations for Programming
  Languages} (\bibinfo{edition}{second} ed.)}.
\newblock \bibinfo{publisher}{Cambridge University Press}.
\newblock


\bibitem[\protect\citeauthoryear{Honda}{Honda}{1993}]%
        {Honda1993}
\bibfield{author}{\bibinfo{person}{Kohei Honda}.}
  \bibinfo{year}{1993}\natexlab{}.
\newblock \showarticletitle{Types for Dyadic Interaction}. In
  \bibinfo{booktitle}{\emph{Proceedings of 4th International Conference on
  Concurrency Theory}} \emph{(\bibinfo{series}{LNCS})},
  \bibfield{editor}{\bibinfo{person}{Eike Best}} (Ed.).
  \bibinfo{publisher}{Springer}, \bibinfo{pages}{509--523}.
\newblock


\bibitem[\protect\citeauthoryear{Honda, Mukhamedov, Brown, Chen, and
  Yoshida}{Honda et~al\mbox{.}}{2011}]%
        {hondamukhamedovbrownchenyoshida2011}
\bibfield{author}{\bibinfo{person}{Kohei Honda}, \bibinfo{person}{Aybek
  Mukhamedov}, \bibinfo{person}{Gary Brown}, \bibinfo{person}{Tzu-Chun Chen},
  {and} \bibinfo{person}{Nobuko Yoshida}.} \bibinfo{year}{2011}\natexlab{}.
\newblock \showarticletitle{Scribbling Interactions with a Formal Foundation}.
  In \bibinfo{booktitle}{\emph{ICDCIT 2011}} \emph{(\bibinfo{series}{LNCS})},
  Vol.~\bibinfo{volume}{6536}. \bibinfo{publisher}{Springer},
  \bibinfo{address}{Bhubaneshwar, India}, \bibinfo{pages}{55--75}.
\newblock
\showISBNx{978-3-642-19055-1}


\bibitem[\protect\citeauthoryear{Honda, Vasconcelos, and Kubo}{Honda
  et~al\mbox{.}}{1998}]%
        {HondaVasconcelosKubo1998}
\bibfield{author}{\bibinfo{person}{Kohei Honda},
  \bibinfo{person}{Vasco~Thudichum Vasconcelos}, {and} \bibinfo{person}{Makoto
  Kubo}.} \bibinfo{year}{1998}\natexlab{}.
\newblock \showarticletitle{Language Primitives and Type Discipline for
  Structured Communication-Based Programming}. In
  \bibinfo{booktitle}{\emph{Proc.~7th ESOP}} \emph{(\bibinfo{series}{LNCS})},
  \bibfield{editor}{\bibinfo{person}{Chris Hankin}} (Ed.),
  Vol.~\bibinfo{volume}{1381}. \bibinfo{publisher}{Springer},
  \bibinfo{address}{Lisbon, Portugal}, \bibinfo{pages}{122--138}.
\newblock


\bibitem[\protect\citeauthoryear{Honda, Yoshida, and Carbone}{Honda
  et~al\mbox{.}}{2008}]%
        {HondaYoshidaCarbone2008}
\bibfield{author}{\bibinfo{person}{Kohei Honda}, \bibinfo{person}{Nobuko
  Yoshida}, {and} \bibinfo{person}{Marco Carbone}.}
  \bibinfo{year}{2008}\natexlab{}.
\newblock \showarticletitle{Multiparty Asynchronous Session Types}. In
  \bibinfo{booktitle}{\emph{Proc.~35th {ACM} Symp. POPL}},
  \bibfield{editor}{\bibinfo{person}{Phil Wadler}} (Ed.).
  \bibinfo{publisher}{{ACM} Press}, \bibinfo{address}{San Francisco, CA, USA},
  \bibinfo{pages}{273--284}.
\newblock
\showISBNx{978-1-59593-689-9}


\bibitem[\protect\citeauthoryear{Honda, Yoshida, and Carbone}{Honda
  et~al\mbox{.}}{2016}]%
        {DBLP:journals/jacm/HondaYC16}
\bibfield{author}{\bibinfo{person}{Kohei Honda}, \bibinfo{person}{Nobuko
  Yoshida}, {and} \bibinfo{person}{Marco Carbone}.}
  \bibinfo{year}{2016}\natexlab{}.
\newblock \showarticletitle{Multiparty Asynchronous Session Types}.
\newblock \bibinfo{journal}{\emph{J. {ACM}}} \bibinfo{volume}{63},
  \bibinfo{number}{1} (\bibinfo{year}{2016}), \bibinfo{pages}{9:1--9:67}.
\newblock
\urldef\tempurl%
\url{https://doi.org/10.1145/2827695}
\showDOI{\tempurl}


\bibitem[\protect\citeauthoryear{Hu, Yoshida, and Honda}{Hu
  et~al\mbox{.}}{2008}]%
        {HuYoshidaHonda2008}
\bibfield{author}{\bibinfo{person}{Raymond Hu}, \bibinfo{person}{Nobuko
  Yoshida}, {and} \bibinfo{person}{Kohei Honda}.}
  \bibinfo{year}{2008}\natexlab{}.
\newblock \showarticletitle{Session-Based Distributed Programming in {Java}}.
  In \bibinfo{booktitle}{\emph{22nd ECOOP}} \emph{(\bibinfo{series}{LNCS})},
  \bibfield{editor}{\bibinfo{person}{Jan Vitek}} (Ed.),
  Vol.~\bibinfo{volume}{5142}. \bibinfo{publisher}{Springer},
  \bibinfo{address}{Paphos, Cyprus}, \bibinfo{pages}{516--541}.
\newblock
\showISBNx{978-3-540-70591-8}


\bibitem[\protect\citeauthoryear{Igarashi, Thiemann, Vasconcelos, and
  Wadler}{Igarashi et~al\mbox{.}}{2017}]%
        {igarashithiemannvasconceloswadler2017}
\bibfield{author}{\bibinfo{person}{Atsushi Igarashi}, \bibinfo{person}{Peter
  Thiemann}, \bibinfo{person}{Vasco~T. Vasconcelos}, {and}
  \bibinfo{person}{Philip Wadler}.} \bibinfo{year}{2017}\natexlab{}.
\newblock \showarticletitle{Gradual Session Types}.
\newblock \bibinfo{journal}{\emph{Proc. ACM Program. Lang.}}
  \bibinfo{volume}{1}, \bibinfo{number}{ICFP}, Article \bibinfo{articleno}{38}
  (\bibinfo{date}{Sept.} \bibinfo{year}{2017}), \bibinfo{numpages}{28}~pages.
\newblock
\urldef\tempurl%
\url{https://doi.org/10.1145/3110282}
\showDOI{\tempurl}


\bibitem[\protect\citeauthoryear{Kobayashi}{Kobayashi}{2002}]%
        {DBLP:conf/unu/Kobayashi02}
\bibfield{author}{\bibinfo{person}{Naoki Kobayashi}.}
  \bibinfo{year}{2002}\natexlab{}.
\newblock \showarticletitle{Type Systems for Concurrent Programs}. In
  \bibinfo{booktitle}{\emph{10th Anniversary Colloquium of {UNU/IIST}}}
  \emph{(\bibinfo{series}{Lecture Notes in Computer Science})},
  Vol.~\bibinfo{volume}{2757}. \bibinfo{publisher}{Springer},
  \bibinfo{pages}{439--453}.
\newblock


\bibitem[\protect\citeauthoryear{Kobayashi, Pierce, and Turner}{Kobayashi
  et~al\mbox{.}}{1996}]%
        {KobayashiPierceTurner1996}
\bibfield{author}{\bibinfo{person}{Naoki Kobayashi},
  \bibinfo{person}{Benjamin~C. Pierce}, {and} \bibinfo{person}{David~N.
  Turner}.} \bibinfo{year}{1996}\natexlab{}.
\newblock \showarticletitle{Linearity and the pi-calculus}. In
  \bibinfo{booktitle}{\emph{Proc. 1996 {ACM} Symp. POPL}}.
  \bibinfo{publisher}{ACM Press}, \bibinfo{address}{St. Petersburg Beach, FL,
  USA}, \bibinfo{pages}{358--371}.
\newblock
\showISBNx{0-89791-769-3}


\bibitem[\protect\citeauthoryear{Krishnaswami, Pradic, and Benton}{Krishnaswami
  et~al\mbox{.}}{2015}]%
        {DBLP:conf/popl/KrishnaswamiPB15}
\bibfield{author}{\bibinfo{person}{Neelakantan~R. Krishnaswami},
  \bibinfo{person}{Pierre Pradic}, {and} \bibinfo{person}{Nick Benton}.}
  \bibinfo{year}{2015}\natexlab{}.
\newblock \showarticletitle{Integrating Linear and Dependent Types}. In
  \bibinfo{booktitle}{\emph{{POPL}}}. \bibinfo{publisher}{{ACM}},
  \bibinfo{pages}{17--30}.
\newblock


\bibitem[\protect\citeauthoryear{Lindley and Morris}{Lindley and
  Morris}{2014}]%
        {DBLP:journals/corr/LindleyM14}
\bibfield{author}{\bibinfo{person}{Sam Lindley} {and}
  \bibinfo{person}{J.~Garrett Morris}.} \bibinfo{year}{2014}\natexlab{}.
\newblock \showarticletitle{Sessions as Propositions}. In
  \bibinfo{booktitle}{\emph{Proceedings 7th Workshop on Programming Language
  Approaches to Concurrency and Communication-cEntric Software, {PLACES} 2014,
  Grenoble, France, 12 April 2014.}} \emph{(\bibinfo{series}{{EPTCS}})},
  \bibfield{editor}{\bibinfo{person}{Alastair~F. Donaldson} {and}
  \bibinfo{person}{Vasco~T. Vasconcelos}} (Eds.), Vol.~\bibinfo{volume}{155}.
  \bibinfo{pages}{9--16}.
\newblock
\urldef\tempurl%
\url{https://doi.org/10.4204/EPTCS.155.2}
\showDOI{\tempurl}


\bibitem[\protect\citeauthoryear{Lindley and Morris}{Lindley and
  Morris}{2016}]%
        {DBLP:conf/icfp/LindleyM16}
\bibfield{author}{\bibinfo{person}{Sam Lindley} {and}
  \bibinfo{person}{J.~Garrett Morris}.} \bibinfo{year}{2016}\natexlab{}.
\newblock \showarticletitle{Talking Bananas: Structural Recursion for Session
  Types}. In \bibinfo{booktitle}{\emph{{ICFP}}}. \bibinfo{publisher}{{ACM}},
  \bibinfo{pages}{434--447}.
\newblock


\bibitem[\protect\citeauthoryear{Mazurak and Zdancewic}{Mazurak and
  Zdancewic}{2010}]%
        {DBLP:conf/icfp/MazurakZ10}
\bibfield{author}{\bibinfo{person}{Karl Mazurak} {and} \bibinfo{person}{Steve
  Zdancewic}.} \bibinfo{year}{2010}\natexlab{}.
\newblock \showarticletitle{Lolliproc: to concurrency from classical linear
  logic via curry-howard and control}. In \bibinfo{booktitle}{\emph{{ICFP}}}.
  \bibinfo{publisher}{{ACM}}, \bibinfo{pages}{39--50}.
\newblock


\bibitem[\protect\citeauthoryear{McBride}{McBride}{2016}]%
        {DBLP:conf/birthday/McBride16}
\bibfield{author}{\bibinfo{person}{Conor McBride}.}
  \bibinfo{year}{2016}\natexlab{}.
\newblock \showarticletitle{I Got Plenty o' Nuttin'}. In
  \bibinfo{booktitle}{\emph{A List of Successes That Can Change the World -
  Essays Dedicated to Philip Wadler on the Occasion of His 60th Birthday}}
  \emph{(\bibinfo{series}{LNCS})}, \bibfield{editor}{\bibinfo{person}{Sam
  Lindley}, \bibinfo{person}{Conor McBride}, \bibinfo{person}{Philip~W.
  Trinder}, {and} \bibinfo{person}{Donald Sannella}} (Eds.),
  Vol.~\bibinfo{volume}{9600}. \bibinfo{publisher}{Springer},
  \bibinfo{pages}{207--233}.
\newblock
\showISBNx{978-3-319-30935-4}
\urldef\tempurl%
\url{https://doi.org/10.1007/978-3-319-30936-1_12}
\showDOI{\tempurl}


\bibitem[\protect\citeauthoryear{Morris and McKinna}{Morris and
  McKinna}{2019}]%
        {DBLP:journals/pacmpl/MorrisM19}
\bibfield{author}{\bibinfo{person}{J.~Garrett Morris} {and}
  \bibinfo{person}{James McKinna}.} \bibinfo{year}{2019}\natexlab{}.
\newblock \showarticletitle{Abstracting Extensible Data Types: or, Rows by Any
  Other Name}.
\newblock \bibinfo{journal}{\emph{{PACMPL}}} \bibinfo{volume}{3},
  \bibinfo{number}{{POPL}} (\bibinfo{year}{2019}),
  \bibinfo{pages}{12:1--12:28}.
\newblock
\urldef\tempurl%
\url{https://dl.acm.org/citation.cfm?id=3290325}
\showURL{%
\tempurl}


\bibitem[\protect\citeauthoryear{Nishimura}{Nishimura}{1998}]%
        {Nishimura1998}
\bibfield{author}{\bibinfo{person}{Susumu Nishimura}.}
  \bibinfo{year}{1998}\natexlab{}.
\newblock \showarticletitle{Static Typing for Dynamic Messages}. In
  \bibinfo{booktitle}{\emph{Proc.~25th {ACM} Symp. POPL}},
  \bibfield{editor}{\bibinfo{person}{Luca Cardelli}} (Ed.).
  \bibinfo{publisher}{ACM Press}, \bibinfo{address}{San Diego, CA, USA},
  \bibinfo{pages}{266--278}.
\newblock
\showISBNx{0-89791-979-3}
\urldef\tempurl%
\url{https://doi.org/10.1145/268946.268968}
\showDOI{\tempurl}


\bibitem[\protect\citeauthoryear{Padovani}{Padovani}{2017a}]%
        {DBLP:conf/esop/Padovani17}
\bibfield{author}{\bibinfo{person}{Luca Padovani}.}
  \bibinfo{year}{2017}\natexlab{a}.
\newblock \showarticletitle{Context-Free Session Type Inference}. In
  \bibinfo{booktitle}{\emph{{ESOP}}} \emph{(\bibinfo{series}{Lecture Notes in
  Computer Science})}, Vol.~\bibinfo{volume}{10201}.
  \bibinfo{publisher}{Springer}, \bibinfo{pages}{804--830}.
\newblock


\bibitem[\protect\citeauthoryear{Padovani}{Padovani}{2017b}]%
        {DBLP:journals/jfp/Padovani17}
\bibfield{author}{\bibinfo{person}{Luca Padovani}.}
  \bibinfo{year}{2017}\natexlab{b}.
\newblock \showarticletitle{A Simple Library Implementation of Binary
  Sessions}.
\newblock \bibinfo{journal}{\emph{J. Funct. Program.}}  \bibinfo{volume}{27}
  (\bibinfo{year}{2017}), \bibinfo{pages}{e4}.
\newblock
\urldef\tempurl%
\url{https://doi.org/10.1017/S0956796816000289}
\showDOI{\tempurl}


\bibitem[\protect\citeauthoryear{Pierce and Turner}{Pierce and Turner}{2000}]%
        {PierceTurner2000-toplas}
\bibfield{author}{\bibinfo{person}{Benjamin~C. Pierce} {and}
  \bibinfo{person}{David~N. Turner}.} \bibinfo{year}{2000}\natexlab{}.
\newblock \showarticletitle{Local Type Inference}.
\newblock \bibinfo{journal}{\emph{{ACM} TOPLAS}} \bibinfo{volume}{22},
  \bibinfo{number}{1} (\bibinfo{year}{2000}), \bibinfo{pages}{1--44}.
\newblock
\showISSN{0164-0925}
\urldef\tempurl%
\url{https://doi.org/10.1145/345099.345100}
\showDOI{\tempurl}


\bibitem[\protect\citeauthoryear{POPL~2010}{POPL~2010}{2010}]%
        {DBLP:conf/popl/2010}
POPL~2010 \bibinfo{year}{2010}\natexlab{}.
\newblock \bibinfo{booktitle}{\emph{Proc.~37th {ACM} Symp. POPL}}.
  \bibinfo{publisher}{{ACM} Press}, \bibinfo{address}{Madrid, Spain}.
\newblock
\showISBNx{978-1-60558-479-9}


\bibitem[\protect\citeauthoryear{Sangiorgi}{Sangiorgi}{1998}]%
        {DBLP:journals/iandc/Sangiorgi98}
\bibfield{author}{\bibinfo{person}{Davide Sangiorgi}.}
  \bibinfo{year}{1998}\natexlab{}.
\newblock \showarticletitle{An Interpretation of Typed Objects into Typed
  pi-Calculus}.
\newblock \bibinfo{journal}{\emph{Inf. Comput.}} \bibinfo{volume}{143},
  \bibinfo{number}{1} (\bibinfo{year}{1998}), \bibinfo{pages}{34--73}.
\newblock


\bibitem[\protect\citeauthoryear{Scalas and Yoshida}{Scalas and
  Yoshida}{2016}]%
        {DBLP:conf/ecoop/ScalasY16}
\bibfield{author}{\bibinfo{person}{Alceste Scalas} {and}
  \bibinfo{person}{Nobuko Yoshida}.} \bibinfo{year}{2016}\natexlab{}.
\newblock \showarticletitle{Lightweight Session Programming in {Scala}}. In
  \bibinfo{booktitle}{\emph{{ECOOP}}} \emph{(\bibinfo{series}{LIPIcs})},
  Vol.~\bibinfo{volume}{56}. \bibinfo{publisher}{Schloss Dagstuhl -
  Leibniz-Zentrum fuer Informatik}, \bibinfo{pages}{21:1--21:28}.
\newblock


\bibitem[\protect\citeauthoryear{Shi and Xi}{Shi and Xi}{2013}]%
        {DBLP:journals/scp/ShiX13}
\bibfield{author}{\bibinfo{person}{Rui Shi} {and} \bibinfo{person}{Hongwei
  Xi}.} \bibinfo{year}{2013}\natexlab{}.
\newblock \showarticletitle{A Linear Type System for Multicore Programming in
  {ATS}}.
\newblock \bibinfo{journal}{\emph{Science of Computer Programming}}
  \bibinfo{volume}{78}, \bibinfo{number}{8} (\bibinfo{year}{2013}),
  \bibinfo{pages}{1176--1192}.
\newblock
\urldef\tempurl%
\url{https://doi.org/10.1016/j.scico.2012.09.005}
\showDOI{\tempurl}


\bibitem[\protect\citeauthoryear{Sj{\"{o}}berg, Casinghino, Ahn, Collins, III,
  Fu, Kimmell, Sheard, Stump, and Weirich}{Sj{\"{o}}berg et~al\mbox{.}}{2012}]%
        {DBLP:journals/corr/abs-1202-2923}
\bibfield{author}{\bibinfo{person}{Vilhelm Sj{\"{o}}berg},
  \bibinfo{person}{Chris Casinghino}, \bibinfo{person}{Ki~Yung Ahn},
  \bibinfo{person}{Nathan Collins}, \bibinfo{person}{Harley D.~Eades III},
  \bibinfo{person}{Peng Fu}, \bibinfo{person}{Garrin Kimmell},
  \bibinfo{person}{Tim Sheard}, \bibinfo{person}{Aaron Stump}, {and}
  \bibinfo{person}{Stephanie Weirich}.} \bibinfo{year}{2012}\natexlab{}.
\newblock \showarticletitle{Irrelevance, Heterogeneous Equality, and
  Call-by-value Dependent Type Systems}. In
  \bibinfo{booktitle}{\emph{Proceedings Fourth Workshop on Mathematically
  Structured Functional Programming, {MSFP} 2012, Tallinn, Estonia, 25 March
  2012.}} \emph{(\bibinfo{series}{{EPTCS}})},
  \bibfield{editor}{\bibinfo{person}{James Chapman} {and}
  \bibinfo{person}{Paul~Blain Levy}} (Eds.), Vol.~\bibinfo{volume}{76}.
  \bibinfo{pages}{112--162}.
\newblock
\urldef\tempurl%
\url{https://doi.org/10.4204/EPTCS.76.9}
\showDOI{\tempurl}


\bibitem[\protect\citeauthoryear{Swamy, Chen, Fournet, Strub, Bhargavan, and
  Yang}{Swamy et~al\mbox{.}}{2013}]%
        {DBLP:journals/jfp/SwamyCFSBY13}
\bibfield{author}{\bibinfo{person}{Nikhil Swamy}, \bibinfo{person}{Juan Chen},
  \bibinfo{person}{C{\'{e}}dric Fournet}, \bibinfo{person}{Pierre{-}Yves
  Strub}, \bibinfo{person}{Karthikeyan Bhargavan}, {and} \bibinfo{person}{Jean
  Yang}.} \bibinfo{year}{2013}\natexlab{}.
\newblock \showarticletitle{Secure Distributed Programming With Value-Dependent
  Types}.
\newblock \bibinfo{journal}{\emph{J. Funct. Program.}} \bibinfo{volume}{23},
  \bibinfo{number}{4} (\bibinfo{year}{2013}), \bibinfo{pages}{402--451}.
\newblock
\urldef\tempurl%
\url{https://doi.org/10.1017/S0956796813000142}
\showDOI{\tempurl}


\bibitem[\protect\citeauthoryear{Takeuchi, Honda, and Kubo}{Takeuchi
  et~al\mbox{.}}{1994}]%
        {TakeuchiHondaKubo1994}
\bibfield{author}{\bibinfo{person}{Kaku Takeuchi}, \bibinfo{person}{Kohei
  Honda}, {and} \bibinfo{person}{Makoto Kubo}.}
  \bibinfo{year}{1994}\natexlab{}.
\newblock \showarticletitle{An Interaction-Based Language and its Typing
  System}.
\newblock In \bibinfo{booktitle}{\emph{6th International {PARLE} Conference}},
  \bibfield{editor}{\bibinfo{person}{C.~Halatsis},
  \bibinfo{person}{D.~Maritsas}, \bibinfo{person}{G.~Philokyprou}, {and}
  \bibinfo{person}{S.~Theodoridis}} (Eds.). \bibinfo{series}{LNCS},
  Vol.~\bibinfo{volume}{817}. \bibinfo{publisher}{Springer},
  \bibinfo{address}{Athens, Greece}, \bibinfo{pages}{398--413}.
\newblock




\bibitem[\protect\citeauthoryear{Toninho, Caires, and Pfenning}{Toninho
  et~al\mbox{.}}{2011}]%
        {ToninhoCairesPfenning2011}
\bibfield{author}{\bibinfo{person}{Bernardo Toninho}, \bibinfo{person}{Lu\'{i}s
  Caires}, {and} \bibinfo{person}{Frank Pfenning}.}
  \bibinfo{year}{2011}\natexlab{}.
\newblock \showarticletitle{Dependent Session Types via Intuitionistic Linear
  Type Theory}. In \bibinfo{booktitle}{\emph{PPDP}},
  \bibfield{editor}{\bibinfo{person}{Peter Schneider-Kamp} {and}
  \bibinfo{person}{Michael Hanus}} (Eds.). \bibinfo{publisher}{ACM},
  \bibinfo{address}{Odense, Denmark}, \bibinfo{pages}{161--172}.
\newblock
\showISBNx{978-1-4503-0776-5}


\bibitem[\protect\citeauthoryear{Toninho and Yoshida}{Toninho and
  Yoshida}{2018}]%
        {DBLP:conf/fossacs/ToninhoY18}
\bibfield{author}{\bibinfo{person}{Bernardo Toninho} {and}
  \bibinfo{person}{Nobuko Yoshida}.} \bibinfo{year}{2018}\natexlab{}.
\newblock \showarticletitle{Depending on Session-Typed Processes}. In
  \bibinfo{booktitle}{\emph{FoSSaCS}} \emph{(\bibinfo{series}{Lecture Notes in
  Computer Science})}, Vol.~\bibinfo{volume}{10803}.
  \bibinfo{publisher}{Springer}, \bibinfo{pages}{128--145}.
\newblock


\bibitem[\protect\citeauthoryear{Vasconcelos}{Vasconcelos}{2012}]%
        {DBLP:journals/iandc/Vasconcelos12}
\bibfield{author}{\bibinfo{person}{Vasco~T. Vasconcelos}.}
  \bibinfo{year}{2012}\natexlab{}.
\newblock \showarticletitle{Fundamentals of Session Types}.
\newblock \bibinfo{journal}{\emph{Information and Control}}
  \bibinfo{volume}{217} (\bibinfo{year}{2012}), \bibinfo{pages}{52--70}.
\newblock


\bibitem[\protect\citeauthoryear{Vasconcelos, Ravara, and Gay}{Vasconcelos
  et~al\mbox{.}}{2006}]%
        {VasconcelosRavaraGay2006}
\bibfield{author}{\bibinfo{person}{Vasco~T. Vasconcelos},
  \bibinfo{person}{Ant{\'{o}}nio Ravara}, {and} \bibinfo{person}{Simon~J.
  Gay}.} \bibinfo{year}{2006}\natexlab{}.
\newblock \showarticletitle{Type Checking a Multithreaded Functional Language
  with Session Types}.
\newblock \bibinfo{journal}{\emph{Theoretical Computer Science}}
  \bibinfo{volume}{368}, \bibinfo{number}{1-2} (\bibinfo{year}{2006}),
  \bibinfo{pages}{64--87}.
\newblock


\bibitem[\protect\citeauthoryear{Vasconcelos and Tokoro}{Vasconcelos and
  Tokoro}{1993}]%
        {DBLP:conf/isotas/VasconcelosT93}
\bibfield{author}{\bibinfo{person}{Vasco~Thudichum Vasconcelos} {and}
  \bibinfo{person}{Mario Tokoro}.} \bibinfo{year}{1993}\natexlab{}.
\newblock \showarticletitle{A Typing System for a Calculus of Objects}. In
  \bibinfo{booktitle}{\emph{{ISOTAS}}} \emph{(\bibinfo{series}{Lecture Notes in
  Computer Science})}, Vol.~\bibinfo{volume}{742}.
  \bibinfo{publisher}{Springer}, \bibinfo{pages}{460--474}.
\newblock


\bibitem[\protect\citeauthoryear{Wadler}{Wadler}{2012}]%
        {Wadler2012}
\bibfield{author}{\bibinfo{person}{Philip Wadler}.}
  \bibinfo{year}{2012}\natexlab{}.
\newblock \showarticletitle{Propositions as Sessions}. In
  \bibinfo{booktitle}{\emph{ICFP'12}},
  \bibfield{editor}{\bibinfo{person}{Robby~Bruce Findler}} (Ed.).
  \bibinfo{publisher}{ACM}, \bibinfo{address}{Copenhagen, Denmark},
  \bibinfo{pages}{273--286}.
\newblock
\showISBNx{978-1-4503-1054-3}


\bibitem[\protect\citeauthoryear{Walker}{Walker}{2005}]%
        {Walker2005-attapl}
\bibfield{author}{\bibinfo{person}{David Walker}.}
  \bibinfo{year}{2005}\natexlab{}.
\newblock \showarticletitle{Substructural Type Systems}.
\newblock In \bibinfo{booktitle}{\emph{Advanced Topics in Types and Programming
  Languages}}, \bibfield{editor}{\bibinfo{person}{Benjamin~C. Pierce}} (Ed.).
  \bibinfo{publisher}{MIT Press}, Chapter~1.
\newblock


\bibitem[\protect\citeauthoryear{Wu and Xi}{Wu and Xi}{2017}]%
        {DBLP:journals/corr/WuX17}
\bibfield{author}{\bibinfo{person}{Hanwen Wu} {and} \bibinfo{person}{Hongwei
  Xi}.} \bibinfo{year}{2017}\natexlab{}.
\newblock \bibinfo{title}{Dependent Session Types}.
\newblock \bibinfo{howpublished}{http://arxiv.org/abs/1704.07004}.
  (\bibinfo{year}{2017}).
\newblock
\newblock
\shownote{arXiv CoRR.}


\end{thebibliography}
